\documentclass{theoretics}
\pdfoutput=1

\title{Characterizing Omega-Regularity through Finite-Memory Determinacy of Games on Infinite Graphs}

\ThCSauthor[affil1]{Patricia Bouyer}{bouyer@lmf.cnrs.fr}[https://orcid.org/0000-0002-2823-0911]
\ThCSauthor[affil2]{Mickael Randour}{mickael.randour@umons.ac.be}[https://orcid.org/0000-0001-8777-2385]
\ThCSauthor[affil2,affil1]{Pierre Vandenhove}{pierre.vdhove@gmail.com}[https://orcid.org/0000-0001-5834-1068]

\ThCSaffil[affil1]{Université Paris-Saclay, CNRS, ENS Paris-Saclay, Laboratoire Méthodes Formelles, 91190, Gif-sur-Yvette, France}
\ThCSaffil[affil2]{F.R.S.-FNRS \& UMONS -- Université de Mons, Mons, Belgium}
\ThCSaffil[affil3]{}
\ThCSaffil[affil4]{}

\ThCSyear{2023}
\ThCSarticlenum{1}
\ThCSreceived{May 24, 2022}
\ThCSrevised{Aug 9, 2022}
\ThCSaccepted{Nov 22, 2022}
\ThCSpublished{Jan 14, 2023}
\ThCSkeywords{two-player games on graphs, infinite arenas, finite-memory determinacy, optimal strategies, \texorpdfstring{$\omega$}{omega}-regular languages}
\ThCSdoi{10.46298/theoretics.23.1}
\ThCSshortnames{P. Bouyer, M. Randour and P. Vandenhove}
\ThCSshorttitle{Characterizing Omega-Regularity through Finite-Memory Determinacy}
\ThCSthanks{An invited paper from STACS 2022. This work has been partially supported by the ANR Project MAVeriQ (ANR-20-CE25-0012). Mickael Randour is an F.R.S.-FNRS Research Associate and a member of the TRAIL Institute.
Pierre Vandenhove is an F.R.S.-FNRS Research Fellow.}

\usepackage{amsmath}
\usepackage{tikz}
\usetikzlibrary{arrows,automata,calc,decorations.pathmorphing,decorations.pathreplacing,shapes.geometric}
\tikzset{decoration={snake,amplitude=.5mm,segment length=3mm,post length=0.8mm,pre length=0mm}}

\usepackage{color}
\usepackage{xspace}
\usepackage{amsfonts,amssymb}
\usepackage{mathtools}
\usepackage{tikz}

\newcommand{\IN}{\mathbb{N}}
\newcommand{\IZ}{\mathbb{Z}}
\newcommand{\IQ}{\mathbb{Q}}
\newcommand{\IR}{\mathbb{R}}

\renewcommand{\epsilon}{\varepsilon}

\DeclarePairedDelimiter\ceil{\lceil}{\rceil}
\DeclarePairedDelimiter\floor{\lfloor}{\rfloor}

\newcommand{\emptyPth}{\ensuremath{\lambda}}
\newcommand{\hist}{\ensuremath{h}}
\newcommand{\play}{\ensuremath{\rho}}
\newcommand{\lang}{\ensuremath{L}}
\newcommand{\clr}{\ensuremath{c}}
\newcommand{\colors}{\ensuremath{C}}
\newcommand{\s}{\ensuremath{s}}
\newcommand{\states}{\ensuremath{S}}
\newcommand{\edge}{\ensuremath{e}}
\newcommand{\edges}{\ensuremath{E}}
\newcommand{\edgeIn}{\ensuremath{\mathsf{in}}}
\newcommand{\edgeOut}{\ensuremath{\mathsf{out}}}
\newcommand{\arena}{\ensuremath{\mathcal{A}}}
\newcommand{\col}{\ensuremath{\mathsf{col}}}
\newcommand{\colHatFin}{\ensuremath{\mathsf{col}}}
\newcommand{\colHatInf}{\ensuremath{\mathsf{col}}}
\newcommand{\arenaFull}{\ensuremath{(\states, \states_1, \states_2, \edges)}}
\newcommand{\inverse}[1]{#1^{-1}}
\newcommand{\game}{\ensuremath{\mathcal{G}}}
\newcommand{\gameFull}{\ensuremath{(\arena, \wc)}}
\newcommand{\strat}{\ensuremath{\sigma}}
\newcommand{\strats}{\ensuremath{\Sigma}}
\newcommand*{\player}[1]{\ensuremath{\mathcal{P}_{#1}}}
\newcommand{\Pone}{\ensuremath{\player{1}}}
\newcommand{\Ptwo}{\ensuremath{\player{2}}}

\newcommand{\Plays}{\ensuremath{\mathsf{Plays}}}
\newcommand{\Hists}{\ensuremath{\mathsf{Hists}}}

\newcommand{\win}{\ensuremath{\mathsf{win}}}
\newcommand{\lose}{\ensuremath{\mathsf{lose}}}

\newcommand{\memSet}{\ensuremath{\Lambda}}
\newcommand*{\winCyc}[1]{\ensuremath{\Cycles_{#1}^\win}}
\newcommand*{\loseCyc}[1]{\ensuremath{\Cycles_{#1}^\lose}}
\newcommand*{\winCycWord}[1]{\ensuremath{\Cycles_\memSkel^{\win, #1}}}
\newcommand*{\loseCycWord}[1]{\ensuremath{\Cycles_\memSkel^{\lose, #1}}}
\newcommand*{\winCycWordAtmtn}[2]{\ensuremath{\Cycles_{#2}^{\win, #1}}}
\newcommand*{\loseCycWordAtmtn}[2]{\ensuremath{\Cycles_{#2}^{\lose, #1}}}
\newcommand{\cc}[1]{\ensuremath{\col(#1)}}
\newcommand{\emptyPthSymbol}{\ensuremath{\varepsilon}}

\newcommand{\word}{\ensuremath{w}}
\newcommand{\memState}{\ensuremath{m}}
\newcommand{\memStates}{\ensuremath{M}}
\newcommand{\memUpd}{\ensuremath{\alpha_{\mathsf{upd}}}}
\newcommand{\memNxt}{\ensuremath{\alpha_{\mathsf{nxt}}}}
\newcommand{\memUpdHat}{\ensuremath{\memUpd^*}}
\newcommand{\memInit}{\ensuremath{\memState_{\mathsf{init}}}}
\newcommand{\memSkel}{\ensuremath{\mathcal{M}}}
\newcommand{\memSkelFull}{\ensuremath{(\memStates, \memInit, \memUpd)}}
\newcommand{\memCycles}{\ensuremath{\Cycles_\memSkel}}
\newcommand*{\memCyclesOn}[1]{\ensuremath{\Cycles_{#1}}}
\newcommand*{\memPathsOn}[2]{\ensuremath{\Pi_{#1, #2}}}
\newcommand{\memWinCycles}{\ensuremath{\memCycles^\win}}
\newcommand{\memLoseCycles}{\ensuremath{\memCycles^\lose}}
\newcommand{\minStateAtmtn}{\ensuremath{\memSkel_\prefEq}}
\newcommand{\minStateStates}{\ensuremath{\memStates_\prefEq}}
\newcommand{\minStateInit}{\ensuremath{\memInit^\prefEq}}
\newcommand{\minStateUpd}{\ensuremath{\memUpd^\prefEq}}
\newcommand{\minStateAtmtnFull}{\ensuremath{(\minStateStates, \minStateInit, \minStateUpd)}}

\newcommand{\memWordSolo}{\ensuremath{\mathsf{skel}_\memSkel}}
\newcommand*{\memWord}[1]{\ensuremath{\memWordSolo(#1)}}
\newcommand{\parAtmtn}{\ensuremath{(\memSkel, \pri)}}
\newcommand{\parLang}{\ensuremath{\lang_{\parAtmtn}}}

\newcommand{\memProduct}{\ensuremath{\otimes}}

\newcommand*{\Buchi}[1]{\ensuremath{\mathsf{B\text{ü}chi}(#1)}}

\newcommand{\wc}{\ensuremath{W}} % winning condition
\newcommand*{\comp}[1]{\overline{#1}} % complement
\newcommand{\memSkelTriv}{\ensuremath{\memSkel_{\mathsf{triv}}}}

\newcommand*{\card}[1]{\ensuremath{\lvert#1\rvert}}

\newcommand{\Cycles}{\ensuremath{\Gamma}}
\newcommand{\prefOrd}{\ensuremath{\preceq}}
\newcommand{\invPrefOrd}{\ensuremath{\succeq}}
\newcommand{\cycOrd}{\ensuremath{\mathrel{\lhd}}}
\newcommand{\cycOrdInv}{\ensuremath{\mathrel{\rhd}}}
\newcommand{\cyc}{\ensuremath{\gamma}}
\newcommand{\memPth}{\ensuremath{\pi}}
\newcommand{\cycBis}{\ensuremath{\nu}}
\newcommand*{\cycStates}[1]{\ensuremath{\mathsf{st}(#1)}}
\newcommand*{\cycVal}[1]{\ensuremath{\mathsf{val}(#1)}}
\newcommand{\dom}{\ensuremath{\mathsf{domBy}}} % for dominated
\newcommand*{\compar}[1]{\ensuremath{\mathsf{comp(#1)}}} % for competing
\newcommand{\wit}{\ensuremath{\overline{\cyc}}} % for witness

\newcommand{\witBisBar}{\ensuremath{\overline{\cycBis}}}
\newcommand{\witBar}{\ensuremath{\wit}}
\newcommand{\emptyWord}{\ensuremath{\varepsilon}}
\newcommand{\pri}{\ensuremath{p}}
\newcommand{\priCyc}{\ensuremath{p_\Cycles}}
\newcommand{\maxPri}{\ensuremath{p^*}}

\newcommand{\prefEq}{\ensuremath{\sim}}
\newcommand*{\eqClass}[1]{\ensuremath{[#1]_\prefEq}}
\newcommand*{\eqClassCyc}[1]{\ensuremath{[#1]_\cycEq}}
\newcommand{\cycEq}{\ensuremath{\simeq}}
\newcommand{\prefOrdBis}{\ensuremath{\succ}}
\newcommand{\strictPrefOrd}{\ensuremath{\prec}}
\newcommand{\strictInvPrefOrd}{\ensuremath{\succ}}
\newcommand*{\quotient}[2]{{\raisebox{.2em}{$#1$}\!\left/\raisebox{-.2em}{$#2$}\right.}}

\newcommand{\disc}{\ensuremath{\lambda}}
\newcommand*{\DSSolo}[1]{\ensuremath{\mathsf{DS}_{#1}}}
\newcommand*{\DS}[2]{\ensuremath{\DSSolo{#1}(#2)}}
\newcommand{\DSObj}{\ensuremath{\mathsf{DS}_{\disc}^{\ge 0}}}
\newcommand{\gapSolo}{\ensuremath{\mathsf{gap}}}
\newcommand*{\gap}[1]{\ensuremath{\gapSolo(#1)}}
\newcommand{\MPSolo}{\ensuremath{\mathsf{MP}}}
\newcommand*{\MP}[1]{\ensuremath{\MPSolo(#1)}}
\newcommand{\MPObj}{\ensuremath{\mathsf{MP}^{\ge 0}}}
\newcommand{\TPSolo}{\ensuremath{\mathsf{TP}}}
\newcommand*{\TP}[1]{\ensuremath{\TPSolo(#1)}}
\newcommand{\TPObj}{\ensuremath{\mathsf{TP}^{\ge 0}}}
\newcommand{\maxDS}{\ensuremath{\mathsf{maxDS}}}
\newcommand{\minDS}{\ensuremath{\mathsf{minDS}}}

\newcommand{\disjUnion}{\ensuremath{\uplus}}
\newcommand{\bigDisjUnion}{\ensuremath{\biguplus}}
\newcommand{\intervalcc}[1]{\mathopen[#1\mathclose]}
\newcommand{\intervalco}[1]{\mathopen[#1\mathclose)}
\newcommand{\intervaloo}[1]{\mathopen(#1\mathclose)}

\tikzstyle{rond}=[draw,circle,minimum height=7mm]
\tikzstyle{oval}=[draw,ellipse,minimum height=7mm]
\tikzstyle{diamant}=[draw,diamond,minimum height=9mm,minimum width=9mm,aspect=1]
\tikzstyle{carre}=[draw,minimum width=6mm,minimum height=6mm]

\begin{document}

\maketitle

\begin{abstract}
	We consider zero-sum games on infinite graphs, with objectives specified as sets of infinite words over some alphabet of \emph{colors}.
	A well-studied class of objectives is the one of \emph{$\omega$-regular objectives}, due to its relation to many natural problems in theoretical computer science.
	We focus on the strategy complexity question: given an objective, how much memory does each player require to play as well as possible?
	A classical result is that finite-memory strategies suffice for both players when the objective is $\omega$-regular.
	We show a reciprocal of that statement: when both players can play optimally with a \emph{chromatic} finite-memory structure (i.e., whose updates can only observe colors) in all infinite game graphs, then the objective must be $\omega$-regular.
	This provides a game-theoretic characterization of $\omega$-regular objectives, and this characterization can help in obtaining memory bounds.
	Moreover, a by-product of our characterization is a new \emph{one-to-two-player lift}: to show that chromatic finite-memory structures suffice to play optimally in two-player games on infinite graphs, it suffices to show it in the simpler case of one-player games on infinite graphs.
	We illustrate our results with the family of discounted-sum objectives, for which $\omega$-regularity depends on the value of some parameters.
\end{abstract}

\section{Introduction} \label{sec:intro}
\subparagraph*{Games on graphs and synthesis.}
We study \emph{zero-sum turn-based games on infinite graphs}.
In such games, two players, $\Pone$ and $\Ptwo$, interact for an infinite duration on a graph, called an \emph{arena}, whose state space is partitioned into states controlled by $\Pone$ and states controlled by $\Ptwo$.
The game starts in some state of the arena, and the player controlling the current state may choose the next state following an edge of the arena.
Moves of the players in the game are prescribed by their \emph{strategy}, which can use information about the past of the play.
Edges of the arena are labeled with a (possibly infinite) alphabet of \emph{colors}, and the interaction of the players in the arena generates an \emph{infinite word} over this alphabet of colors.
These infinite words can be used to specify the players' objectives: a \emph{winning condition} is a set of infinite words, and $\Pone$ wins a game on a graph if the infinite word generated by its interaction with $\Ptwo$ on the game graph belongs to this winning condition --- otherwise, $\Ptwo$ wins.

This game-theoretic model has applications to the \emph{reactive synthesis} problem~\cite{BCJ18}: a system (modeled as $\Pone$) wants to guarantee some specification (the winning condition) against an uncontrollable environment (modeled as $\Ptwo$).
Finding a \emph{winning strategy} in the game for $\Pone$ corresponds to building a controller for the system that achieves the specification against all possible behaviors of the environment.

\subparagraph*{Strategy complexity.}
We are interested in the \emph{strategy complexity} question: given a winning condition, how \emph{complex} must winning strategies be, and how \emph{simple} can they be?
We are interested in establishing the sufficient and necessary amount of memory to play \emph{optimally}.
We consider in this work that an \emph{optimal strategy} in an arena must be winning from any state from which winning is possible (a property sometimes called \emph{uniformity} in the literature).
The amount of memory relates to how much information about the past is needed to play in an optimal way.
With regard to reactive synthesis, this has an impact in practice on the resources required for an optimal controller.

Three classes of strategies are often distinguished, depending on the number of states of memory they use: memoryless, finite-memory, and infinite-memory strategies.
A notable subclass of finite-memory strategies is the class of strategies that can be implemented with finite-memory structures that only observe the sequences of colors (and not the sequences of states nor edges).
Such memory structures are called \emph{chromatic}~\cite{KopThesis}.
By contrast, finite-memory structures that have access to the states and edges of arenas are called \emph{general}.
Chromatic memory structures are syntactically less powerful and may require more states than general ones~\cite{Cas22}, but have the benefit that they can be defined independently of arenas.

We seek to characterize the winning conditions for which chromatic-finite-memory strategies suffice to play optimally against arbitrarily complex strategies, for both players, in all finite and infinite arenas.
We call this property \emph{chromatic-finite-memory determinacy}.
This property generalizes \emph{memoryless determinacy}, which describes winning conditions for which memoryless strategies suffice to play optimally for both players in all arenas.
Our work follows a line of research~\cite{BLORV20,BORV21} giving various characterizations of chromatic-finite-memory determinacy for games on \emph{finite} arenas (see Remark~\ref{rem:AIFMvsFM} for more details).

\subparagraph*{\texorpdfstring{$\omega$}{Omega}-regular languages.}
A class of winning conditions commonly arising as natural specifications for reactive systems (it encompasses, e.g., linear temporal logic specifications~\cite{Pnu77}) consists of the \emph{$\omega$-regular languages}.
They are, among other characterizations, the languages of infinite words that can be described by a \emph{finite parity automaton}~\cite{Mos84}.
It is known that all $\omega$-regular languages are chromatic-finite-memory determined, which is due to the facts that an $\omega$-regular language is expressible with a parity automaton, and that \emph{parity conditions} admit memoryless optimal strategies~\cite{Kla94,Zie98}.
Multiple works study the strategy complexity of $\omega$-regular languages, giving, e.g., precise general memory requirements for all Muller conditions~\cite{DJW97} or a characterization of the chromatic memory requirements of Muller conditions~\cite[Theorem~28]{Cas22}.

A result in the other direction is given by Colcombet and Niwi\'nski~\cite{CN06}: they showed that if a \emph{prefix-independent} winning condition is memoryless-determined in infinite arenas, then this winning condition must be a parity condition.
As parity conditions are memoryless-determined, this provides an elegant characterization of parity conditions from a strategic perspective, under prefix-independence assumption.

\subparagraph*{Congruence.}
A well-known tool to study a language $\lang$ of finite (resp.\ infinite) words is its \emph{right congruence relation $\prefEq_\lang$}: for two finite words $\word_1$ and $\word_2$, we write $\word_1 \prefEq_\lang \word_2$ if for all finite (resp.\ infinite) words $\word$, $\word_1\word \in \lang$ if and only if $\word_2\word\in\lang$.
There is a natural deterministic (potentially infinite) automaton recognizing the equivalence classes of the right congruence, called the \emph{minimal-state automaton of $\prefEq_\lang$}~\cite{Sta83,MS97}.

The relation between a regular language of \emph{finite} words and its right congruence is given by the Myhill-Nerode theorem~\cite{Ner58}, which provides a natural bijection between the states of the minimal deterministic automaton recognizing a regular language and the equivalence classes of its right congruence relation.
Consequences of this theorem are that a language is regular if and only if its right congruence has finitely many equivalence classes, and a regular language can be recognized by the minimal-state automaton of its right congruence.

For the theory of languages of \emph{infinite} words, the situation is not so simple: $\omega$-regular languages have a right congruence with finitely many equivalence classes, but having finitely many equivalence classes does not guarantee $\omega$-regularity (for example, a language is \emph{prefix-independent} if and only if its right congruence has exactly one equivalence class, but this does not imply $\omega$-regularity).
Moreover, $\omega$-regular languages cannot necessarily be recognized by adding a natural acceptance condition (parity, Rabin, Muller\ldots) to the minimal-state automaton of their right congruence~\cite{AF21}.
There has been multiple works about the links between a language of infinite words and the minimal-state automaton of its right congruence; one relevant question is to understand when a language can be recognized by this minimal-state automaton~\cite{Sta83,MS97,AF21}.

\subparagraph*{Contributions.}
We characterize the $\omega$-regularity of a language of infinite words $\wc$ through the strategy complexity of the zero-sum turn-based games on infinite graphs with winning condition $\wc$: the $\omega$-regular languages are \emph{exactly} the chromatic-finite-memory determined languages (seen as winning conditions) (Theorem~\ref{thm:regular}).
As discussed earlier, it is well-known that $\omega$-regular languages admit chromatic-finite-memory optimal strategies~\cite{Mos84,Zie98,Cas22} --- our results yield the other implication.
This therefore provides a characterization of $\omega$-regular languages through a game-theoretic and strategic lens.

Our technical arguments consist in providing a precise connection between the representation of $\wc$ as a parity automaton and a chromatic memory structure sufficient to play optimally.
If strategies based on a chromatic finite-memory structure are sufficient to play optimally for both players, then $\wc$ is recognized by a parity automaton built on top of the direct product of the \emph{minimal-state automaton of the right congruence} and this \emph{chromatic memory structure} (Theorem~\ref{thm}).
This result generalizes the work from Colcombet and Niwi\'nski~\cite{CN06} in two ways: by relaxing the prefix-independence assumption about the winning condition, and by generalizing the class of strategies considered from memoryless to chromatic-finite-memory strategies.
We recover their result as a special case.

Moreover, we actually show that chromatic-finite-memory determinacy in \emph{one-player} games of both players is sufficient to show $\omega$-regularity of a language.
As $\omega$-regular languages are chromatic-finite-memory determined in two-player games, we can reduce the problem of chromatic-finite-memory determinacy of a winning condition in \emph{two-player} games to the easier chromatic-finite-memory determinacy in \emph{one-player} games (Theorem~\ref{thm:lift}).
Such a \emph{one-to-two-player lift} holds in multiple classes of zero-sum games, such as deterministic games on finite arenas~\cite{GZ05,BLORV20,Koz22} and stochastic games on finite arenas~\cite{GZ09,BORV21}.
The proofs for finite arenas all rely on an \emph{edge-induction technique} (also used in other works about strategy complexity in finite arenas~\cite{Kop06,GK14,CD16}) that appears unfit to deal with infinite arenas.
Although not mentioned by Colcombet and Niwi\'nski, it was already noticed~\cite{KopThesis} that for prefix-independent winning conditions in games on infinite graphs, a one-to-two-player lift for \emph{memoryless} determinacy follows from~\cite{CN06}.

\subparagraph*{Related works.}
We have already mentioned~\cite{DJW97,Zie98,CN06,Kop07,Cas22} for fundamental results on the memory requirements of $\omega$-regular conditions, ~\cite{GZ05,GZ09,BLORV20,BORV21} for characterizations of ``low'' memory requirements in finite (deterministic and stochastic) arenas, and~\cite{Sta83,MS97,AF21} for links between an $\omega$-regular language and the minimal-state automaton of its right congruence.

One stance of our work is that our assumptions about strategy complexity affect \emph{both} players.
Another intriguing question is to understand when the memory requirements of only \emph{one} player are finite.
In finite arenas, a few results in this direction are sufficient conditions for the existence of memoryless optimal strategies for one player~\cite{Kop06,BFMM11}, and a proof by Kopczy\'nski that the chromatic memory requirements of prefix-independent $\omega$-regular conditions are computable~\cite{Kop07,KopThesis}.

Other articles study the strategy complexity of (non-necessarily $\omega$-regular) winning conditions in infinite arenas; see, e.g.,~\cite{Gim04,GW06,CFH14}.
In such non-$\omega$-regular examples, as can be expected given our main result, at least one player needs infinite memory to play optimally, or the arena model is different from ours (e.g., only allowing finite branching --- we discuss such differences in more depth after Theorem~\ref{thm}).
A particularly interesting example w.r.t.\ our results is considered by Chatterjee and Fijalkow~\cite{CF13}.
They study the strategy complexity of \emph{finitary B\"uchi and parity conditions}, and show that $\Pone$ has memoryless optimal strategies for finitary B\"uchi and chromatic-finite-memory optimal strategies for finitary parity.%
\footnote{We argue in Appendix~\ref{app:CF13} that their result also applies to our slightly different setting.}
However, for these (non-$\omega$-regular) winning conditions, $\Ptwo$ needs infinite memory.
This example illustrates that our main result would not hold if we just focused on the strategy complexity of one player.

We mention other works on finite-memory determinacy in different contexts: finite arenas~\cite{LPR18}, non-zero-sum games~\cite{LP18}, countable one-player stochastic games~\cite{KMSTW20}, concurrent games~\cite{LeR20,BLT22}.

This paper extends and complements a preceding conference version~\cite{BRV22} with additional details and complete proofs of all the statements.

\subparagraph*{Structure.}
We fix definitions and notations in Section~\ref{sec:preliminaries}.
Our main results are discussed in Section~\ref{sec:thm}, and their proofs lie in Sections~\ref{sec:proof1} and~\ref{sec:proof2}.
We provide applications of our results to discounted-sum, mean-payoff, and total-payoff winning conditions in Section~\ref{sec:application}.

\section{Preliminaries} \label{sec:preliminaries}
Let $\colors$ be an arbitrary non-empty set of \emph{colors}.
Given a set $A$, we write $A^*$ for the set of finite sequences of elements of $A$ and $A^\omega$ for the set of infinite sequences of elements of $A$.

\subparagraph*{Arenas.}
We consider two players $\Pone$ and $\Ptwo$.
An arena is a tuple $\arena = \arenaFull$ such that $\states = \states_1 \disjUnion \states_2$ (disjoint union) is a non-empty set of \emph{states} and $\edges\subseteq\states\times\colors\times\states$ is a set of \emph{edges}.
The sets of colors, of states, and of edges may be infinite (of arbitrary cardinality), and we allow arenas with infinite branching.
States in $\states_1$ are controlled by $\Pone$ and states in $\states_2$ are controlled by $\Ptwo$.
Given $\edge\in\edges$, we denote by $\edgeIn$, $\col$, and $\edgeOut$ the projections to its first, second, and third component, respectively (i.e., $\edge = (\edgeIn(\edge), \col(\edge), \edgeOut(\edge))$).
We assume arenas to be \emph{non-blocking}: for all $\s\in\states$, there exists $\edge\in\edges$ such that $\edgeIn(\edge) = \s$.

Let $\arena = \arenaFull$ be an arena with $\s\in\states$.
We denote by $\Plays(\arena, \s)$ the set of \emph{plays of~$\arena$ from~$\s$}, that is, infinite sequences of edges $\play = \edge_1\edge_2\ldots \in \edges^\omega$ such that $\edgeIn(\edge_1) = \s$ and for all $i\ge 1$, $\edgeOut(\edge_i) = \edgeIn(\edge_{i+1})$.
For $\play\in\Plays(\arena, \s)$, we denote by $\colHatInf(\play)$ the infinite sequence of colors obtained from applying the $\col$ function to each edge in $\play$.
We denote by $\Hists(\arena, \s)$ the set of \emph{histories of $\arena$ from $\s$}, which are all finite prefixes of plays of $\arena$ from $\s$.
We write $\Plays(\arena)$ and $\Hists(\arena)$ for the sets of all plays of $\arena$ and all histories of $\arena$ (from any state), respectively.
If $\hist = \edge_1\ldots\edge_k$ is a history of $\arena$, we define $\edgeIn(\hist) = \edgeIn(\edge_1)$ and $\edgeOut(\hist) = \edgeOut(\edge_k)$.
For convenience, for every $\s\in\states$, we also consider the \emph{empty history $\emptyPth_\s$} from $\s$, and we set $\edgeIn(\emptyPth_\s) = \edgeOut(\emptyPth_\s) = \s$.
For $i\in\{1, 2\}$, we denote by $\Hists_i(\arena)$ the set of histories $\hist$ such that~$\edgeOut(\hist)\in\states_i$.

An arena $\arena = \arenaFull$ is a \emph{one-player arena of $\Pone$} (resp.\ \emph{$\Ptwo$}) if $\states_2 = \emptyset$ (resp.\ $\states_1 = \emptyset$).

\subparagraph*{Skeletons.}
A \emph{skeleton} is a tuple $\memSkel = \memSkelFull$ such that $\memStates$ is a finite set of \emph{states}, $\memInit\in\memStates$ is an \emph{initial state}, and $\memUpd\colon\memStates\times\colors\to\memStates$ is an \emph{update function}.
We denote by $\memUpdHat$ the natural extension of $\memUpd$ to finite sequences of colors.
We always assume that all states of skeletons are reachable from their initial state.
We define the trivial skeleton $\memSkelTriv$ as the only skeleton with a single state.
Notice that although we require skeletons to have finitely many states, we allow them to have infinitely many transitions (which happens when $\colors$ is infinite).

For $\word = \clr_1\clr_2\ldots\in\colors^\omega$, we define $\memWord{\word}$ as the infinite sequence $(\memState_1, \clr_1)(\memState_2, \clr_2)\ldots \in (\memStates\times\colors)^\omega$ that $\word$ induces in the skeleton ($\memState_1 = \memInit$ and for all $i\ge 1$, $\memUpd(\memState_i, \clr_i) = \memState_{i+1}$).

Let $\memSkel_1 = (\memStates_1, \memInit^1, \memUpd^1)$ and $\memSkel_2 = (\memStates_2, \memInit^2, \memUpd^2)$ be two skeletons.
Their \textit{(direct) product} $\memSkel_1 \memProduct \memSkel_2$ is the skeleton $\memSkelFull$ where $\memStates = \memStates_1 \times \memStates_2$, $\memInit = (\memInit^1, \memInit^2)$, and, for all $\memState_1 \in \memStates_1$, $\memState_2 \in \memStates_2$, $\clr \in \colors$, $\memUpd((\memState_1, \memState_2), \clr) = (\memUpd^1(\memState_1, \clr), \memUpd^2(\memState_2, \clr))$.

\subparagraph*{Strategies.}
Let $\arena = \arenaFull$ be an arena and $i\in\{1, 2\}$.
A \emph{strategy of $\player{i}$ on $\arena$} is a function $\strat_i\colon \Hists_i(\arena) \to \edges$ such that for all $\hist\in\Hists_i(\arena)$, $\edgeOut(\hist) = \edgeIn(\strat_i(\hist))$.
We denote by $\strats_i(\arena)$ the set of strategies of $\player{i}$ on $\arena$.
Given a strategy $\strat_i$ of $\player{i}$, we say that a play $\play = \edge_1\edge_2\ldots$ is \emph{consistent with $\strat_i$} if for all finite prefixes $\hist = \edge_1\ldots\edge_j$ of $\play$ such that $\edgeOut(\hist) \in \states_i$, $\strat_i(\hist) = \edge_{j+1}$.
For $\s\in\states$, we denote by $\Plays(\arena, \s, \strat_i)$ the set of plays from $\s$ that are consistent with $\strat_i$.

For $\memSkel = \memSkelFull$ a skeleton, a strategy $\strat_i\in\strats_i(\arena)$ is \emph{based on (memory) $\memSkel$} if there exists a function $\memNxt\colon \states \times \memStates \to \edges$ such that for all $\s\in\states_i$, $\strat_i(\emptyPth_\s) = \memNxt(\s, \memInit)$, and for all non-empty paths $\hist\in\Hists_i(\arena)$, $\strat_i(\hist) = \memNxt(\edgeOut(\hist), \memUpdHat(\memInit, \colHatFin(\hist)))$.
A strategy is \emph{memoryless} if it is based on $\memSkelTriv$.

\begin{remark}
	Our memory model is \emph{chromatic}~\cite{KopThesis}, i.e., it observes the sequences of colors and not the sequences of edges of arenas, due to the fact that the argument of the update function of a skeleton is in $\memStates\times\colors$.
	It was recently shown that the amount of memory states required to play optimally for a winning condition using chromatic skeletons may be strictly larger than using \emph{general} memory structures (i.e., using memory structures observing edges)~\cite[Proposition~32]{Cas22}, disproving a conjecture by Kopczy\'nski~\cite{KopThesis}.
	The example provided is a Muller condition (hence an $\omega$-regular condition), in which both kinds of memory requirements are still finite.
	A result in this direction is also provided by Le Roux~\cite[Corollary~1]{LeR20} for games on \emph{finite} arenas: it shows that in many games, a strategy using general finite memory can be swapped for a (larger) chromatic finite memory.

	For games on infinite arenas, which we consider in this article, we do not know whether there exists a winning condition with \emph{finite} general memory requirements, but \emph{infinite} chromatic memory requirements.
	Our results focus on chromatic memory requirements.
	% \lipicsEnd
\end{remark}

\subparagraph*{Winning conditions.}
A \emph{(winning) condition} is a set $\wc\subseteq \colors^\omega$.
When a winning condition $\wc$ is clear in the context, we say that an infinite word $\word\in\colors^\omega$ is winning if $\word\in\wc$, and losing if $\word\not\in\wc$.
For a winning condition $\wc$ and a word $\word\in\colors^*$, we write $\inverse{\word}\wc = \{\word'\in\colors^\omega\mid \word\word'\in\wc\}$ for the set of \emph{winning continuations of $\word$}.
We write $\comp{\wc}$ for the complement $\colors^\omega\setminus\wc$ of a winning condition $\wc$.

A \emph{game} is a tuple $\game = \gameFull$ where $\arena$ is an arena and $\wc$ is a winning condition.

\subparagraph*{Optimality and determinacy.}
Let $\game = (\arena = \arenaFull, \wc)$ be a game, and $\s\in\states$.
We say that \emph{$\strat_1\in\strats_1(\arena)$ is winning from $\s$} if $\colHatInf(\Plays(\arena, \s, \strat_1)) \subseteq \wc$, and we say that \emph{$\strat_2\in\strats_2(\arena)$ is winning from $\s$} if $\colHatInf(\Plays(\arena, \s, \strat_2)) \subseteq \comp{\wc}$.

A strategy of $\player{i}$ is \emph{optimal for $\player{i}$ in $(\arena, \wc)$} if it is winning from all the states from which~$\player{i}$ has a winning strategy.
We often write \emph{optimal for $\player{i}$ in $\arena$} if the winning condition $\wc$ is clear from the context.
We stress that this notion of optimality requires a \emph{single} strategy to be winning from \emph{all} the winning states (a property sometimes called \emph{uniformity}).

A winning condition $\wc$ is \emph{determined} if for all games $\game = (\arena = \arenaFull, \wc)$, for all $\s\in\states$, either $\Pone$ or $\Ptwo$ has a winning strategy from $\s$.
Let $\memSkel$ be a skeleton.
We say that a winning condition $\wc$ is \emph{$\memSkel$-determined} if $(i)$ $\wc$ is determined and $(ii)$ in all arenas $\arena$, both players have an optimal strategy based on $\memSkel$.
A winning condition $\wc$ is \emph{one-player $\memSkel$-determined} if in all one-player arenas $\arena$ of $\Pone$, $\Pone$ has an optimal strategy based on $\memSkel$ \emph{and} in all one-player arenas $\arena$ of $\Ptwo$, $\Ptwo$ has an optimal strategy based on $\memSkel$.
A winning condition $\wc$ is (one-player) \emph{memoryless-determined} if it is (one-player) $\memSkelTriv$-determined.
A winning condition~$\wc$ is \emph{(one-player) chromatic-finite-memory determined} if there exists a skeleton $\memSkel$ such that it is (one-player) $\memSkel$-determined.

\begin{remark}
	In the (one-player) $\memSkel$-determinacy definition, we enforce that $\memSkel$ suffices to play optimally for \emph{both} players (in their respective one-player arenas).
	The memory requirements of players differ in general: a typical example is the case of \emph{Rabin conditions}, for which $\Pone$ has memoryless optimal strategies, but for which $\Ptwo$ (who tries to achieve a \emph{Streett condition}) may need memory~\cite{Zie98}.
	In practice, if in all arenas, $\Pone$ has optimal strategies based on a skeleton~$\memSkel_1$ and $\Ptwo$ has optimal strategies based on a skeleton $\memSkel_2$, then both players have optimal strategies based on $\memSkel_1 \memProduct \memSkel_2$ (they can play the same strategies, simply not taking into account the information given to them by the other skeleton); hence, $\wc$ is $(\memSkel_1 \memProduct \memSkel_2)$-determined.
	% \lipicsEnd
\end{remark}

\begin{remark} \label{rem:AIFMvsFM}
	It might seem surprising that for chromatic-finite-memory determinacy, we require the existence of a \emph{single} skeleton that suffices to play optimally in \emph{all} arenas, rather than the seemingly weaker existence, for each arena, of a finite skeleton (which may depend on the arena) that suffices to play optimally.
	In infinite arenas, it turns out that these notions are equivalent.
	\begin{lemma} \label{lem:AIFMisFM}
		Let $\wc\subseteq\colors^\omega$ be a winning condition.
		The following are equivalent:
		\begin{enumerate}
			\item for all arenas $\arena$, there exists a skeleton $\memSkel^\arena$ such that both players have an optimal strategy based on $\memSkel^\arena$ in $\arena$;\label{enum:1}
			\item $\wc$ is chromatic-finite-memory determined.\label{enum:2}
		\end{enumerate}
	\end{lemma}
	\begin{proof}
		It is clear that \ref{enum:2}.$\implies$\ref{enum:1}., as~\ref{enum:2}.\ means that there is a (fixed) skeleton that suffices in each arena.
		We now show \ref{enum:1}.$\implies$\ref{enum:2}.
		We proceed by contraposition.
		Assume that~\ref{enum:2}.\ does not hold, i.e., that for all skeletons $\memSkel$, there exists $\arena^\memSkel = (\states^\memSkel, \states_1^\memSkel, \states_2^\memSkel, \edges^\memSkel)$ such that at least one player does not have an optimal strategy based on $\memSkel$ in $\arena^\memSkel$.
		We consider the arena $\arena = (\bigDisjUnion_\memSkel \states^\memSkel, \bigDisjUnion_\memSkel \states_1^\memSkel, \bigDisjUnion_\memSkel \states_2^\memSkel, \bigDisjUnion_\memSkel \edges^\memSkel)$ consisting in the ``disjoint union'' over all skeletons $\memSkel$ of the arenas $\arena^\memSkel$.
		Clearly, no strategy based on a skeleton suffices to play optimally in $\arena$; this shows that~\ref{enum:1}.\ does not hold.
	\end{proof}
 \begin{sloppypar}
	When restricted to finite arenas, we do not have an equivalence between these two notions (hence the distinction between finite-memory determinacy and \emph{arena-independent} finite-memory determinacy~\cite{BLORV20,BORV21}).
	Our proof of Lemma~\ref{lem:AIFMisFM} exploits that an infinite union of arenas is still an arena, which is not true when restricted to finite arenas.
 \end{sloppypar}
	% \lipicsEnd
\end{remark}

\subparagraph*{\texorpdfstring{$\omega$}{Omega}-regular languages.}
We define a \emph{parity automaton} as a pair $\parAtmtn$ where $\memSkel$ is a skeleton and $\pri\colon \memStates\times\colors \to \{0, \ldots, n\}$; function $\pri$ assigns \emph{priorities} to every transition of $\memSkel$.
This definition implies that we consider deterministic and complete parity automata (i.e., in every state, reading a color leads to exactly one state).
Following~\cite{CCF21}, if $\memSkel$ is a skeleton, we say that a parity automaton $(\memSkel', \pri)$ is \emph{defined on top of $\memSkel$} if $\memSkel' = \memSkel$.

A parity automaton $\parAtmtn$ defines a language $\parLang$ of all the infinite words $\word\in\colors^\omega$ such that, for $\memWord{\word} = (\memState_1, \clr_1)(\memState_2, \clr_2)\ldots$, $\limsup_{i\ge 1} \pri(\memState_i, \clr_i)$ is even.
We say that $\wc\subseteq \colors^\omega$ is \emph{recognized by $\parAtmtn$} if $\wc = \parLang$.
We emphasize that we consider here \emph{transition-based} parity acceptance conditions: we assign priorities to transitions, and not to states of $\memSkel$.
For further information on links between state-based and transition-based acceptance conditions, we refer to~\cite{Cas22}.
If a language of infinite words can be recognized by a parity automaton, it is called \emph{$\omega$-regular}.

\begin{remark} \label{rem:finColors}
	Formally, a deterministic parity automaton should be defined on a finite set of colors, whereas here our set of colors $\colors$ can have any cardinality.
	However, given a parity automaton $(\memSkel = \memSkelFull, \pri)$, as there are finitely many states in $\memSkel$ and finitely many priorities, there are in practice only finitely many ``truly different'' classes of colors: two colors $\clr_1, \clr_2\in\colors$ can be assumed to be equal (w.r.t.\ $\parAtmtn$) if $\memUpd(\cdot, \clr_1) = \memUpd(\cdot, \clr_2)$ and $\pri(\cdot, \clr_1) = \pri(\cdot, \clr_2)$.
	% \lipicsEnd
\end{remark}

\subparagraph*{Right congruence.}
For $\sim$ an equivalence relation, we call the \emph{index of $\sim$} the number of equivalence classes of $\sim$.
We denote by $\eqClass{a}$ the equivalence class of an element $a$ for $\prefEq$.

Let $\wc$ be a winning condition.
We define the \emph{right congruence ${\prefEq_\wc} \subseteq \colors^*\times\colors^*$ of $\wc$} as $\word_1\prefEq_\wc\word_2$ if $\inverse{\word_1}\wc = \inverse{\word_2}\wc$ (meaning that $\word_1$ and $\word_2$ have the same winning continuations).
Relation $\prefEq_\wc$ is an equivalence relation.
When $\wc$ is clear from the context, we write $\prefEq$ for $\prefEq_\wc$.
We denote by $\emptyWord$ the empty word.
When $\prefEq$ has finite index, we can associate a natural skeleton $\minStateAtmtn = \minStateAtmtnFull$ to $\prefEq$ such that $\minStateStates$ is the set of equivalence classes of $\prefEq$, $\minStateInit = \eqClass{\emptyWord}$, and $\minStateUpd(\eqClass{\word}, \clr) = \eqClass{\word\clr}$.
This transition function is well-defined since it follows from the definition of $\prefEq$ that if $\word_1 \prefEq \word_2$, then for all $\clr\in\colors$, $\word_1\clr \prefEq \word_2\clr$.
Hence, the choice of representatives for the equivalence classes does not have an impact in this definition.
We call skeleton $\minStateAtmtn$ the \emph{minimal-state automaton of $\prefEq$}~\cite{Sta83,MS97}.

\section{Concepts and characterization} \label{sec:thm}
We define two concepts at the core of our characterization, one of them dealing with \emph{prefixes} and the other one dealing with \emph{cycles}.
Let $\wc\subseteq\colors^\omega$ be a winning condition and $\memSkel = \memSkelFull$ be a skeleton.
We first introduce some notations to refer to sequences of transitions of skeletons.

%\subparagraph*{Manipulating skeletons.}
We say that a non-empty sequence $\memPth = (\memState_1, \clr_1)\ldots(\memState_k,\clr_k)\in (\memStates\times\colors)^+$ is a \emph{path of $\memSkel$ (from $\memState_1$ to $\memUpd(\memState_k, \clr_k)$)} if for all $i\in\{1,\ldots,k-1\}$, $\memUpd(\memState_i, \clr_i) = \memState_{i+1}$.
For convenience, we also consider every element $(\memState, \emptyPthSymbol)$ for $\memState\in\memStates$ and $\emptyPthSymbol\notin\colors$ to be an \emph{empty path of $\memSkel$ (from $\memState$ to $\memState$)}.
A non-empty path of $\memSkel$ from $\memState$ to $\memState'$ is a \emph{cycle of $\memSkel$ (on $\memState$)} if $\memState = \memState'$.
Cycles of $\memSkel$ are usually denoted by letter $\cyc$.
For $\memPth = (\memState_1, \clr_1)\ldots(\memState_k,\clr_k)$ a path of $\memSkel$, we define $\cycStates{\memPth}$ to be the set $\{\memState_1,\ldots,\memState_k\}$, and $\colHatFin(\memPth)$ to be the sequence $\clr_1\ldots\clr_k\in\colors^*$.
For an infinite sequence $(\memState_1, \clr_1)(\memState_2, \clr_2)\ldots\in(\memStates\times\colors)^\omega$, we also write $\colHatInf((\memState_1, \clr_1)(\memState_2, \clr_2)\ldots)$ for the infinite sequence $\clr_1\clr_2\ldots\in\colors^\omega$.
If $(\memState, \clr)\in\memStates\times\colors$ occurs in a path $\memPth$ of $\memSkel$, we call $(\memState, \clr)$ a \emph{transition of $\memPth$} and we write $(\memState, \clr)\in\memPth$.

For $\memState,\memState'\in\memStates$, we write $\memPathsOn{\memState}{\memState'}$ for the set of paths of $\memSkel$ from $\memState$ to $\memState'$, $\memCyclesOn{\memState}$ for the set of cycles of $\memSkel$ on $\memState$, and $\memCycles$ for the set of all cycles of $\memSkel$ (on any skeleton state).
We extend notation $\col$ to sets of paths or cycles (e.g., $\cc{\memCycles} = \{\colHatFin(\cyc) \in \colors^+ \mid \cyc \in \memCycles\}$).

\subparagraph*{Prefix-independence.}
Let $\prefEq$ be the right congruence of $\wc$.

\begin{definition}
	Condition $\wc$ is \emph{$\memSkel$-prefix-independent} if for all $\memState\in\memStates$, for all $\word_1, \word_2\in\cc{\memPathsOn{\memInit}{\memState}}$, $\word_1 \prefEq \word_2$.
\end{definition}
In other words, $\wc$ is $\memSkel$-prefix-independent if finite words reaching the same state of $\memSkel$ from its initial state have the same winning continuations.
The classical notion of \emph{prefix-independence} is equivalent to $\memSkelTriv$-prefix-independence (as all finite words have the exact same set of winning continuations, which is $\wc$).
If $\prefEq$ has finite index, $\wc$ is in particular $\minStateAtmtn$-prefix-independent: indeed, two finite words reach the same state of $\minStateAtmtn$ (if and) only if they are equivalent for $\prefEq$.
Any skeleton $\memSkel$ such that $\wc$ is $\memSkel$-prefix-independent must have at least one state for each equivalence class of $\prefEq$, but may have multiple states tied to the same equivalence class.

\subparagraph*{Cycle-consistency.}
For $\word\in\colors^*$, we define
\[
	\winCycWord{\word} = \{\cyc \in \memCyclesOn{\memState} \mid \memState = \memUpdHat(\memInit, \word)\ \text{and}\ (\colHatFin(\cyc))^\omega \in \inverse{\word}\wc\}
\]
as the cycles on the skeleton state reached by $\word$ in $\memSkel$ that induce winning words when repeated infinitely many times after $\word$.
We define
\[
\loseCycWord{\word} = \{\cyc \in \memCyclesOn{\memState} \mid \memState = \memUpdHat(\memInit, \word)\ \text{and}\ (\colHatFin(\cyc))^\omega \in \inverse{\word}\comp{\wc}\}
\]
as their losing counterparts.
We emphasize that cycles are allowed to go through the same edge multiple times.

\begin{definition}
	Condition $\wc$ is \emph{$\memSkel$-cycle-consistent} if for all $\word\in\colors^*$, $(\cc{\winCycWord{\word}})^\omega \subseteq \inverse{\word}\wc$ and $(\cc{\loseCycWord{\word}})^\omega \subseteq \inverse{\word}\comp{\wc}$.
\end{definition}
What this says is that after any finite word, if we concatenate infinitely many winning (resp.\ losing) cycles on the skeleton state reached by that word, then it only produces winning (resp.\ losing) infinite words.

\begin{example} \label{ex:buchiABuchiB}
	For $x\in\colors$, let $\Buchi{x}$ be the set of infinite words on $\colors$ that see color $x$ infinitely often.
	Let $\colors = \{a, b, c\}$.
	Condition $\wc = \Buchi{a}\cap\Buchi{b}$ is $\memSkelTriv$-prefix-independent, but not $\memSkelTriv$-cycle-consistent: for any $\word\in\colors^*$, $a$ and $b$ are both in $\cc{\loseCycWordAtmtn{\word}{\memSkelTriv}}$ (as $\word a^\omega$ and $\word b^\omega$ are losing), but word $\word(ab)^\omega$ is winning.
	However, $\wc$ is $\memSkel$-cycle-consistent for the skeleton~$\memSkel$ with two states $\memInit$ and $\memState_2$ represented in Figure~\ref{fig:buchiABuchiBSkel}.
	For finite words reaching $\memInit$, the losing cycles only see $a$ and $c$, and combining infinitely many of them gives an infinite word without~$b$, which is a losing continuation of any finite word.
	The winning cycles are the ones that go to~$\memState_2$ and then go back to $\memInit$, as they must see both $a$ and $b$; combining infinitely many of them guarantees a winning continuation after any finite word.
	A similar reasoning applies to state $\memState_2$.
	Notice that $\wc$ is also $\memSkel$-prefix-independent.
	With regard to memory requirements, condition~$\wc$ is not $\memSkelTriv$-determined but is $\memSkel$-determined.%
	\begin{figure}[t]
		\centering
		\begin{tikzpicture}[every node/.style={font=\small,inner sep=1pt}]
			\draw (0,0) node[diamant] (s1) {$\memInit$};
			\draw ($(s1)+(2,0)$) node[diamant] (s2) {$\memState_2$};
			\draw (s1) edge[-latex',out=30,in=150] node[above=2pt] {$b$} (s2);
			\draw ($(s1)+(0,1.1)$) edge[-latex'] (s1);
			\draw (s1) edge[-latex',out=150,in=210,distance=0.8cm] node[left=2pt] {$a, c$} (s1);
			\draw (s2) edge[-latex',out=-150,in=-30] node[below=2pt] {$a$} (s1);
			\draw (s2) edge[-latex',out=30,in=-30,distance=0.8cm] node[right=2pt] {$b, c$} (s2);
		\end{tikzpicture}
		\caption{Skeleton $\memSkel$ such that $\wc = \Buchi{a}\cap\Buchi{b}$ is $\memSkel$-cycle-consistent (Example~\ref{ex:buchiABuchiB}).
		In figures, we use rhombuses (resp.\ circles, squares) to depict skeleton states (resp.\ arena states controlled by $\Pone$, arena states controlled by $\Ptwo$).
  }
		\label{fig:buchiABuchiBSkel}
	\end{figure}
	% \lipicsEnd
\end{example}

\subparagraph*{Properties of these concepts.}
Both $\memSkel$-prefix-independence and $\memSkel$-cycle-consistency hold symmetrically for a winning condition and its complement, and are stable by product with an arbitrary skeleton (as products generate even smaller sets of prefixes and cycles to consider).

\begin{lemma} \label{lem:stableProduct}
	Let $\wc\subseteq\colors^\omega$ be a winning condition and $\memSkel$ be a skeleton.
	Then, $\wc$ is $\memSkel$-prefix-independent (resp.\ $\memSkel$-cycle-consistent) if and only if $\comp{\wc}$ is $\memSkel$-prefix-independent (resp.\ $\memSkel$-cycle-consistent).
	If $\wc$ is $\memSkel$-prefix-independent (resp.\ $\memSkel$-cycle-consistent), then for all skeletons $\memSkel'$, $\wc$ is $(\memSkel\memProduct\memSkel')$-prefix-independent (resp.\ $(\memSkel\memProduct\memSkel')$-cycle-consistent).
\end{lemma}

\begin{proof}
	Let $\memSkel = \memSkelFull$.

	We assume that $\wc$ is $\memSkel$-prefix-independent.
	Thus, for all $\memState\in\memStates$, for all $\word_1, \word_2\in\cc{\memPathsOn{\memInit}{\memState}}$, $\word_1 \prefEq \word_2$, i.e., $\inverse{\word_1}\wc = \inverse{\word_2}\wc$.
	This last equality is equivalent to $\comp{\inverse{\word_1}\wc} = \comp{\inverse{\word_2}\wc}$, which can be rewritten as $\inverse{\word_1}\comp{\wc} = \inverse{\word_2}\comp{\wc}$.
	This shows that $\comp{\wc}$ is $\memSkel$-prefix-independent.

	To show that $\wc$ is $\memSkel$-cycle-consistent if and only if $\comp{\wc}$ is $\memSkel$-cycle-consistent, notice that the winning cycles for $\wc$ are exactly the losing cycles for $\comp{\wc}$, and vice versa.

	Let $\memSkel' = (\memStates', \memInit', \memUpd')$ be a skeleton.
	We assume that $\wc$ is $\memSkel$-prefix-independent and we show that $\wc$ is $(\memSkel\memProduct\memSkel')$-prefix-independent.
	The sets of prefixes to consider are smaller in $\memSkel\memProduct\memSkel'$ than in $\memSkel$: for all $(\memState, \memState')\in\memStates\times\memStates'$, $\cc{\memPathsOn{(\memInit, \memInit')}{(\memState, \memState')}} \subseteq \cc{\memPathsOn{\memInit}{\memState}}$.
	Therefore, for all $\word_1, \word_2\in\cc{\memPathsOn{(\memInit, \memInit')}{(\memState, \memState')}}$, we also have $\word_1, \word_2\in\cc{\memPathsOn{\memInit}{\memState}}$, so by $\memSkel$-prefix-independence, $\word_1\prefEq\word_2$.

	We now assume that $\wc$ is $\memSkel$-cycle-consistent and we show that $\wc$ is $(\memSkel\memProduct\memSkel')$-cycle-consistent.
	The sets of winning and losing cycles to consider are smaller in $\memSkel\memProduct\memSkel'$ than in $\memSkel$: for all $\word\in\colors^*$, $\cc{\winCycWordAtmtn{\word}{\memSkel\memProduct\memSkel'}} \subseteq \cc{\winCycWordAtmtn{\word}{\memSkel}}$ and $\cc{\loseCycWordAtmtn{\word}{\memSkel\memProduct\memSkel'}} \subseteq \cc{\loseCycWordAtmtn{\word}{\memSkel}}$.
	By $\memSkel$-cycle-consistency, for all $\word\in\colors^*$, we have $(\cc{\winCycWord{\word}})^\omega \subseteq \inverse{\word}\wc$ and $(\cc{\loseCycWord{\word}})^\omega \subseteq \inverse{\word}\comp{\wc}$, so we also have $(\cc{\winCycWordAtmtn{\word}{\memSkel\memProduct\memSkel'}})^\omega \subseteq \inverse{\word}\wc$ and $(\cc{\loseCycWordAtmtn{\word}{\memSkel\memProduct\memSkel'}})^\omega \subseteq \inverse{\word}\comp{\wc}$.
\end{proof}

An interesting property of languages defined by a parity automaton $\parAtmtn$ is that they satisfy both aforementioned concepts with skeleton $\memSkel$.
\begin{lemma} \label{lem:parConcepts}
	Let $\wc\subseteq\colors^\omega$ be a winning condition and $\parAtmtn$ be a parity automaton.
	If $\wc$ is recognized by $\parAtmtn$, then $\wc$ is $\memSkel$-prefix-independent and $\memSkel$-cycle-consistent.
\end{lemma}
\begin{proof}
	By definition of the parity acceptance condition, any two finite words reaching the same state of the skeleton have the same winning continuations.
	Therefore, $\wc$ is $\memSkel$-prefix-independent.

	Also, the winning (resp.\ losing) cycles of $\memSkel$ after any finite word are exactly the ones that have an even (resp.\ odd) maximal priority.
	Therefore, combining infinitely many winning (resp.\ losing) cycles can only produce a winning (resp.\ losing) infinite word.
\end{proof}

\subparagraph*{Main results.}
We state our main technical tool.
We recall that one-player $\memSkel$-determinacy of a winning condition $\wc$ is both about one-player arenas of $\Pone$ (trying to achieve a word in $\wc$) \emph{and} of $\Ptwo$ (trying to achieve a word in $\comp{\wc}$).
\newcommand{\firstItem}{If there exists a skeleton $\memSkel$ such that $\wc$ is one-player $\memSkel$-determined, then $\prefEq$ has finite index (in particular, $\wc$ is $\minStateAtmtn$-prefix-independent) and~$\wc$ is $\memSkel$-cycle-consistent.}
\newcommand{\secondItem}{If there exists a skeleton $\memSkel$ such that $\wc$ is $\memSkel$-prefix-independent and $\memSkel$-cycle-consistent, then $\wc$ is $\omega$-regular and can be recognized by a deterministic parity automaton defined on top of $\memSkel$.}
\begin{theorem} \label{thm}
	Let $\wc\subseteq\colors^\omega$ be a winning condition and $\prefEq$ be its right congruence.
	\begin{enumerate}
		\item \firstItem
		\item \secondItem
	\end{enumerate}
\end{theorem}

We prove this theorem in Sections~\ref{sec:proof1} and~\ref{sec:proof2}.
We state two consequences of this result that were already mentioned in the introduction: a strategic characterization of $\omega$-regular languages, and a novel one-to-two-player-lift.

\begin{theorem}[Characterization] \label{thm:regular}
	Let $\wc\subseteq\colors^\omega$ be a language of infinite words.
	Language $\wc$ is $\omega$-regular if and only if it is chromatic-finite-memory determined (in infinite arenas).
\end{theorem}
\begin{proof}
	One implication is well-known~\cite{Mos84,Zie98}: if $\wc$ is $\omega$-regular, then it can be recognized by a deterministic parity automaton whose skeleton can be used as a memory that suffices to play optimally for both players, in arenas of any cardinality.
	For the other direction, if $\wc$ is chromatic-finite-memory determined, then there exists a skeleton $\memSkel$ such that $\wc$ is $\memSkel$-determined.
	In particular, $\wc$ is one-player $\memSkel$-determined, so by Theorem~\ref{thm}, $\prefEq$ has finite index and $\wc$ is $\memSkel$-cycle-consistent.
	In particular, by Lemma~\ref{lem:stableProduct}, $\wc$ is $(\minStateAtmtn\memProduct\memSkel)$-prefix-independent and $(\minStateAtmtn\memProduct\memSkel)$-cycle-consistent, so $\wc$ is $\omega$-regular and can be recognized by a deterministic parity automaton defined on top of~$\minStateAtmtn\memProduct\memSkel$.
\end{proof}

\begin{theorem}[One-to-two-player lift] \label{thm:lift}
	Let $\wc\subseteq\colors^\omega$ be a winning condition.
	Language $\wc$ is \textbf{one-player} chromatic-finite-memory determined if and only if it is chromatic-finite-memory determined.
\end{theorem}
\begin{proof}
	The implication from two-player to one-player determinacy is trivial.
	The other implication is given by Theorem~\ref{thm}: if $\wc$ is one-player $\memSkel$-determined, then $\prefEq$ has finite index and $\wc$ is $\memSkel$-cycle-consistent.
	Again by Lemma~\ref{lem:stableProduct} and Theorem~\ref{thm}, as $\wc$ can be recognized by a parity automaton defined on top of $\minStateAtmtn\memProduct\memSkel$, $\wc$ is determined and strategies based on $\minStateAtmtn\memProduct\memSkel$ suffice to play optimally in all two-player arenas.
\end{proof}

We discuss two specific situations in which we can easily derive interesting consequences using our results: the prefix-independent case, and the case where the minimal-state automaton suffices to play optimally.
\subparagraph*{Prefix-independent case.}
If a condition $\wc$ is prefix-independent (i.e., $\prefEq$ has index $1$ and $\minStateAtmtn = \memSkelTriv$), and skeleton $\memSkel$ suffices to play optimally in one-player games, then $\wc$ is recognized by a parity automaton defined on top of $\memSkelTriv\memProduct\memSkel$, which is isomorphic to $\memSkel$.
This implies that the exact same memory $\memSkel$ can be used by both players to play optimally in two-player arenas, with no increase in memory.
Note that, in general, when $\memSkel$ suffices to play optimally in one-player arenas, we do not know whether taking the product of $\memSkel$ with $\minStateAtmtn$ is necessary to play optimally in two-player arenas.
Still, the question is automatically solved for prefix-independent conditions.

If, moreover, $\memSkel = \memSkelTriv$ (i.e., memoryless strategies suffice to play optimally in one-player arenas), we recover exactly the result from Colcombet and Niwi\'nski~\cite{CN06}: $\wc$ can be recognized by a parity automaton defined on top of $\memSkelTriv$, so we can directly assign a priority to each color with a function $\pri \colon \colors \to \{0, \ldots, n\}$ such that an infinite word $\word = \clr_1\clr_2\ldots \in \colors^\omega$ is in $\wc$ if and only if $\limsup_{i\ge 1} \pri(\clr_i)$ is even.

\subparagraph*{$\minStateAtmtn$-determined case.}
An interesting property of some $\omega$-regular languages is that they can be recognized by defining an acceptance condition on top of the minimal-state automaton of their right congruence~\cite{MS97}, which is a useful property for the learning of languages~\cite{AF21}.
Here, Theorem~\ref{thm} implies that $\wc$ can be recognized by defining a transition-based parity acceptance condition on top of the minimal-state automaton $\minStateAtmtn$ if and only if $\wc$ is $\minStateAtmtn$-determined (and more precisely, if and only if $\wc$ is $\minStateAtmtn$-cycle-consistent).
The transition-based parity acceptance condition was not considered in the cited results~\cite{MS97,AF21}.

\begin{corollary}
	Let $\wc\subseteq\colors^\omega$ be an $\omega$-regular language and $\minStateAtmtn$ be the minimal-state automaton of its right congruence.
	The following are equivalent:
	\begin{enumerate}
		\item $\wc$ is recognized by defining a transition-based parity acceptance condition on top of~$\minStateAtmtn$; \label{enum:informative}
		\item $\wc$ is $\minStateAtmtn$-determined;\label{enum:determined}
		\item $\wc$ is $\minStateAtmtn$-cycle-consistent.\label{enum:cyc}
	\end{enumerate}
\end{corollary}
\begin{proof}
	Implication \ref{enum:informative}.$\implies$\ref{enum:determined}.\ follows from the memoryless determinacy of parity games~\cite{Zie98}.
	Implication \ref{enum:determined}.$\implies$\ref{enum:cyc}.\ follows from the first item of Theorem~\ref{thm}.
	Implication \ref{enum:cyc}.$\implies$\ref{enum:informative}.\ follows from the second item of Theorem~\ref{thm}: we have by definition that $\wc$ is $\minStateAtmtn$-prefix-independent, so if it is additionally $\minStateAtmtn$-cycle-consistent, then $\wc$ can be recognized by a parity automaton defined on top of $\minStateAtmtn$.
\end{proof}

\subparagraph*{Classes of arenas.}
We discuss how much Theorem~\ref{thm} depends upon our model of arenas.

There are multiple conditions that are chromatic-finite-memory determined if we only consider \emph{finite arenas} (finitely many states and edges), but which are not in infinite arenas.
A few examples are discounted-sum games~\cite{Sha53}, mean-payoff games~\cite{EM79}, total-payoff games~\cite{GZ04}, one-counter games~\cite{BBE10} which are all memoryless-determined in finite arenas but which require infinite memory to play optimally in some infinite arenas (we discuss some of these in Section~\ref{sec:application}).
In particular, Theorem~\ref{thm:regular} tells us that the derived winning conditions are not $\omega$-regular.

Strangely, the fact that our arenas have colors on \emph{edges} and not on \emph{states} is crucial for the result.
Indeed, there exists a winning condition (a generalization of a parity condition with infinitely many priorities~\cite{GW06}) that is memoryless-determined in state-labeled infinite arenas, but not in edge-labeled infinite arenas (as we consider here).
This particularity was already discussed~\cite{CN06}, and it was also shown that the same condition is memoryless-determined in edge-labeled arenas with finite branching.
Therefore, the fact that we allow \emph{infinite branching} in our arenas is also necessary for Theorem~\ref{thm:regular}.
Another example of a winning condition with finite memory requirements in finitely branching arenas for one player but infinite memory requirements in infinitely branching arenas is presented in~\cite[Section~4]{CFH14}.

In Sections~\ref{sec:proof1} and~\ref{sec:proof2}, we prove respectively the first item and the second item of Theorem~\ref{thm}.

\section{Two properties of one-player chromatic-finite-memory determinacy} \label{sec:proof1}
Let $\wc \subseteq \colors^\omega$ be a winning condition, $\prefEq$ be the right congruence of $\wc$, and $\memSkel = \memSkelFull$ be a skeleton, fixed for this section.
We aim to show the first item of Theorem~\ref{thm}, which is that for a skeleton $\memSkel$, one-player $\memSkel$-determinacy of $\wc$ implies that $\prefEq$ has finite index and that $\wc$ is $\memSkel$-cycle-consistent.

\subparagraph*{Finite index of \texorpdfstring{$\prefEq$}{\~}.}
For $\word_1, \word_2\in\colors^*$, we define $\word_1\prefOrd\word_2$ if $\inverse{\word_1}\wc \subseteq \inverse{\word_2}\wc$ (meaning that any continuation that is winning after $\word_1$ is also winning after $\word_2$).
Relation ${\prefOrd} \subseteq {\colors^*\times \colors^*}$ is a preorder.
Notice that $\prefEq$ is equal to $\prefOrd \cap \invPrefOrd$.
We also define the strict preorder $\strictPrefOrd {=} \prefOrd {\setminus} \prefEq$.

We will use preorder $\prefOrd$ to deduce that $\prefEq$ has finite index by showing that under hypotheses about the optimality of strategies based on $\memSkel$ in one-player arenas, $(i)$ on each subset $\cc{\memPathsOn{\memInit}{\memState}}$ of $\colors^*$ for $\memState\in\memStates$, preorder $\prefOrd$ is total (Lemma~\ref{lem:total}) $(ii)$ on each subset $\cc{\memPathsOn{\memInit}{\memState}}$ of~$\colors^*$ for $\memState\in\memStates$, preorder $\prefOrd$ has no infinite increasing nor decreasing sequence (Lemma~\ref{lem:infDec}).

\begin{lemma} \label{lem:total}
	Assume $\Pone$ has optimal strategies based on $\memSkel$ on all its \textbf{one-player} arenas.
	Then, for all $\memState\in\memStates$, preorder $\prefOrd$ is total on $\cc{\memPathsOn{\memInit}{\memState}}$.
\end{lemma}
\begin{proof}
	Let $\memState\in\memStates$.
	Let $\word_1, \word_2 \in \cc{\memPathsOn{\memInit}{\memState}}$; we show that $\word_1\not\prefOrd\word_2$ implies $\word_2\prefOrd\word_1$.
	If $\word_1\not\prefOrd\word_2$, then there exists $\word_1' \in \colors^\omega$ such that $\word_1\word_1' \in \wc$ and $\word_2\word_1' \notin \wc$.
	We show that $\word_2 \prefOrd \word_1$, i.e., that $\inverse{\word_2}\wc \subseteq \inverse{\word_1}\wc$.
	Let $\word_2'\in\inverse{\word_2}\wc$.
	We build an infinite one-player arena of $\Pone$, depicted in Figure~\ref{fig:total}, that merges the ends of finite chains for $\word_1$ and $\word_2$ and the starts of the infinite chains for $\word_1'$ and for $\word_2'$ in a state $\s$.%
	\begin{figure}[t]
		\centering
		\begin{tikzpicture}[every node/.style={font=\small,inner sep=1pt}]
			\draw (0,0) node[rond] (s1) {};
			\draw ($(s1)-(0,1.2)$) node[rond] (s2) {};
			\draw ($(s1)!0.5!(s2)+(2,0)$) node[rond] (s3) {$\s$};
			\draw ($(s1)+(4,0)$) node[inner sep=3pt] (s4) {$\ldots$};
			\draw ($(s2)+(4,0)$) node[inner sep=3pt] (s5) {$\ldots$};
			\draw ($(s1)-(0.8,0)$) edge[-latex'] (s1);
			\draw ($(s2)-(0.8,0)$) edge[-latex'] (s2);
			\draw (s1) edge[-latex',decorate] node[above=4pt] {$\word_1$} (s3);
			\draw (s2) edge[-latex',decorate] node[below=4pt] {$\word_2$} (s3);
			\draw (s3) edge[-latex',decorate] node[above=4pt] {$\word_1'$} (s4);
			\draw (s3) edge[-latex',decorate] node[below=4pt] {$\word_2'$} (s5);
		\end{tikzpicture}
		\caption{Arena built in the proof of Lemma~\ref{lem:total}.
				 Squiggly arrows indicate a sequence of edges.}
		\label{fig:total}
	\end{figure}

	It is possible to win after seeing $\word_1$ or $\word_2$, by choosing respectively $\word_1'$ or $\word_2'$ in the merged state $\s$.
	Moreover, there must be a strategy based on $\memSkel$ that wins from the starts of the chains of both $\word_1$ and $\word_2$, which means that in both cases the same choice has to be made in $\s$ (as memory state $\memState$ is reached in both cases).
	Continuing to $\word_1'$ in $\s$ would be losing after $\word_2$, so $\word_2'$ must be winning after $\word_1$.
	Therefore, $\word_2'\in\inverse{\word_1}\wc$.
\end{proof}

\begin{lemma} \label{lem:infDec}
	Assume $\Pone$ has optimal strategies based on $\memSkel$ in all its \textbf{one-player} arenas.
	For all $\memState \in \memStates$, there is no infinitely decreasing sequence of finite words for $\prefOrd$ in $\cc{\memPathsOn{\memInit}{\memState}}$.
\end{lemma}
\begin{proof}
	Let $\memState\in\memStates$.
	Assume by contraposition that there is an infinitely decreasing sequence of finite words $\word_1 \prefOrdBis \word_2 \prefOrdBis \word_3 \prefOrdBis \ldots$, with $\word_i\in\cc{\memPathsOn{\memInit}{\memState}}$ for $i\ge 1$.
	Then for all $i \ge 1$, there exists $\word_i'\in\colors^\omega$ such that $\word_i\word_i'\in\wc$ and $\word_{i+1}\word_i'\notin\wc$.
	We create an infinite one-player arena of~$\Pone$ in which we merge the ends of chains for all $\word_i$ to the starts of chains for all $\word_i'$ --- as was done in Figure~\ref{fig:total}, but with infinitely many words entering $\s$ and leaving $\s$.
	In this arena, for all $i\ge 1$, it is always possible to win from the start of the chain for $\word_i$, but there is no strategy based on~$\memSkel$ winning from all the starts of the chains simultaneously.
	Therefore, $\memSkel$ is not sufficient to play optimally in all one-player arenas of $\Pone$.
\end{proof}

We will also use this last lemma from the point of view of $\Ptwo$.
If we were to define a preorder $\prefOrd'$ for $\Ptwo$ (using winning condition $\comp{\wc}$), symmetrically to $\prefOrd$ for $\Pone$, we would obtain $\word_1\prefOrd'\word_2$ if and only if $\word_2\prefOrd\word_1$ because for any finite word $\word\in\colors^*$, $\inverse{\word}\comp{\wc} = \comp{\inverse{\word}\wc}$.

We can now combine the results of Lemmas~\ref{lem:total} and~\ref{lem:infDec} to find that $\prefEq$ has finite index if $\wc$ is one-player $\memSkel$-determined.

\begin{lemma} \label{lem:prefInd}
	If both $\Pone$ and $\Ptwo$ have optimal strategies based on $\memSkel$ in their one-player arenas (i.e., if $\wc$ is one-player $\memSkel$-determined), then the right congruence $\prefEq$ has finite index.
\end{lemma}
\begin{proof}
	Using Lemma~\ref{lem:infDec} along with the hypothesis about $\Pone$, we have that for all $\memState\in\memStates$, there are no infinitely decreasing sequence of words in $\cc{\memPathsOn{\memInit}{\memState}}$ for $\prefOrd$.
	Using the same result replacing $\Pone$ with $\Ptwo$, we obtain that there is no infinitely decreasing sequence for $\prefOrd'$, or in other words, that there is no infinitely increasing sequence for $\prefOrd$.
	For $\memState\in\memStates$, as $\prefOrd$ is total in $\cc{\memPathsOn{\memInit}{\memState}}$ (Lemma~\ref{lem:total}), we conclude that there are only finitely many equivalence classes of $\prefEq$ in $\cc{\memPathsOn{\memInit}{\memState}}$.
	As $\memStates$ is finite, there are only finitely many equivalence classes of $\prefEq$ in~$\bigcup_{\memState\in\memStates} \cc{\memPathsOn{\memInit}{\memState}} = \colors^*$.
\end{proof}
Under one-player chromatic-finite-memory determinacy of $\wc$, we can therefore consider the minimal-state automaton $\minStateAtmtn$ of $\prefEq$.

\subparagraph*{\texorpdfstring{$\memSkel$}{M}-cycle-consistency of \texorpdfstring{$\wc$}{W}.}
We now prove in a straightforward way that one-player $\memSkel$-determinacy of $\wc$ implies $\memSkel$-cycle-consistency of $\wc$.

\begin{lemma} \label{lem:repeating}
	If both $\Pone$ and $\Ptwo$ have optimal strategies based on $\memSkel$ in their one-player arenas (i.e., if $\wc$ is one-player $\memSkel$-determined), then winning condition $\wc$ is $\memSkel$-cycle-consistent.
\end{lemma}
\begin{proof}
	Let $\word\in\colors^*$ and $\memState = \memUpdHat(\memInit, \word)$.
	We show that $(\cc{\winCycWord{\word}})^\omega \subseteq \inverse{\word}\wc$.
	If $\cc{\winCycWord{\word}}$ is empty, this is true.
	If not, let $\word_1, \word_2,\ldots{}$ be an infinite sequence of finite words in $\cc{\winCycWord{\word}}$ --- we show that the infinite word $\word_1\word_2\ldots{}$ is in $\inverse{\word}\wc$.
	We consider the infinite one-player arena of $\Ptwo$ depicted in Figure~\ref{fig:cycCons}: it starts with a chain for $\word$ from a state $\s_1$ to a state $\s_2$, and $\s_2$ offers a choice between cycles for each finite word in $\{\word_1,\word_2,\ldots\}$.
	In this arena, $\Ptwo$ has no winning strategy based on $\memSkel$ from $\s_1$, since the same memory state $\memState$ is always reached in $\s_2$ (hence the same choice must always be made in $\s_2$), and repeating any cycle in $\cc{\winCycWord{\word}}$ forever after~$\word$ is winning for $\Pone$ by definition of $\cc{\winCycWord{\word}}$.
	Therefore, $\Ptwo$ also has no winning strategy at all, which means in particular that the infinite word $\word_1\word_2\ldots{}$ must be a winning continuation of $\word$ --- $\word_1\word_2\ldots{}$ is in $\inverse{\word}\wc$.
	Hence, $(\cc{\winCycWord{\word}})^\omega \subseteq \inverse{\word}\wc$.%
	\begin{figure}[t]
		\centering
		\begin{tikzpicture}[every node/.style={font=\small,inner sep=1pt}]
			\draw (0,0) node[carre] (s1) {$\s_1$};
			\draw ($(s1)+(2,0)$) node[carre] (s2) {$\s_2$};
			\draw ($(s1)-(0.8,0)$) edge[-latex'] (s1);
			\draw (s1) edge[-latex',decorate] node[above=3pt] {$\word$} (s2);
			\draw (s2) edge[-latex',decorate,in=30,out=-30,distance=1.2cm,very thick] node[right=4pt] {$\word_1, \word_2, \ldots$} (s2);
		\end{tikzpicture}
		\caption{Infinite one-player arena of $\Ptwo$ used in the proof of Lemma~\ref{lem:repeating}.
		The thick squiggly arrow indicates a choice between sequences of edges for any word in $\{\word_1, \word_2, \ldots\}$.}
		\label{fig:cycCons}
	\end{figure}

	Using a similar one-player arena of $\Pone$, we can show symmetrically that $(\cc{\loseCycWord{\word}})^\omega \subseteq \inverse{\word}\comp{\wc}$ for all $\word\in\colors^*$.
\end{proof}

The reciprocal of this result is false, as shown in the following example.

\begin{example} \label{ex:abC}
	Let $\colors = \{a, b\}$ and $\wc = ab\colors^\omega$.
	If we consider the skeleton $\memSkel$ in Figure~\ref{fig:cycConsInsufficientToPlay} (left), then $\wc$ is $\memSkel$-cycle-consistent: for all finite words $\word$ except for $\emptyWord$ and $a$, either all continuations are winning (if $\word \in abC^*$) or all continuations are losing.
	If $\word$ is $\emptyWord$ or $a$, then it reaches state $\memInit$ of $\memSkel$, and the only cycles on $\memInit$ are in $a^+$, are losing, and are losing when infinitely many of them are combined into an infinite word.
	But this automaton does not suffice to play optimally in arena $\arena$ in Figure~\ref{fig:cycConsInsufficientToPlay} (center), as seeing $a$ does not change the state.

	Notice that the minimal-state automaton $\minStateAtmtn$, in Figure~\ref{fig:cycConsInsufficientToPlay} (right), has four states (corresponding to equivalence classes $\eqClass{\emptyWord}$, $\eqClass{a}$, $\eqClass{ab}$, and $\eqClass{b}$) and suffices to play optimally.%
	% \lipicsEnd
\end{example}
\begin{figure}[tbh]
\centering
\begin{minipage}{0.33\columnwidth}
	\centering
	\begin{tikzpicture}[every node/.style={font=\small,inner sep=1pt}]
		\draw (0,0) node[diamant] (s1) {$\memInit$};
		\draw ($(s1)+(2,0)$) node[diamant] (s2) {$\memState_2$};
		\draw ($(s1)+(0,1.1)$) edge[-latex'] (s1);
		\draw (s1) edge[-latex'] node[above=2pt] {$b$} (s2);
		\draw (s1) edge[-latex',in=-150,out=150,distance=0.8cm] node[left=2pt] {$a$} (s1);
		\draw (s2) edge[-latex',in=30,out=-30,distance=0.8cm] node[right=2pt] {$a, b$} (s2);
	\end{tikzpicture}
\end{minipage}%
\begin{minipage}{0.26\columnwidth}
	\centering
	\begin{tikzpicture}[every node/.style={font=\small,inner sep=1pt}]
		\draw (0,0) node[rond] (s1) {};
		\draw ($(s1)+(0,0.8)$) edge[-latex'] (s1);
		\draw (s1) edge[-latex',in=-150,out=150,distance=0.8cm] node[left=2pt] {$a$} (s1);
		\draw (s1) edge[-latex',in=30,out=-30,distance=0.8cm] node[right=2pt] {$b$} (s1);
	\end{tikzpicture}
\end{minipage}%
\begin{minipage}{0.41\columnwidth}
	\centering
	\begin{tikzpicture}[every node/.style={font=\small,inner sep=1pt}]
		\draw (0,0) node[diamant] (eps) {$\eqClass{\emptyWord}$};
		\draw ($(eps)+(2,1)$) node[diamant] (a) {$\eqClass{a}$};
		\draw ($(a)+(2,0)$) node[diamant] (ab) {$\eqClass{ab}$};
		\draw ($(eps)+(2,-1)$) node[diamant] (b) {$\eqClass{b}$};
		\draw ($(eps)+(0,1.1)$) edge[-latex'] (eps);
		\draw (eps) edge[-latex'] node[above=2pt] {$a$} (a);
		\draw (eps) edge[-latex'] node[below=2pt] {$b$} (b);
		\draw (a) edge[-latex'] node[above=2pt] {$b$} (ab);
		\draw (b) edge[-latex',in=30,out=-30,distance=0.8cm] node[right=2pt] {$a, b$} (b);
		\draw (ab) edge[-latex',in=30,out=-30,distance=0.8cm] node[right=2pt] {$a, b$} (ab);
	\end{tikzpicture}
\end{minipage}
\caption{Skeleton $\memSkel$ (left), arena $\arena$ (center) and skeleton $\minStateAtmtn$ (right) used in Example~\ref{ex:abC}.}
\label{fig:cycConsInsufficientToPlay}
\end{figure}

\begin{remark}
	As discussed in Section~\ref{sec:thm}, Theorem~\ref{thm} does not hold if we assume chromatic-finite-memory determinacy in arenas in which \emph{states} rather than edges are labeled with colors.
	Lemma~\ref{lem:repeating} is an example of a step in the proof of Theorem~\ref{thm} that would not work with state-labeled arenas: the construction in Figure~\ref{fig:cycCons} would not work (there would have to be a color labeling $\s_2$ seen at the start of every cycle, but words $\word_i$ cannot all start with the same color in general).
	There is a winning condition that is memoryless-determined in state-labeled arenas~\cite{GW06} for which it is straightforward to show that it is not $\memSkelTriv$-cycle-consistent.
	% \lipicsEnd
\end{remark}

We will often use a weaker implication of $\memSkel$-cycle-consistency, which is that a finite combination of winning cycles is still a winning cycle (i.e., if $\cyc, \cyc'\in\winCycWord{\word}$, then $\cyc\cyc'\in\winCycWord{\word}$).%

\subparagraph*{Wrap-up of the section.}
Thanks to the results from this section, we deduce the first item of Theorem~\ref{thm}.

\begin{corollary}[First item of Theorem~\ref{thm}]
	\firstItem
\end{corollary}
\begin{proof}
	Follows from Lemmas~\ref{lem:prefInd} and~\ref{lem:repeating}.
\end{proof}

In particular, we obtain from the previous result that if both players have optimal strategies based on $\memSkel$ in their one-player arenas, then $\wc$ is both $(\minStateAtmtn\memProduct\memSkel)$-prefix-independent and $(\minStateAtmtn\memProduct\memSkel)$-cycle-consistent (using Lemma~\ref{lem:stableProduct}).

\begin{remark}
	If we compare Example~\ref{ex:buchiABuchiB} ($\wc = \Buchi{a}\cap\Buchi{b}$) and Example~\ref{ex:abC} ($\wc = ab\colors^\omega$), we see that we can easily classify the prefixes of the former, but that information is not sufficient to play optimally: we need some more information to classify cycles.
	For the latter, it is possible to find a skeleton classifying cycles that is insufficient to play optimally, but a good classification of the prefixes suffices to play optimally.
	In general, in order to understand $\wc$, we need to have information about prefixes and about cycles, which is why, intuitively, skeleton $\minStateAtmtn\memProduct\memSkel$ turns out to be useful.
	% \lipicsEnd
\end{remark}

\begin{remark}
In the proofs of this section, we only ever used arenas with countably many states and edges.
This implies that we can actually formulate a slightly stronger version of Theorem~\ref{thm:lift} (one-to-two-player lift): chromatic-finite-memory determinacy in one-player \emph{countable} arenas is equivalent to chromatic-finite-memory determinacy in arenas of any cardinality.
% \lipicsEnd
\end{remark}

\section{From properties of a winning condition to \texorpdfstring{$\omega$}{omega}-regularity} \label{sec:proof2}
In this section, we fix a language $\wc\subseteq\colors^\omega$ and a skeleton $\memSkel = \memSkelFull$, and we assume that \textbf{$\wc$ is $\memSkel$-prefix-independent and $\memSkel$-cycle-consistent}.
Our goal is to show that $\wc$ can be recognized by a parity automaton defined on top of $\memSkel$ and is thus $\omega$-regular.
To do that, we show in multiple steps how to assign a priority to each transition of $\memSkel$ through a function $\pri\colon \memStates\times\colors\to \{0, \ldots, n\}$ so that $\wc$ is recognized by the parity automaton~$(\memSkel, \pri)$.

\subparagraph*{Simplified notations.}
\begin{sloppypar}
In this section, as we have $\memSkel$-prefix-independence and $\memSkel$-cycle-consistency assumptions about $\wc$, we extend some notations from Section~\ref{sec:preliminaries} for conciseness.
\end{sloppypar}

As $\wc$ is $\memSkel$-prefix-independent, for $\memState\in\memStates$, we write $\inverse{\memState}\wc$ for the set of infinite words that equals $\inverse{\word}\wc$ for any $\word\in\cc{\memPathsOn{\memInit}{\memState}}$.
Notice in particular that $\inverse{\memInit}\wc = \inverse{\emptyWord}\wc = \wc$.
Moreover, as we consider the property of $\memSkel$-cycle-consistency along with $\memSkel$-prefix-independence, the definition of $\memSkel$-cycle-consistency can be written by only quantifying over states of $\memSkel$ and not over all finite words.
The reason is that there are then only finitely many classes of finite words that matter, which correspond to the states of $\memSkel$.
We define a few more notations that only make sense under the $\memSkel$-prefix-independent hypothesis.
Let
\[
\winCyc{\memState} = \{\cyc\in\memCyclesOn{\memState} \mid (\colHatFin(\cyc))^\omega \in \inverse{\memState}\wc\}
\]
be the cycles on $\memState$ that induce winning words when repeated infinitely many times from $\memState$, and $\loseCyc{\memState}$ be their losing counterparts.
In this case, $\wc$ is $\memSkel$-cycle-consistent if and only if for all $\memState\in\memStates$, $(\cc{\winCyc{\memState}})^\omega \subseteq \inverse{\memState}\wc$ and $(\cc{\loseCyc{\memState}})^\omega \subseteq \inverse{\memState}\comp{\wc}$.
We call elements of $\winCyc{\memState}$ (resp.\ $\loseCyc{\memState}$) \emph{winning} (resp.\ \emph{losing}) \emph{cycles on $\memState$}.
The set of winning (resp.\ losing) cycles of $\memSkel$ (on any state) is denoted $\memWinCycles$ (resp.\ $\memLoseCycles$).
We write $\cycVal{\cyc}$ for the \emph{value} of a cycle: $\win$ if $\cyc\in\memWinCycles$, and $\lose$ if $\cyc\in\memLoseCycles$.

\subparagraph*{Proof ideas.}
Our intermediate technical lemmas will focus on cycles of $\memSkel$, how they relate to each other, and what happens when we combine them.
Our main tool is to define a preorder on cycles, which will help assign priorities to transitions of $\memSkel$ --- the aim being to define a parity condition on top of $\memSkel$ that recognizes $\wc$.
Intuitively, for some state $\memState$ of $\memSkel$, $\cyc\in\winCyc{\memState}$, and $\cyc'\in\loseCyc{\memState}$, we look at which cycle ``dominates'' the other, that is whether the combined cycle $\cyc\cyc'$ is in $\winCyc{\memState}$ (in which case $\cyc$ dominates $\cyc'$) or in $\loseCyc{\memState}$ (in which case $\cyc'$ dominates $\cyc$).
We will formalize this and show how to extend this idea to cycles that may not share any common state.

\begin{remark}
	One may wonder why we seek to define a parity condition on top of $\memSkel$ to prove that $\wc$ is $\omega$-regular, rather than a more general \emph{Muller condition} which would achieve the same goal.
	Indeed, using $\memSkel$-cycle-consistency and a recent result by Casares et al.~\cite[Section~5]{CCF21}, it is straightforward to show that we could relabel such a Muller automaton with a parity condition defining the same language.

	One of the obstacles in our context is that we may start with infinitely many colors; in order to prove $\omega$-regularity of $\wc$, we need to show at some point that many colors can be assumed to be equal (for $\wc$) in order to get finitely many classes of ``equivalent'' colors.
	The way we manage that, using the aforementioned idea of ordering cycles, actually brings us very close to directly defining a relevant parity condition on top of $\memSkel$ --- it does not appear that our proof technique can be easily simplified by trying to obtain a Muller condition.
	% \lipicsEnd
\end{remark}

\subparagraph*{Combining cycles on the same skeleton state.}
We first prove that ``shifting'' the start of a cycle does not alter its value.

\begin{lemma}[Shift independence] \label{lem:indOrd}
	Let $\memState_1, \memState_2\in\memStates$ be two states of $\memSkel$.
	Let $\memPth_1\in\memPathsOn{\memState_1}{\memState_2}$ and $\memPth_2\in\memPathsOn{\memState_2}{\memState_1}$; $\memPth_1\memPth_2$ is a cycle on $\memState_1$ and $\memPth_2\memPth_1$ is a cycle on $\memState_2$.
	Then, $\cycVal{\memPth_1\memPth_2} = \cycVal{\memPth_2\memPth_1}$.
\end{lemma}
\begin{proof}
	For all $\word_1\in\cc{\memPathsOn{\memState_1}{\memState_2}}$ and $\word_2\in\colors^\omega$, notice that
	\begin{align*}
		\word_1\word_2\in\inverse{\memState_1}\wc
		&\Longleftrightarrow \exists\word\in\cc{\memPathsOn{\memInit}{\memState_1}}, \word_1\word_2\in\inverse{\word}\wc \\
		&\Longleftrightarrow \exists\word\in\cc{\memPathsOn{\memInit}{\memState_1}}, \word_2\in\inverse{(\word\word_1)}\wc \\
		&\Longleftrightarrow \exists\word'\in\cc{\memPathsOn{\memInit}{\memState_2}}, \word_2\in\inverse{(\word')}\wc \\
		&\Longleftrightarrow \word_2\in\inverse{\memState_2}\wc.
	\end{align*}
	The third equivalence is due to the fact that $\word\word_1$ is in $\cc{\memPathsOn{\memInit}{\memState_2}}$ for the left-to-right implication, and to $\memSkel$-prefix-independence for the right-to-left implication; if there exists $\word'\in\cc{\memPathsOn{\memInit}{\memState_2}}$ such that $\word_2\in\inverse{(\word')}\wc$, then the same is true for any word in $\cc{\memPathsOn{\memInit}{\memState_2}}$.

	Going back to the statement of the lemma, we have that
	\begin{align*}
		\memPth_1&\memPth_2\in\winCyc{\memState_1}\\
		&\Longleftrightarrow (\colHatFin(\memPth_1\memPth_2))^\omega\in\inverse{\memState_1}\wc \\
		&\Longleftrightarrow \colHatFin(\memPth_1)(\colHatFin(\memPth_2\memPth_1))^\omega\in\inverse{\memState_1}\wc &&\text{as $(\colHatFin(\memPth_1\memPth_2))^\omega = \colHatFin(\memPth_1)(\colHatFin(\memPth_2\memPth_1))^\omega$} \\
		&\Longleftrightarrow (\colHatFin(\memPth_2\memPth_1))^\omega\in\inverse{\memState_2}\wc &&\text{by the above property as $\colHatFin(\memPth_1)\in\cc{\memPathsOn{\memState_1}{\memState_2}}$}\\
		&\Longleftrightarrow \memPth_2\memPth_1\in\winCyc{\memState_2}.
	\end{align*}
	Hence, the values of $\memPth_1\memPth_2$ and $\memPth_2\memPth_1$ are always the same.
\end{proof}

In particular, this result implies that swapping two cycles on the same skeleton state does not alter the value: if $\cyc, \cyc'\in\memCyclesOn{\memState}$, then $\cycVal{\cyc\cyc'} = \cycVal{\cyc'\cyc}$.

The next two lemmas are used to show that although cycles of $\memSkel$ that are taken infinitely often might have an impact on the result of a play, their \emph{relative number of repetitions} is not relevant (i.e., $\cycVal{\cyc\cyc'} = \cycVal{\cyc^k(\cyc')^l}$ for any $k, l \ge 1$).
These two proofs and statements are very close to~\cite[Lemmas~9,~10, and~11]{CN06} and are a direct generalization to a larger class of winning conditions.

\begin{lemma} \label{lem:uniformDominating}
	Let $\memState\in\memStates$.
	Let $\memSet, \memSet'\subseteq \memCyclesOn{\memState}$ be non-empty sets of cycles on $\memState$.
	We have
	\[
	\forall \cyc'\in\memSet', \exists \cyc\in\memSet, \cyc\cyc' \in \winCyc{\memState} \implies
	\exists \cyc\in\memSet, \forall \cyc'\in\memSet', \cyc\cyc' \in \winCyc{\memState}.
	\]
\end{lemma}
This lemma says that if all cycles from $\memSet'$ can be made winning by adjoining them a cycle from $\memSet$, then we can actually find a single cycle from $\memSet$ that makes all cycles from $\memSet'$ winning.%
\begin{proof}
	We assume the premise of the implication, and by contradiction, we assume that the conclusion is false.
	We therefore assume that
	\[
		\forall \cyc'\in\memSet', \exists \cyc\in\memSet, \cyc\cyc' \in \winCyc{\memState}
		\quad \text{and} \quad
		\forall\cyc\in\memSet, \exists\cyc'\in\memSet', \cyc\cyc' \in \loseCyc{\memState}.
	\]
	Let $\cyc_1$ be any word in $\memSet$.
	We build inductively an infinite sequence starting with $\cyc_1$ by alternating the use of the two assumptions.
	For $i \ge 1$, we take $\cyc_i'\in\memSet'$ such that $\cyc_i\cyc_i' \in \loseCyc{\memState}$ (using the second assumption), and we then take $\cyc_{i+1}\in\memSet$ such that $\cyc_{i+1}\cyc_i'\in\winCyc{\memState}$ (using the first assumption).

	We consider the infinite sequence $\cyc_1\cyc_1'\cyc_2\cyc_2'\cyc_3\ldots\in(\memStates\times\colors)^\omega$ such that for all $i\ge 1$, $\cyc_i\cyc_i'\in\loseCyc{\memState}$ and $\cyc_i'\cyc_{i+1}\in\winCyc{\memState}$ (we use that the order of cycles on $\memState$ does not matter, shown in Lemma~\ref{lem:indOrd}).
	We show that the infinite word $\colHatInf(\cyc_1\cyc_1'\cyc_2\cyc_2'\ldots)$ is both in $\inverse{\memState}\wc$ and in $\inverse{\memState}\comp{\wc}$ by pairing cycles in two different ways:
	\begin{itemize}
		\item the infinite sequence $(\cyc_1\cyc_1')(\cyc_2\cyc_2')\ldots{}$ is a sequence of losing cycles on $\memState$ and its projection to colors is therefore in $\inverse{\memState}\comp{\wc}$ by using $\memSkel$-cycle-consistency.
		\item the infinite word $\colHatInf( \cyc_1(\cyc_1'\cyc_2)(\cyc_2'\cyc_3)\ldots)$ is in $\inverse{\memState}\wc$ if and only if $\colHatInf((\cyc_1'\cyc_2)(\cyc_2'\cyc_3)\ldots)$ is in $\inverse{\memState}\wc$ by using that $\cyc_1\in\memCyclesOn{\memState}$ and $\memSkel$-prefix-independence of $\wc$.
		The sequence $(\cyc_1'\cyc_2)(\cyc_2'\cyc_3)\ldots$ is a sequence of winning cycles on $\memState$ and its projection to colors is in $\inverse{\memState}\wc$ by using $\memSkel$-cycle-consistency.
	\end{itemize}
	As $\inverse{\memState}\wc \cap \inverse{\memState}\comp{\wc} = \emptyset$, we have our contradiction.
\end{proof}

\begin{lemma}[Repetition independence] \label{lem:indRep}
	Let $\memState\in\memStates$.
	Let $\cyc, \cyc'\in\memCyclesOn{\memState}$ such that $\cyc\cyc'\in\winCyc{\memState}$.
	We have $\cyc(\cyc')^+ \subseteq \winCyc{\memState}$.
\end{lemma}
\begin{proof}
	We have that $\cyc$ or $\cyc'$ is in $\winCyc{\memState}$ --- otherwise, $\cyc\cyc'$ would be in $\loseCyc{\memState}$ by $\memSkel$-cycle-consistency.
	If $\cyc'$ is in $\winCyc{\memState}$, we notice that any element of $\cyc(\cyc')^+$ can be written as $(\cyc\cyc')(\cyc')^n$ for some $n\ge 0$, which is a combination of winning cycles on $\memState$.
	Using $\memSkel$-cycle-consistency, we thus get that $\cyc(\cyc')^+ \subseteq \winCyc{\memState}$.

	It is left to deal with the case $\cyc\in\winCyc{\memState}$ and $\cyc'\in\loseCyc{\memState}$.
	We first show by induction that for $n\ge 1$, $\cyc^n(\cyc')^n\in\winCyc{\memState}$.
	This is true by hypothesis for $n = 1$.
	We now assume it is true for some $n \ge 1$, and we show it is true for $n + 1$.
	By Lemma~\ref{lem:indOrd}, we have that $\cyc^{n+1}(\cyc')^{n+1}\in\winCyc{\memState}$ if and only if $\cyc^{n}(\cyc')^{n+1}\cyc = (\cyc^{n}(\cyc')^{n})(\cyc'\cyc)\in\winCyc{\memState}$, by swapping the order of $\cyc$ and $\cyc^{n}(\cyc')^{n+1}$.
	By induction hypothesis, $\cyc^{n}(\cyc')^{n}\in\winCyc{\memState}$; by hypothesis and by Lemma~\ref{lem:indOrd}, $\cyc'\cyc\in\winCyc{\memState}$.
	Therefore, by $\memSkel$-cycle-consistency, $(\cyc^{n}(\cyc')^{n})(\cyc'\cyc)$ is also in $\winCyc{\memState}$.

	We now define $\memSet = \cyc^+$ and $\memSet' = (\cyc')^+$.
	We have that for all elements $(\cyc')^n$ of $\memSet'$ (with $n \ge 1$), we have that $\cyc^n$ (an element of $\memSet$) is such that $\cyc^n(\cyc')^n\in\winCyc{\memState}$.
	Therefore the hypothesis of Lemma~\ref{lem:uniformDominating} is verified for $\memSet$ and $\memSet'$, which implies that there exists $n\ge 1$ such that $\cyc^n(\cyc')^+ \subseteq \winCyc{\memState}$.

	We assume w.l.o.g.\ that $n = \min \{n\in\IN \mid \cyc^n(\cyc')^+ \subseteq \winCyc{\memState}\}$.
	For all $k$, $k \ge n$, we also have that $\cyc^k(\cyc')^+ = \cyc^{k-n}(\cyc^n(\cyc')^+) \subseteq \winCyc{\memState}$ by $\memSkel$-cycle-consistency.
	We intend to show that $n = 1$, which would end the proof of the lemma as this would show that $\cyc^1(\cyc')^+ = \cyc(\cyc')^+ \subseteq \winCyc{\memState}$.

	We assume by contradiction that $n > 1$.
	Then there must exist $k\in\IN$ such that $\cyc^{n-1}(\cyc')^k\in\loseCyc{\memState}$.
	We also have that $(\cyc')^k\cyc^{n-1}$ is in $\loseCyc{\memState}$ by Lemma~\ref{lem:indOrd}, which implies that $\cyc^{n-1}(\cyc')^k(\cyc')^k\cyc^{n-1}$ is also in $\loseCyc{\memState}$ by $\memSkel$-cycle-consistency.
	But then by Lemma~\ref{lem:indOrd}, this cycle has the same value as $\cyc^{2n-2}(\cyc')^{2k}$, which must therefore be in $\loseCyc{\memState}$.
	This is a contradiction since $n > 1$ implies that $2n - 2 \ge n$.

	We conclude that $\cyc(\cyc')^+\subseteq \winCyc{\memState}$.
\end{proof}

Thanks to this result, we can now show that any two consecutive cycles can always be swapped without changing the value of a longer cycle.
\begin{corollary}[Cycle-order independence] \label{lem:indOrdThree}
	Let $\memState\in\memStates$.
	Let $\cyc_1, \cyc_2, \cyc_3\in\memCyclesOn{\memState}$.
	Then,\newline $\cycVal{\cyc_1\cyc_2\cyc_3} = \cycVal{\cyc_1\cyc_3\cyc_2}$.
\end{corollary}
\begin{proof}
	We assume by contradiction that cycles $\cyc_1\cyc_2\cyc_3$ and $\cyc_1\cyc_3\cyc_2$ have a different value; w.l.o.g., that $\cyc_1\cyc_2\cyc_3\in\winCyc{\memState}$ and that $\cyc_1\cyc_3\cyc_2\in\loseCyc{\memState}$.
	By $\memSkel$-cycle-consistency, at least one cycle among $\cyc_1$, $\cyc_2$ and $\cyc_3$ is winning and one is losing.
	We assume w.l.o.g.\ that $\cyc_1\in\winCyc{\memState}$ and $\cyc_2\in\loseCyc{\memState}$.
	We also assume that $\cyc_3\in\winCyc{\memState}$; the other case can be dealt with by symmetry.
	Notice that we necessarily have that $\cyc_3\cyc_2$ is in $\loseCyc{\memState}$; otherwise, $\cyc_1\cyc_3\cyc_2 = \cyc_1(\cyc_3\cyc_2)$ would be in $\winCyc{\memState}$ by $\memSkel$-cycle-consistency.
	For the same reason, $\cyc_2\cyc_1$ is in $\loseCyc{\memState}$.
	We have
	\begin{align*}
		\win
		&= \cycVal{\cyc_1\cyc_2\cyc_3} &&\text{by hypothesis} \\
		&= \cycVal{(\cyc_3\cyc_1)\cyc_2} &&\text{by Lemma~\ref{lem:indOrd}} \\
		&= \cycVal{(\cyc_3\cyc_1)(\cyc_2)^2} &&\text{by Lemma~\ref{lem:indRep}} \\
		&= \cycVal{\cyc_2\cyc_3\cyc_1\cyc_2} &&\text{by Lemma~\ref{lem:indOrd}}.
	\end{align*}
	However, this last cycle can be written as a combination of two losing cycles $(\cyc_2\cyc_3)$ and $(\cyc_1\cyc_2)$, and should therefore be losing by $\memSkel$-cycle-consistency.
	This is a contradiction.
\end{proof}

\subparagraph*{Combining cycles on different skeleton states.}
We can now also strengthen Lemma~\ref{lem:indRep} (``repetition independence'') to show that even ``non-consecutive subcycles'' in a longer cycle can be repeated without affecting the value of the long cycle.

\begin{corollary}[Repetition independence, strong version] \label{lem:indRepStrong}
	Let $\memState_1, \memState_2\in\memStates$.
	Let $\cyc_1\in\memCyclesOn{\memState_1}$, $\cyc_2\in\memCyclesOn{\memState_2}$, $\wit_1\in\memPathsOn{\memState_1}{\memState_2}$, and $\wit_2\in\memPathsOn{\memState_2}{\memState_1}$.
	Then, $\cycVal{\cyc_1\wit_1\cyc_2\wit_2} = \cycVal{\cyc_1(\wit_1\wit_2)^n\wit_1\cyc_2\wit_2}$ for all $n \ge 0$.
\end{corollary}
The situation is depicted in Figure~\ref{fig:indCycle} (left).
Notice first that we can see $\cyc_1\wit_1\cyc_2\wit_2$ as a combination of two cycles $\cyc_1$ and $\wit_1\cyc_2\wit_2$ on $\memState_1$, we therefore already know that the value of $\cyc_1\wit_1\cyc_2\wit_2$ on $\memState_1$ is the same as the one of $(\cyc_1)^k(\wit_1\cyc_2\wit_2)^l$ for all $k, l \ge 1$.
This second cycle can be seen as two cycles $\wit_2\wit_1$ and $\cyc_2$ on $\memState_2$, we therefore know that the value of $\wit_1\cyc_2\wit_2$ on $\memState_2$ is the same as the one of $(\cyc_2)^k(\wit_2\wit_1)^l$ for all $k, l \ge 1$.
However, these two facts do not directly give the result as cycle $\wit_1\wit_2$ does not appear ``consecutively'' in $\cyc_1\wit_1\cyc_2\wit_2$.
\begin{proof}
	We have that
	\begin{align*}
		\cycVal{\cyc_1\wit_1\cyc_2\wit_2}
		&= \cycVal{\cyc_1(\wit_1\cyc_2\wit_2)(\wit_1\cyc_2\wit_2)} &&\text{by Lemma~\ref{lem:indRep} on $\memState_1$} \\
		&= \cycVal{\cyc_1\wit_1(\cyc_2)(\wit_2\wit_1)\cyc_2\wit_2} \\
		&= \cycVal{\cyc_1\wit_1(\wit_2\wit_1)(\cyc_2)\cyc_2\wit_2}&&\text{by Lemma~\ref{lem:indOrdThree} on $\memState_2$} \\
		&= \cycVal{\cyc_1(\wit_1\wit_2)\wit_1(\cyc_2)^2\wit_2} \\
		&= \cycVal{\cyc_1(\wit_1\wit_2)\wit_1\cyc_2\wit_2} &&\text{by Lemma~\ref{lem:indRep} on $\memState_2$.}
	\end{align*}
	This shows the result for $n = 1$; applying Lemma~\ref{lem:indRep} gives the result for all $n \ge 1$.
\end{proof}

Another important property that will help define an interesting preorder on cycles is that the value of a combination of two cycles is independent from the skeleton state chosen to compare pairs of cycles: if two cycles both go through two states $\memState_1$ and $\memState_2$ of $\memSkel$, then combining them around $\memState_1$ or around $\memState_2$ yields the same value.

\begin{lemma}[Crossing-point independence] \label{lem:indCycle}
	Let $\memState_1, \memState_2\in\memStates$ be two states of $\memSkel$.
	Let $\memPth_1, \memPth_1'\in\memPathsOn{\memState_1}{\memState_2}$ and $\memPth_2, \memPth_2'\in\memPathsOn{\memState_2}{\memState_1}$.
	We have that $\cycVal{\memPth_1\memPth_2\memPth_1'\memPth_2'} = \cycVal{\memPth_2\memPth_1\memPth_2'\memPth_1'}$.
\end{lemma}
\begin{figure}[t]
	\centering
	\begin{minipage}{0.5\columnwidth}
		\centering
		\begin{tikzpicture}[every node/.style={font=\small,inner sep=1pt}]
			\draw (0,0) node[diamant] (m1) {$\memState_1$};
			\draw ($(m1)+(2.5,0)$) node[diamant] (m2) {$\memState_2$};
			\draw (m1) edge[-latex',out=180-30,in=180+30,decorate,distance=0.8cm] node[left=3pt] {$\cyc_1$} (m1);
			\draw (m1) edge[-latex',out=30,in=180-30,decorate] node[above=3pt] {$\wit_1$} (m2);
			\draw (m2) edge[-latex',out=180+30,in=-30,decorate] node[below=3pt] {$\wit_2$} (m1);
			\draw (m2) edge[-latex',out=-30,in=+30,decorate,distance=0.8cm] node[right=3pt] {$\cyc_2$} (m2);
		\end{tikzpicture}
	\end{minipage}%
	\begin{minipage}{0.5\columnwidth}
		\centering
		\begin{tikzpicture}[every node/.style={font=\small,inner sep=1pt}]
			\draw (0,0) node[diamant] (s1) {$\memState_1$};
			\draw ($(s1)+(3.5,0)$) node[diamant] (s2) {$\memState_2$};
			\draw (s1) edge[-latex',out=45,in=180-45,decorate,distance=1.4cm] node[above=3pt] {$\memPth_1$} (s2);
			\draw (s2) edge[-latex',out=180-10,in=10,decorate] node[above=3pt] {$\memPth_2$} (s1);
			\draw (s1) edge[-latex',out=-45,in=-180+45,decorate,distance=1.4cm] node[below=3pt] {$\memPth_1'$} (s2);
			\draw (s2) edge[-latex',out=-165,in=-15,decorate] node[below=3pt] {$\memPth_2'$} (s1);
		\end{tikzpicture}
	\end{minipage}
	\caption{Depiction of the statement of Corollary~\ref{lem:indRepStrong} (left) and Lemma~\ref{lem:indCycle} (right).}
	\label{fig:indCycle}
\end{figure}
The intuition of this lemma is that if we take two cycles (in the statement, $\memPth_1\memPth_2$ and $\memPth_1'\memPth_2'$) that have two common states ($\memState_1$ and $\memState_2$), the chosen starting state to combine the two cycles (the combination is $(\memPth_1\memPth_2)(\memPth_1'\memPth_2')$ if $\memState_1$ is chosen, and $(\memPth_2\memPth_1)(\memPth_2'\memPth_1')$ if $\memState_2$ is chosen) has no impact on the value of the combination.
This situation is depicted in Figure~\ref{fig:indCycle} (right).
\begin{proof}
	If $\memPth_1\memPth_2$ and $\memPth_1'\memPth_2'$ are both in $\winCyc{\memState_1}$ or both in $\loseCyc{\memState_1}$, then $\memPth_2\memPth_1$ and $\memPth_2'\memPth_1'$ are also respectively both in $\winCyc{\memState_2}$ or both in $\loseCyc{\memState_2}$ by Lemma~\ref{lem:indOrd}.
	Therefore, we have our result using $\memSkel$-cycle-consistency.

	For the remaining cases, we assume w.l.o.g.\ that $\memPth_1\memPth_2\in\winCyc{\memState_1}$ and $\memPth_1'\memPth_2'\in\loseCyc{\memState_2}$.
	We will assume (again w.l.o.g.) that combining them is winning, i.e., that $\memPth_1\memPth_2\memPth_1'\memPth_2'\in\winCyc{\memState_1}$.
	Our goal is to show that $\memPth_2\memPth_1\memPth_2'\memPth_1'$ is also in $\winCyc{\memState_2}$.
	We assume by contradiction that it is not, i.e., that $\memPth_2\memPth_1\memPth_2'\memPth_1' \in \loseCyc{\memState_2}$.

	Observe that as $\memPth_1\memPth_2\memPth_1'\memPth_2'\in\winCyc{\memState_1}$, we also have $(\memPth_2\memPth_1')(\memPth_2'\memPth_1)\in\winCyc{\memState_2}$ by Lemma~\ref{lem:indOrd}.
	Hence, at least one of $\memPth_2\memPth_1'$ or $\memPth_2'\memPth_1$ must be a winning cycle on $\memState_2$, otherwise their combination would be losing on $\memState_2$ by $\memSkel$-cycle-consistency.
	Equivalently, by Lemma~\ref{lem:indOrd}, at least one of $\memPth_1'\memPth_2$ or $\memPth_1\memPth_2'$ must be a winning cycle on $\memState_1$.

	Similarly, as $\memPth_2\memPth_1\memPth_2'\memPth_1'$ is in $\loseCyc{\memState_2}$, we have that $(\memPth_1\memPth_2')(\memPth_1'\memPth_2)$ is in $\loseCyc{\memState_1}$.
	Hence, at least one of $\memPth_1\memPth_2'$ and $\memPth_1'\memPth_2$ is a losing cycle on $\memState_1$ by $\memSkel$-cycle-consistency.

	Our conclusions imply that exactly one of $\memPth_1\memPth_2'$ and $\memPth_1'\memPth_2$ is winning on $\memState_1$, and one is losing on $\memState_1$.
	Without loss of generality, we assume that $\memPth_1\memPth_2'\in\winCyc{\memState_1}$ and $\memPth_1'\memPth_2\in\loseCyc{\memState_1}$.

	We now have a value for all four two-word cycles on $\memState_1$ (and therefore for all four two-word cycles on $\memState_2$ by Lemma~\ref{lem:indOrd}): $\memPth_1\memPth_2$ and $\memPth_1\memPth_2'$ are in $\winCyc{\memState_1}$, and $\memPth_1'\memPth_2$ and $\memPth_1'\memPth_2'$ are in $\loseCyc{\memState_1}$.
	If we look at four-word cycles, we have already assumed w.l.o.g.\ that $\memPth_1\memPth_2\memPth_1'\memPth_2'\in\winCyc{\memState_1}$ and that $\memPth_1\memPth_2'\memPth_1'\memPth_2\in\loseCyc{\memState_1}$.
	We still do not know whether $\memPth_1\memPth_2\memPth_1'\memPth_2$ and $\memPth_1\memPth_2'\memPth_1'\memPth_2'$ are winning or losing --- no matter how we express them as two two-word cycles, one two-word cycle is winning and the other one is losing.
	We study the value of these two four-word cycles.

	Consider the cycle $(\memPth_2\memPth_1'\memPth_2'\memPth_1)(\memPth_2'\memPth_1\memPth_2\memPth_1')$ on $\memState_2$.
	It is winning, since $\memPth_2\memPth_1'\memPth_2'\memPth_1$ and $\memPth_2'\memPth_1\memPth_2\memPth_1'$ are both in $\winCyc{\memState_2}$: this can be shown using Lemma~\ref{lem:indOrd} and the fact that $\memPth_1\memPth_2\memPth_1'\memPth_2'$ is in $\winCyc{\memState_1}$.
	Therefore, by shifting the start of the cycle, $(\memPth_1'\memPth_2'\memPth_1\memPth_2')(\memPth_1\memPth_2\memPth_1'\memPth_2)$ is in $\winCyc{\memState_1}$.
	By $\memSkel$-cycle-consistency, this means that at least one of $\memPth_1'\memPth_2'\memPth_1\memPth_2'$ (equivalently, $\memPth_1\memPth_2'\memPth_1'\memPth_2'$) and $\memPth_1\memPth_2\memPth_1'\memPth_2$ is winning on~$\memState_1$.

	Similarly, we have that the cycle $(\memPth_2\memPth_1\memPth_2'\memPth_1')(\memPth_2'\memPth_1'\memPth_2\memPth_1)$ is in $\loseCyc{\memState_2}$.
	Therefore, by shifting the start of the cycle, $(\memPth_1\memPth_2'\memPth_1'\memPth_2')(\memPth_1'\memPth_2\memPth_1\memPth_2)$ is in $\loseCyc{\memState_1}$.
	This means that at least one of $\memPth_1\memPth_2'\memPth_1'\memPth_2'$ or $\memPth_1'\memPth_2\memPth_1\memPth_2$ (equivalently, $\memPth_1\memPth_2\memPth_1'\memPth_2$) is in $\loseCyc{\memState_1}$.

	Our conclusions imply that exactly one of $\memPth_1\memPth_2\memPth_1'\memPth_2$ and $\memPth_1\memPth_2'\memPth_1'\memPth_2'$ is winning on $\memState_1$, and one is losing on $\memState_1$.
	We consider both cases and draw a contradiction in each case.

	Assume that $\cycVal{\memPth_1\memPth_2\memPth_1'\memPth_2} = \lose$.
	Now consider the cycle $\cycBis = (\memPth_2'\memPth_1)(\memPth_2\memPth_1')^2$ on $\memState_2$.
	We have
	\begin{align*}
		\cycVal{\cycBis}
		&= \cycVal{\memPth_2'\memPth_1\memPth_2\memPth_1'} &&\text{by Lemma~\ref{lem:indRep}} \\
		&= \cycVal{\memPth_1\memPth_2\memPth_1'\memPth_2'} &&\text{by Lemma~\ref{lem:indOrd}} \\
		&= \win &&\text{by hypothesis.}
	\end{align*}
	However, we also have
	\begin{align*}
		\cycVal{\cycBis}
		&= \cycVal{(\memPth_1\memPth_2\memPth_1'\memPth_2)(\memPth_1'\memPth_2')} &&\text{by Lemma~\ref{lem:indOrd}},
	\end{align*}
	and since $\cycVal{\memPth_1\memPth_2\memPth_1'\memPth_2} = \cycVal{\memPth_1'\memPth_2'} = \lose$, we also have $\cycVal{\cycBis} = \lose$ by $\memSkel$-cycle-consistency.
	This is a contradiction.

	Assume now that $\cycVal{\memPth_1\memPth_2'\memPth_1'\memPth_2'} = \lose$.
	Now consider the cycle $\cycBis = (\memPth_1\memPth_2)(\memPth_1'\memPth_2')^2$ on $\memState_1$.
	We have
	\begin{align*}
		\cycVal{\cycBis}
		&= \cycVal{\memPth_1\memPth_2\memPth_1'\memPth_2'} &&\text{by Lemma~\ref{lem:indRep}} \\
		&= \win &&\text{by hypothesis.}
	\end{align*}
	However, we also have
	\begin{align*}
		\cycVal{\cycBis}
		&= \cycVal{(\memPth_2\memPth_1')(\memPth_2'\memPth_1'\memPth_2'\memPth_1)} &&\text{by Lemma~\ref{lem:indOrd}},
	\end{align*}
	and since $\cycVal{\memPth_1'\memPth_2} = \cycVal{\memPth_1\memPth_2'\memPth_1'\memPth_2'} = \lose$, we also have $\cycVal{\cycBis} = \lose$ by $\memSkel$-cycle-consistency.
	This is a contradiction.
\end{proof}

\begin{remark} \label{rem:crossPoint}
	A consequence of the previous lemma is that when two cycles $\cyc, \cyc'\in\memCycles$ share at least one common state (i.e., $\cycStates{\cyc} \cap \cycStates{\cyc'} \neq \emptyset$), we can write $\cyc\cyc'$ for any cycle that, starting from a common state, sees first $\cyc$ and then $\cyc'$, without necessarily specifying on which common state the cycle starts; we allow such a shortcut as the value of $\cyc\cyc'$ is not impacted by the choice of the common skeleton state.
	This convention is used in the following definition.
	% \lipicsEnd
\end{remark}

\subparagraph*{Competing cycles.}
\begin{definition}
	Let $\cyc, \cyc' \in \memCycles$ with $\cycVal{\cyc} \neq \cycVal{\cyc'}$.
	We say that \emph{$\cyc$ and $\cyc'$ are competing} if there exists $\wit\in\memCycles$ such that $\cycStates{\cyc} \cap \cycStates{\wit} \neq \emptyset$, $\cycStates{\cyc'} \cap \cycStates{\wit} \neq \emptyset$, $\cycVal{\cyc\wit} = \cycVal{\cyc}$, and $\cycVal{\cyc'\wit} = \cycVal{\cyc'}$.
	In this case, we say that $\wit$ is a \emph{witness that $\cyc$ and $\cyc'$ are competing}, or that \emph{the competition of $\cyc$ and $\cyc'$ is witnessed by $\wit$}.
\end{definition}
Our requirement for cycle $\wit$ means that it intersects the states of both $\cyc$ and $\cyc'$, but is not influential enough to ``alter'' the values of $\cyc$ and $\cyc'$ when it is combined with them.
If $\cycVal{\cyc} \neq \cycVal{\cyc'}$ and $\cycStates{\cyc} \cap \cycStates{\cyc'} \neq \emptyset$, then if $\cycVal{\cyc\cyc'} = \cycVal{\cyc}$ (resp.\ $\cycVal{\cyc\cyc'} = \cycVal{\cyc'}$), we have that $\cyc'$ (resp.\ $\cyc$) witnesses that $\cyc$ and $\cyc'$ are competing (the argument uses Lemma~\ref{lem:indRep}).
In short, any two cycles of opposite values that share a common state are competing, and two cycles of opposite values that do not share a common state may or may not be competing.%

If two cycles are competing, we want to determine which one \emph{dominates} the other.
\begin{definition}
	Let $\cyc, \cyc' \in \memCycles$ with $\cycVal{\cyc} \neq \cycVal{\cyc'}$ be two competing cycles, and $\wit$ be a witness of this competition.
	For some $\memState\in\cycStates{\cyc}$ and $\memState'\in\cycStates{\cyc'}$, it is thus possible to decompose~$\wit$ as two (possibly empty) paths $\wit_1$ and $\wit_2$ such that $\wit = \wit_1\wit_2$, $\wit_1\in\memPathsOn{\memState}{\memState'}$, and $\wit_2\in\memPathsOn{\memState'}{\memState}$.
	We define that \emph{$\cyc$ dominates $\cyc'$} if $\cycVal{\cyc\wit_1\cyc'\wit_2} = \cycVal{\cyc}$, and \emph{$\cyc'$ dominates $\cyc$} if $\cycVal{\cyc\wit_1\cyc'\wit_2} = \cycVal{\cyc'}$.
\end{definition}
To be well-defined, this \emph{domination} notion needs to be independent from the choice of witness.%

\begin{lemma}[Witness independence] \label{lem:indWit}
	Let $\cyc, \cyc' \in \memCycles$ with $\cycVal{\cyc} \neq \cycVal{\cyc'}$.
	Let $\wit_1, \wit_2\in\memCycles$ be two witnesses that $\cyc$ and $\cyc'$ are competing.
	Then, $\cyc$ dominates $\cyc'$ taking $\wit_1$ as a witness if and only if $\cyc$ dominates $\cyc'$ taking $\wit_2$ as a witness.
\end{lemma}
\begin{proof}
	We assume w.l.o.g.\ $\cycVal{\cyc} = \win$ and $\cycVal{\cyc'} = \lose$.
	As $\wit_1$ witnesses that $\cyc$ and $\cyc'$ are competing, there exists $\memState_1\in\cycStates{\cyc}$, $\memState_1'\in\cycStates{\cyc'}$ such that $\wit_1 = \wit_{1, 1}\wit_{1, 2}$ with $\wit_{1, 1}\in\memPathsOn{\memState_1}{\memState'_1}$, $\wit_{1, 2}\in\memPathsOn{\memState'_1}{\memState_1}$.
	Similarly, as $\wit_2$ witnesses that $\cyc$ and $\cyc'$ are competing, there exists $\memState_2\in\cycStates{\cyc}$, $\memState_2'\in\cycStates{\cyc'}$ such that $\wit_2 = \wit_{2, 1}\wit_{2, 2}$ with $\wit_{2, 1}\in\memPathsOn{\memState_2}{\memState'_2}$, $\wit_{2, 2}\in\memPathsOn{\memState'_2}{\memState_2}$.
	We can also write $\cyc = \cyc_1\cyc_2$ with $\cyc_1\in\memPathsOn{\memState_1}{\memState_2}$ and $\cyc_2\in\memPathsOn{\memState_2}{\memState_1}$, and $\cyc' = \cyc'_1\cyc'_2$ with $\cyc_1\in\memPathsOn{\memState'_1}{\memState'_2}$ and $\cyc'_2\in\memPathsOn{\memState'_2}{\memState'_1}$.

	\begin{figure}[t]
		\centering
		\begin{tikzpicture}[every node/.style={font=\small,inner sep=1pt}]
			\draw (0,0) node[diamant] (m1) {$\memState_1$};
			\draw ($(m1)-(0,2.8)$) node[diamant] (m2) {$\memState_2$};
			\draw ($(m1)+(3.5,0)$) node[diamant] (m1') {$\memState_1'$};
			\draw ($(m1)+(3.5,-2.8)$) node[diamant] (m2') {$\memState_2'$};
			\draw (m1) edge[-latex',decorate,out=30,in=180-30] node[below=3pt] {$\wit_{1, 1}$} (m1');
			\draw (m1') edge[-latex',decorate,out=-180+30,in=-30] node[below=3pt] {$\wit_{1, 2}$} (m1);
			\draw (m1) edge[-latex',decorate,out=-120,in=120] node[left=3pt] {$\cyc_1$} (m2);
			\draw (m2) edge[-latex',decorate,out=60,in=-60] node[right=3pt] {$\cyc_2$} (m1);
			\draw (m2) edge[-latex',decorate,out=30,in=180-30] node[above=3pt] {$\wit_{2, 1}$} (m2');
			\draw (m2') edge[-latex',decorate,out=-180+30,in=-30] node[above=3pt] {$\wit_{2, 2}$} (m2);
			\draw (m1') edge[-latex',decorate,out=-120,in=120] node[left=3pt] {$\cyc_1'$} (m2');
			\draw (m2') edge[-latex',decorate,out=60,in=-60] node[right=3pt] {$\cyc_2'$} (m1');
		\end{tikzpicture}
		\caption{Situation in the proof of Lemma~\ref{lem:indWit}.}
		\label{fig:indWit}
	\end{figure}

	The situation is depicted in Figure~\ref{fig:indWit}.
	Note that it is possible that $\memState_1 = \memState_2$ (in which case we can assume $\cyc = \cyc_1$, $\cyc_2 = (\memState_1, \emptyPthSymbol)$) or similarly that $\memState_1 = \memState_1'$, $\memState'_1 = \memState'_2$, and/or $\memState_2 = \memState_2'$.

\begin{sloppypar}
	We assume by contradiction that $\cyc$ dominates $\cyc'$ taking $\wit_1$ as a witness, but that~$\cyc'$ dominates~$\cyc$ taking $\wit_2$ as a witness.
	In other words, $\cycVal{\cyc_1\cyc_2\wit_{1, 1}\cyc'_1\cyc'_2\wit_{1, 2}} = \win$ and $\cycVal{\cyc_2\cyc_1\wit_{2,1}\cyc'_2\cyc'_1\wit_{2,2}} = \lose$.
	We consider the concatenation of both these cycles, shifting the second one to make it a cycle on $\memState_1$,
 \end{sloppypar}
	\[
	\cycBis = (\cyc_1\cyc_2\wit_{1, 1}\cyc'_1\cyc'_2\wit_{1, 2})(\cyc_1\wit_{2,1}\cyc'_2\cyc'_1\wit_{2,2}\cyc_2).
	\]
	It is possible to express $\cycBis$ directly as a combination of two losing cycles: $\wit_{1, 1}\cyc'_1\cyc'_2\wit_{1, 2}$ is losing (by definition of witness), and
	\begin{align*}
		\cycVal{\cyc_1\wit_{2,1}\cyc'_2\cyc'_1&\wit_{2,2}\cyc_2\cyc_1\cyc_2} \\
		&= \cycVal{\wit_{2,1}\cyc'_2\cyc'_1\wit_{2,2}(\cyc_2\cyc_1)^2} &&\text{by Lemma~\ref{lem:indOrd}}\\
		&= \cycVal{\wit_{2,1}\cyc'_2\cyc'_1\wit_{2,2}\cyc_2\cyc_1} &&\text{by Lemma~\ref{lem:indRep}} \\
		&= \lose &&\text{as $\cyc'$ dominates $\cyc$ taking $\wit_2$ as a witness}.
	\end{align*}
	Cycle $\cycBis$ is therefore losing by $\memSkel$-cycle-consistency.

	Now, notice that $\cycBis$ can be written as three cycles on $\memState_2'$ after being shifted in the following way:
	\[
	\cycVal{\cycBis} = \cycVal{(\cyc'_2\wit_{1,2}\cyc_1\wit_{2,1})(\cyc'_2\cyc'_1)(\wit_{2,2}\cyc_2\cyc_1\cyc_2\wit_{1,1}\cyc'_1)}.
	\]
	By Lemma~\ref{lem:indOrdThree} (``cycle-order independence''), it has the same value as
	\[
	\cycVal{\cycBis} = \cycVal{(\cyc'_2\wit_{1,2}\cyc_1\wit_{2,1})(\wit_{2,2}\cyc_2\cyc_1\cyc_2\wit_{1,1}\cyc'_1)(\cyc'_2\cyc'_1)}.
	\]
	As before, this cycle can be shifted and written as two winning cycles $\cyc_1\wit_{2,1}\wit_{2,2}\cyc_2$ and\linebreak $\cyc_1\cyc_2\wit_{1,1}(\cyc'_1\cyc'_2)^2\wit_{1,2}$ and is therefore winning by $\memSkel$-cycle-consistency.

	Cycle $\cycBis$ is both winning and losing, a contradiction.
\end{proof}

\begin{example} \label{ex:competition}
	We illustrate competition and domination of cycles on a parity automaton (even though at this point in the proof, we have not yet shown that $\wc$ is $\omega$-regular).
	We consider the parity automaton $\parAtmtn$ from Figure~\ref{fig:parityAtmtn} (left), with winning condition $\wc = \parLang$.
	Condition~$\wc$ is $\memSkel$-prefix-independent and $\memSkel$-cycle-consistent (Lemma~\ref{lem:parConcepts}).
	We give a few examples of competition and domination between cycles.
	The winning cycle $(\memState_1, a)(\memState_2, a)$ dominates losing cycle $(\memState_1, b)$ but is dominated by losing cycle $(\memState_1, c)$.
	Cycle $(\memState_2, b)$ is winning but is not competing with $(\memState_1, b)$.
	In particular, their competition is not witnessed by $(\memState_1, a)(\memState_2, a)$ since combining it with $(\memState_1, b)$ alters its value ($(\memState_1, b)$ is losing but $(\memState_1, b)(\memState_1, a)(\memState_2, a)$ is winning) --- another potential witness is $(\memState_1, c)(\memState_1, a)(\memState_2, a)$, but combining it with $(\memState_2, b)$ alters its value.
	However, cycle $(\memState_2, b)$ is competing with and dominated by $(\memState_1, c)$: their competition is witnessed, e.g., by $(\memState_1, a)(\memState_2, a)$ and by $(\memState_1, b)(\memState_1, a)(\memState_2, a)$.
	\begin{figure}
		\centering
		\begin{minipage}{0.55\columnwidth}
			\centering
			\begin{tikzpicture}[every node/.style={font=\small,inner sep=1pt}]
				\draw (0,0) node[diamant] (m1) {$\memState_1$};
				\draw ($(m1)+(2.5,0)$) node[diamant] (m2) {$\memState_2$};
				\draw ($(m1)+(0,0.8)$) edge[-latex'] (m1);
				\draw (m1) edge[-latex',out=180+60,in=270+30,distance=0.8cm] node[below=3pt] {$c \mid 3$} (m1);
				\draw (m1) edge[-latex',out=180-30,in=180+30,distance=0.8cm] node[left=3pt] {$b \mid 1$} (m1);
				\draw (m1) edge[-latex',out=30,in=180-30] node[above=3pt] {$a \mid 2$} (m2);
				\draw (m2) edge[-latex',out=180+30,in=-30] node[below=3pt] {$a \mid 2$} (m1);
				\draw (m2) edge[-latex',out=-30,in=+30,distance=0.8cm] node[right=3pt] {$b, c \mid 0$} (m2);
			\end{tikzpicture}
		\end{minipage}%
		\begin{minipage}{0.45\columnwidth}
			\centering
			\begin{tikzpicture}[every node/.style={font=\small,inner sep=1pt}]
				\draw (0,0) node[] (c) {$\eqClassCyc{(\memState_1, c)}$};
				\draw ($(c)+(-1.5,-1.5)$) node[] (aa) {$\eqClassCyc{(\memState_1, a)(\memState_2, a)}$};
				\draw ($(aa)+(0,-1.5)$) node[] (b) {$\eqClassCyc{(\memState_1, b)}$};
				\draw ($(c)+(1.5,-1.5)$) node[] (bc) {$\eqClassCyc{(\memState_2, b)}$};
				\draw (c) edge[] node[above left,xshift=-2pt,yshift=-1pt] {\rotatebox[origin=c]{90}{$\cycOrd$}} (aa);
				\draw (aa) edge[] node[left=2pt] {\rotatebox[origin=c]{90}{$\cycOrd$}} (b);
				\draw (c) edge[] node[above right,xshift=2pt,yshift=-1pt] {\rotatebox[origin=c]{90}{$\cycOrd$}} (bc);
			\end{tikzpicture}
		\end{minipage}
		\caption{A parity automaton $\parAtmtn$ (left) with $\colors = \{a, b, c\}$ used in Example~\ref{ex:competition}.
		Notation $\clr \mid k$ on a transition from a state $\memState$ means that $\pri(\memState, \clr) = k$.
	 	A diagram (right) showing the relations between the elements of $\quotient{\memCycles}{\cycEq}$ ordered by partial preorder $\cycOrd$, discussed in Example~\ref{ex:hasse}.}
		\label{fig:parityAtmtn}
	\end{figure}
	% \lipicsEnd
\end{example}

\subparagraph*{Preorder on cycles.}
For a winning (resp.\ losing) cycle $\cyc\in\memCycles$, we define $\compar{\cyc}$ as the set of losing (resp.\ winning) cycles that $\cyc$ is competing with, and $\dom(\cyc)$ as the set of losing (resp.\ winning) cycles from $\compar{\cyc}$ that are dominated by $\cyc$.

We define an ordering $\cycOrd$ of cycles based on these notions.
We write $\cyc' \cycOrd \cyc$ if $\cyc'\in\dom(\cyc)$.
We extend this definition to cycles with the same value: if $\cycVal{\cyc_1} = \cycVal{\cyc_2}$ and there exists a cycle $\cyc'$ such that $\cycVal{\cyc'} \neq \cycVal{\cyc_1}$ with $\cyc_2\in\dom(\cyc')$ and $\cyc'\in\dom(\cyc_1)$, we write $\cyc_2\cycOrd\cyc_1$ ---
intuitively, this condition implies that $\cyc_2$ is less powerful than $\cyc_1$ as we can find a cycle dominating $\cyc_2$ that is itself dominated by $\cyc_1$.
We show that relation $\cycOrd$ is a strict preorder (which is not total in general).

\begin{lemma} \label{lem:strictPreorder}
	Relation $\cycOrd$ is a strict preorder.
\end{lemma}
\begin{proof}
	We first prove that ${\cycOrd}$ is irreflexive, i.e., that for all $\cyc\in\memCycles$, $\cyc \not\cycOrd \cyc$.
	If $\cyc \cycOrd \cyc$, since $\cycVal{\cyc} = \cycVal{\cyc}$, there exists $\cyc'\in\memCycles$ such that $\cycVal{\cyc'} \neq \cycVal{\cyc}$, $\cyc \in \dom(\cyc')$, and $\cyc' \in \dom(\cyc)$.
	But that is not possible since when $\cyc$ and $\cyc'$ are competing, they cannot both dominate the other (no matter the choice of witness, as shown in Lemma~\ref{lem:indWit}).

	We now prove that $\cycOrd$ is transitive.
	We distinguish four cases (we rename cycles in each case to ease the reading by making it so that cycles with a prime symbol have a different value from cycles without a prime symbol).

	If $\cyc_2 \cycOrd \cyc'$ and $\cyc' \cycOrd \cyc_1$ with $\cycVal{\cyc_2} \neq \cycVal{\cyc'}$ and $\cycVal{\cyc'} \neq \cycVal{\cyc_1}$, then $\cycVal{\cyc_2} = \cycVal{\cyc_1}$, and $\cyc_2 \cycOrd \cyc_1$ by definition.

	Let $\cyc'_2 \cycOrd \cyc_2$ and $\cyc_2 \cycOrd \cyc_1$ with $\cycVal{\cyc'_2} \neq \cycVal{\cyc_2}$ and $\cycVal{\cyc_2} = \cycVal{\cyc_1}$.
	We assume w.l.o.g.\ that $\cycVal{\cyc_2} = \cycVal{\cyc_1} = \win$, so there exists $\cyc'_1\in\memCycles$ such that $\cycVal{\cyc'_1} = \lose$, $\cyc_2\cycOrd\cyc_1'$ and $\cyc_1'\cycOrd\cyc_1$.
	We assume that $\cyc_1'\cycOrd\cyc_1$ is witnessed by $\wit$, that $\cyc_2\cycOrd\cyc_1'$ is witnessed by $\wit'$ and that $\cyc_2' \cycOrd \cyc_2$ is witnessed by $\wit''$.
	We assume that $\wit = \wit_1\wit_2$, $\cyc_1' = \cyc_{1,1}'\cyc_{1,2}'$, $\wit' = \wit_1'\wit_2'$, $\cyc_2 = \cyc_{2,1}\cyc_{2,2}$, and $\wit'' = \wit_1''\wit_2''$.
	We use Figure~\ref{fig:transitive} to illustrate the situation and explain where the common states of these cycles are.

	\begin{figure}[t]
		\centering
		\begin{tikzpicture}[every node/.style={font=\small,inner sep=1pt}]
			\draw (0,0) node[diamant] (m1) {};
			\draw ($(m1)+(2.1,0)$) node[diamant] (m2) {};
			\draw ($(m2)+(2.1,0)$) node[diamant] (m3) {};
			\draw ($(m3)+(2.1,0)$) node[diamant] (m4) {};
			\draw ($(m4)+(2.1,0)$) node[diamant] (m5) {};
			\draw ($(m5)+(2.1,0)$) node[diamant] (m6) {};
			\draw (m1) edge[-latex',decorate,out=150,in=210,distance=0.8cm] node[left=3pt] {$\cyc_1$} (m1);
			\draw (m1) edge[-latex',decorate,out=30,in=180-30] node[above=3pt] {$\wit_1$} (m2);
			\draw (m2) edge[-latex',decorate,out=-180+30,in=-30] node[below=3pt] {$\wit_2$} (m1);
			\draw (m2) edge[-latex',decorate,out=30,in=180-30] node[above=3pt] {$\cyc_{1,1}'$} (m3);
			\draw (m3) edge[-latex',decorate,out=-180+30,in=-30] node[below=3pt] {$\cyc_{1,2}'$} (m2);
			\draw (m3) edge[-latex',decorate,out=30,in=180-30] node[above=3pt] {$\wit_1'$} (m4);
			\draw (m4) edge[-latex',decorate,out=-180+30,in=-30] node[below=3pt] {$\wit_2'$} (m3);
			\draw (m4) edge[-latex',decorate,out=30,in=180-30] node[above=3pt] {$\cyc_{2,1}$} (m5);
			\draw (m5) edge[-latex',decorate,out=-180+30,in=-30] node[below=3pt] {$\cyc_{2,2}$} (m4);
			\draw (m5) edge[-latex',decorate,out=30,in=180-30] node[above=3pt] {$\wit_1''$} (m6);
			\draw (m6) edge[-latex',decorate,out=-180+30,in=-30] node[below=3pt] {$\wit_2''$} (m5);
			\draw (m6) edge[-latex',decorate,out=-30,in=30,distance=0.8cm] node[right=3pt] {$\cyc_2'$} (m6);
		\end{tikzpicture}
		\caption{Situation to show transitivity of $\cycOrd$ in Lemma~\ref{lem:strictPreorder}.}
		\label{fig:transitive}
	\end{figure}

	We want to show that $\cyc_2'\cycOrd\cyc_1$.
	To do so, we show that for $\cycBis_1 = \wit_1\cyc'_{1,1}\wit'_1\cyc_{2,1}\wit''_1$ and $\cycBis_2 = \wit''_2\cyc_{2,2}\wit'_2\cyc'_{1,2}\wit_2$, $\cycBis = \cycBis_1\cycBis_2$ is a witness that $\cyc_1$ and $\cyc_2'$ are competing, and that $\cyc_1\cycBis_1\cyc_2'\cycBis_2$ is winning:
	\begin{itemize}
		\item \underline{Cycle $\cyc_1\cycBis$ is winning.}
		We can split this cycle into $\cyc'_{1,2}\wit_2\cyc_1\wit_1\cyc'_{1,1}$ and $\wit'_1\cyc_{2,1}\wit''_1\wit''_2\cyc_{2,2}\wit'_2$.
		The former cycle can be shifted (Lemma~\ref{lem:indOrd}) to $\cyc_1\wit_1\cyc'_{1,1}\cyc'_{1,2}\wit_2$, which is winning since $\cyc_1' \cycOrd \cyc_1$.
		The latter cycle has the same value as $\wit'_1(\cyc_{2,2}\cyc_{2,1})\cyc_{2,1}\wit''_1\wit''_2\cyc_{2,2}\wit'_2$ (Lemma~\ref{lem:indRepStrong}), which can be shifted to $(\wit'_2\wit'_1\cyc_{2,2}\cyc_{2,1})(\cyc_{2,1}\wit''_1\wit''_2\cyc_{2,2})$.
		Both these cycles are winning since $\wit'$ and $\wit''$ are witnesses for competitions involving $\cyc_2$.
		\item \underline{Cycle $\cyc_2'\cycBis$ is losing.}
		We can split this cycle into $\wit''_1\cyc_2'\wit''_2$ (losing because $\wit''$ witnesses a competition involving $\cyc_2'$) and $\cyc_{2,2}\wit'_2\cyc'_{1,2}\wit_2\wit_1\cyc'_{1,1}\wit'_1\cyc_{2,1}$.
		This latter cycle has the same value as the cycle $\cyc_{2,2}\wit'_2(\cyc'_{1,2}\cyc'_{1,1})\cyc'_{1,2}\wit_2\wit_1\cyc'_{1,1}\wit'_1\cyc_{2,1}$ (Lemma~\ref{lem:indRepStrong}), which can itself be split into $\cyc'_{1,2}\wit_2\wit_1\cyc'_{1,1}$ (losing because $\wit$ witnesses a competition involving $\cyc'$) and $\wit'_2\cyc'_{1,2}\cyc'_{1,1}\wit'_1\cyc_{2,1}\cyc_{2,2}$ (losing because $\cyc_2 \cycOrd \cyc_1'$).
		\item \underline{Cycle $\cyc_1\cycBis_1\cyc_2'\cycBis_2$ is winning.}
		Using Lemma~\ref{lem:indRepStrong}, we can show that $\cyc_1\cycBis_1\cyc_2'\cycBis_2$ has the same value as $\cyc_1\cycBis_1(\wit''_2\cyc_{2,2}\cyc_{2,1}\wit''_1)\cyc_2'\cycBis_2$.
		We can split this cycle into $\cyc_1\cycBis_1\cycBis_2 = \cyc_1\cycBis$ (which is winning, as shown above) and $\wit''_2\cyc_{2,2}\cyc_{2,1}\wit''_1\cyc_2'$ (winning since $\cyc'_2\cycOrd \cyc_2$).
	\end{itemize}
	This shows that $\cyc_2'\cycOrd\cyc_1$.

	There are still two cases left to consider.
	The case $\cyc_2' \cycOrd \cyc_1'$ and $\cyc_1' \cycOrd \cyc_1$ with $\cycVal{\cyc_2'} = \cycVal{\cyc_1'}$ and $\cycVal{\cyc_1'} \neq \cycVal{\cyc_1}$ can be dealt with in the same way as the previous case (after noticing that there exists $\cyc_2\in\memCycles$ such that $\cyc_2'\cycOrd \cyc_2$, $\cyc_2 \cycOrd \cyc_1'$ and $\cycVal{\cyc_2} = \cycVal{\cyc_1}$).

	If $\cyc_3 \cycOrd \cyc_2$ and $\cyc_2 \cycOrd \cyc_1$ with $\cycVal{\cyc_3} = \cycVal{\cyc_2} = \cycVal{\cyc_1}$, then there exists in particular $\cyc'\in\memCycles$ such that $\cyc_3 \cycOrd \cyc'$, $\cyc' \cycOrd \cyc_2$ and $\cycVal{\cyc_3} \neq \cycVal{\cyc'}$.
	By a previous case, we conclude from $\cyc' \cycOrd \cyc_2$ and $\cyc_2 \cycOrd \cyc_1$ that $\cyc' \cycOrd \cyc_1$.
	As $\cyc_3 \cycOrd \cyc'$ and $\cyc' \cycOrd \cyc_1$, we have $\cyc_3 \cycOrd \cyc_1$ as desired.
\end{proof}

We define an equivalence relation on the cycles: we write $\cyc_1 \cycEq \cyc_2$ if $\cycVal{\cyc_1} = \cycVal{\cyc_2}$, $\compar{\cyc_1} = \compar{\cyc_2}$, and $\dom(\cyc_1) = \dom(\cyc_2)$.
We show that cycles that are equivalent for $\cycEq$ are in relation with the same elements for $\cycOrd$.

\begin{lemma}
	Let $\cyc_1, \cyc_2, \cyc'\in\memCycles$.
	If $\cyc_1 \cycEq \cyc_2$ and $\cyc'\cycOrd \cyc_1$, then $\cyc' \cycOrd \cyc_2$.
	If $\cyc_1 \cycEq \cyc_2$ and $\cyc_1 \cycOrd \cyc'$, then $\cyc_2 \cycOrd \cyc'$.
	In other words, preorder $\cycOrd$ is compatible with $\cycEq$.
\end{lemma}
\begin{proof}
	The first item is trivial, as the elements smaller than $\cyc_1$ for ${\cycOrd}$ are determined by $\dom(\cyc_1)$, and $\dom(\cyc_1) = \dom(\cyc_2)$.
	For the second item, we distinguish two cases:
	\begin{itemize}
		\item if $\cycVal{\cyc_1} \neq \cycVal{\cyc'}$, then $\cyc_1 \cycOrd \cyc'$ means that $\cyc' \in \compar{\cyc_1}$ and $\cyc' \not\in \dom(\cyc_1)$.
		If $\cyc_1 \cycEq \cyc_2$, the same properties also hold for $\cyc_2$, so $\cyc_2\cycOrd \cyc'$.
		\item if $\cycVal{\cyc_1} = \cycVal{\cyc'}$, then $\cyc_1 \cycOrd \cyc'$ means that there exists $\cyc''$ with $\cycVal{\cyc_1} \neq \cycVal{\cyc''}$ such that $\cyc_1\in\dom(\cyc'')$ and $\cyc''\in\dom(\cyc')$.
		By the previous case, if $\cyc_1 \cycEq \cyc_2$, then $\cyc_2 \cycOrd \cyc''$, so $\cyc_2 \cycOrd \cyc'$ as well.
		\qedhere
	\end{itemize}
\end{proof}
Partial preorder $\cycOrd$ therefore also induces a partial preorder on $\quotient{\memCycles}{\cycEq}$.

\begin{example} \label{ex:hasse}
	We represent the relations between all elements of $\quotient{\memCycles}{\cycEq}$ for the parity automaton considered in Example~\ref{ex:competition} in their ``Hasse diagram'', depicted in Figure~\ref{fig:parityAtmtn} (right).
	Elements that are linked by a line segment are comparable for $\cycOrd$, and elements that are above are greater for $\cycOrd$.
	There are four equivalence classes of cycles, two of them winning and two of them losing.
	Notice for instance that any cycle going through transition $(\memState_1, c)$ is equivalent (for $\cycEq$) to cycle $(\memState_1, c)$: indeed, it is necessarily a losing cycle competing with and dominating all the winning cycles in this skeleton.
	Other examples are given by $(\memState_1, a)(\memState_2, a) \cycEq (\memState_1, a)(\memState_2, b)(\memState_2, a)$ and $(\memState_2, b) \cycEq (\memState_2, c)$.
	% \lipicsEnd
\end{example}

We now prove finiteness of the index of $\cycEq$, by showing that
\begin{itemize}
	\item the \emph{height} of partial preorder ${\cycOrd}$ is finite, i.e., there is no infinite increasing nor decreasing sequence for $\cycOrd$ (Lemma~\ref{lem:finHeight});
	\item the \emph{width} of partial preorder ${\cycOrd}$ on $\quotient{\memCycles}{\cycEq}$ is finite, i.e., there is no infinite set of elements in $\quotient{\memCycles}{\cycEq}$ that are all pairwise incomparable for $\cycOrd$ (Lemma~\ref{lem:finWidth}).
\end{itemize}
We start with two technical lemmas about competition between cycles.

\begin{lemma} \label{lem:weakWitness}
	Let $\cyc_1, \cyc_2, \cyc'\in\memCycles$ be such that $\cycVal{\cyc_1} = \cycVal{\cyc_2} \neq \cycVal{\cyc'}$, $\cyc_2 \cycOrd \cyc'$, and $\cyc' \cycOrd \cyc_1$.
	Let $\wit$ be a witness that $\cyc_2$ and $\cyc'$ are competing such that $\cycStates{\wit} \cap \cycStates{\cyc_1} \cap \cycStates{\cyc_2} \neq \emptyset$.
	Then, $\wit$ also witnesses that $\cyc_1$ and $\cyc'$ are competing.
\end{lemma}
\begin{proof}
	We already know that $\cycVal{\cyc'\wit} = \cycVal{\cyc'}$ and that $\cycStates{\cyc'}\cap\cycStates{\wit}\neq \emptyset$ (as $\wit$ witnesses a competition involving $\cyc'$) and that $\cycStates{\cyc_1}\cap\cycStates{\wit}\neq \emptyset$ (by hypothesis).
	It is left to show that $\cycVal{\cyc_1\wit} = \cycVal{\cyc_1}$.
	Let $\memState\in\cycStates{\wit} \cap \cycStates{\cyc_1} \cap \cycStates{\cyc_2}$; we represent the situation in Figure~\ref{fig:weakWitness} (left), with $\wit = \wit_1\wit_2$.
	Consider first cycle $\wit_1\wit_2\cyc_2$: this cycle is a witness that $\cyc_1$ and $\cyc'$ are competing, since it has common states with those cycles, $\cycVal{\cyc_1(\wit_1\wit_2\cyc_2)} = \cycVal{\cyc_1}$ (both $\cyc_1$ and $\wit_1\wit_2\cyc_2$ have the same value), and $\cycVal{\cyc'(\wit_2\cyc_2\wit_1)} = \cycVal{\cyc'}$ (since $\cyc_2\cycOrd\cyc'$ and $\wit$ is a witness of the competition).
	Therefore, as $\cyc' \cycOrd \cyc_1$, the cycle $\cycBis = \cyc_1\wit_1\cyc'\wit_2\cyc_2$ has the same value as $\cyc_1$.
	By Lemma~\ref{lem:indRepStrong}, cycle $\cycBis$ has the same value as $\cyc_1(\wit_1\wit_2)\wit_1\cyc'\wit_2\cyc_2$, which can be split into $\wit_1\cyc'\wit_2\cyc_2$ (which has the same value as $\cyc'$ since $\cyc_2 \cycOrd \cyc'$ and $\wit$ is a witness of the competition) and $\cyc_1\wit_1\wit_2$.
	Therefore, $\cyc_1\wit_1\wit_2 = \cyc_1\wit$ cannot have the same value as $\cyc'$, otherwise $\cycBis$ would also have the same value as $\cyc'$ by $\memSkel$-cycle-consistency.%
	\begin{figure}[t]
		\centering
		\begin{minipage}{0.45\columnwidth}
			\centering
			\begin{tikzpicture}[every node/.style={font=\small,inner sep=1pt}]
				\draw (0,0) node[diamant] (m1) {$\memState$};
				\draw ($(m1)+(2.5,0)$) node[diamant] (m2) {};
				\draw (m1) edge[-latex',out=90,in=150,decorate,distance=0.8cm] node[above left=2pt] {$\cyc_1$} (m1);
				\draw (m1) edge[-latex',out=210,in=270,decorate,distance=0.8cm] node[below left=2pt] {$\cyc_2$} (m1);
				\draw (m1) edge[-latex',out=30,in=180-30,decorate] node[above=3pt] {$\wit_1$} (m2);
				\draw (m2) edge[-latex',out=180+30,in=-30,decorate] node[below=3pt] {$\wit_2$} (m1);
				\draw (m2) edge[-latex',out=-30,in=30,decorate,distance=0.8cm] node[right=3pt] {$\cyc'$} (m2);
			\end{tikzpicture}
		\end{minipage}%
		\begin{minipage}{0.55\columnwidth}
			\centering
			\begin{tikzpicture}[every node/.style={font=\small,inner sep=1pt}]
				\draw (0,0) node[diamant] (m1) {};
				\draw ($(m1)+(2.1,0)$) node[diamant] (m2) {};
				\draw ($(m2)+(2.1,0)$) node[diamant] (m3) {};
				\draw ($(m1)!0.5!(m2)$) node[] () {$\wit$};
				\draw ($(m2)!0.5!(m3)$) node[] () {$\cyc_1'$};
				\draw (m1) edge[-latex',decorate,out=150,in=210,distance=0.8cm] node[left=3pt] {$\cyc$} (m1);
				\draw (m1) edge[-latex',decorate,out=30,in=180-30] (m2);
				\draw (m2) edge[-latex',decorate,out=-180+30,in=-30] (m1);
				\draw (m2) edge[-latex',decorate,out=30,in=180-30] (m3);
				\draw (m3) edge[-latex',decorate,out=-180+30,in=-30] node[below=3pt] {} (m2);
				\draw (m3) edge[-latex',decorate,out=-30,in=30,distance=0.8cm] node[right=3pt] {$\cyc_2'$} (m3);
			\end{tikzpicture}
		\end{minipage}
		\caption{Situation in the proof of Lemma~\ref{lem:weakWitness} (left) and in the proof of Lemma~\ref{lem:transWeak} (right).}
		\label{fig:weakWitness}
	\end{figure}
\end{proof}

\begin{lemma} \label{lem:transWeak}
	Let $\cyc, \cyc'_1\in\memCycles$ be such that $\cycVal{\cyc} \neq \cycVal{\cyc'_1}$ and $\cyc_1' \cycOrd \cyc$.
	Let $\cyc_2'$ be a cycle such that $\cycVal{\cyc_2'} = \cycVal{\cyc_1'}$ and $\cycStates{\cyc_2'} \cap \cycStates{\cyc_1'} \neq \emptyset$.
	Then, $\cyc$ and $\cyc_2'$ are competing.
\end{lemma}
\begin{proof}
	Let $\wit$ be a witness that $\cyc$ and $\cyc'_1$ are competing; we represent the situation in Figure~\ref{fig:weakWitness} (right).
	We show that $\wit\cyc_1'$ is a witness that $\cyc$ and $\cyc_2'$ are competing.
	As $\cycStates{\wit} \cap \cycStates{\cyc} \neq \emptyset$, we have $\cycStates{\wit\cyc_1'} \cap \cycStates{\cyc} \neq \emptyset$.
	Similarly, as $\cycStates{\cyc_1'} \cap \cycStates{\cyc_2'} \neq \emptyset$, we have $\cycStates{\wit\cyc_1'} \cap \cycStates{\cyc_2'} \neq \emptyset$.
	As $\cyc_1' \cycOrd \cyc$ with witness $\wit$, we have that $\cycVal{\cyc(\wit\cyc_1')} = \cycVal{\cyc}$.
	Moreover, since $\wit$ is a witness for $\cyc_1'$ (and~$\cyc'$), $\cycVal{\wit\cyc_1'} = \cycVal{\cyc_1'}$.
	Therefore $\cycVal{\wit\cyc_1'} = \cycVal{\cyc_2'}$, which implies by $\memSkel$-cycle-consistency that $\cycVal{\cyc_2'(\wit\cyc_1')} = \cycVal{\cyc_2'}$.
\end{proof}

\begin{lemma} \label{lem:finHeight}
	The height of partial preorder ${\cycOrd}$ is finite.
\end{lemma}
\begin{proof}
	By contradiction, let $\cyc_0 \cycOrdInv \cyc_0' \cycOrdInv \cyc_1 \cycOrdInv \cyc_1' \cycOrdInv \cyc_2 \cycOrdInv \ldots{}$ be an infinitely decreasing sequence for $\cycOrd$.
	We assume w.l.o.g.\ that for all $i \ge 0$, $\cycVal{\cyc_i} = \win$ and $\cycVal{\cyc_i'} = \lose$ (if two consecutive cycles are, for example, both winning, we can always insert an intermediate losing cycle between them, by definition).

	For $i \ge 0$, let $\memState$ (resp.\ $\memState'$) be a state of $\memSkel$ that is part of infinitely many sets $\cycStates{\cyc_i}$ (resp.\ $\cycStates{\cyc_i'}$) --- such states necessarily exist as the state space of $\memSkel$ is finite.
	Thanks to transitivity of $\cycOrd$ (Lemma~\ref{lem:strictPreorder}), by keeping only winning cycles $\cyc_i$ such that $\memState\in\cycStates{\cyc_i}$ alternating with losing cycles $\cyc_i'$ such that $\memState'\in\cycStates{\cyc_i'}$, we keep an infinitely decreasing sequence for $\cycOrd$.
	We can therefore assume w.l.o.g., up to renaming cycles, that for all $i\in\IN$, $\memState\in\cycStates{\cyc_i}$ and $\memState'\in\cycStates{\cyc_i'}$.

	\begin{figure}[t]
		\centering
		\begin{tikzpicture}[every node/.style={font=\small,inner sep=1pt}]
			\draw (0,0) node[diamant] (m1) {$\memState_1$};
			\draw ($(m1)-(0,2.8)$) node[diamant] (m2) {$\memState_2$};
			\draw ($(m1)!0.5!(m2)-(2.3,0)$) node[diamant] (m) {$\memState$};
			\draw ($(m1)+(3.5,0)$) node[diamant] (m1') {$\memState_1'$};
			\draw ($(m1)+(3.5,-2.8)$) node[diamant] (m2') {$\memState_2'$};
			\draw (m1) edge[-latex',decorate,out=30,in=180-30] node[above=3pt] {$\wit_{i, 1}$} (m1');
			\draw (m1') edge[-latex',decorate,out=-180+30,in=-30] node[below=3pt] {$\wit_{i, 2}$} (m1);
			\draw (m1') edge[-latex',decorate,out=-60,in=60] node[right=3pt] {$\cyc_{i,1}'$} (m2');
			\draw (m2') edge[-latex',decorate,out=120,in=-120] node[left=3pt] {$\cyc_{i,2}'$} (m1');
			\draw (m2') edge[-latex',decorate,out=180+30,in=-30] node[below=3pt] {$\wit_{i, 1}'$} (m2);
			\draw (m2) edge[-latex',decorate,out=30,in=180-30] node[above=3pt] {$\wit_{i, 2}'$} (m2');
			\draw (m) edge[-latex',decorate,out=75,in=180+15] (m1);
			\draw (m1) edge[-latex',decorate,out=-105,in=15] (m);
			\draw ($(m)!0.5!(m1)$) node[] () {$\cyc_i$};
			\draw (m) edge[-latex',decorate,out=-15,in=105] (m2);
			\draw (m2) edge[-latex',decorate,out=165,in=-75] (m);
			\draw ($(m)!0.5!(m2)$) node[] () {$\cyc_{i+1}$};
			\draw ($(m1)!0.5!(m1')$) node[] () {$\wit_i$};
			\draw ($(m1')!0.5!(m2')$) node[] () {$\cyc_i'$};
			\draw ($(m2')!0.5!(m2)$) node[] () {$\wit_i'$};
		\end{tikzpicture}
		\caption{Situation in the proof of Lemma~\ref{lem:finHeight}.
				 Competition of $\cyc_i$ and $\cyc_i'$ is witnessed by $\wit_i$, and competition of $\cyc_i'$ and $\cyc_{i+1}$ is witnessed by $\wit_i'$.
				 State $\memState'$ appears somewhere along $\cyc_i'$ and is not represented.}
		\label{fig:finHeight}
	\end{figure}%
	We show that the competition of each contiguous pair in sequence $\cyc_0, \cyc_0', \cyc_1, \cyc_1', \cyc_2, \ldots{}$ has a witness that intersects the winning cycle at $\memState$, and the losing cycle at $\memState'$.
	For all $i\ge 0$, let $\wit_i = \wit_{i,1}\wit_{i,2}$ be a witness that $\cyc_i$ and~$\cyc_i'$ are competing, and $\wit_i' = \wit_{i,1}'\wit_{i,2}'$ be a witness that $\cyc_i'$ and~$\cyc_{i+1}$ are competing.
	Let $i\ge 0$; we depict part of the situation in Figure~\ref{fig:finHeight}, with $\cyc_i' = \cyc_{i,1}'\cyc_{i,2}'$.
	Based on the cycles that we already know, we consider the cycle $\witBisBar_i = \wit_{i,1}\cyc_{i,1}'\wit_{i,1}'\cyc_{i+1}\wit_{i,2}'\cyc_{i,2}'\wit_{i,2}$.
	We have that $\memState, \memState'\in\cycStates{\witBisBar_i}$ since $\cyc_{i+1}$ and $\cyc_{i}'$ are part of $\witBisBar$.
	We show that $\witBisBar_i$ witnesses that $\cyc_i$ and~$\cyc_i'$ are competing:
	\begin{itemize}
		\item $\cycVal{\cyc_i\witBisBar_i} = \win$ since $\cyc_i\witBisBar_i$ can be split into $\cyc_{i,2}'\wit_{i,2}\cyc_i\wit_{i,1}\cyc_{i,1}'$ (winning since $\cyc_i \cycOrdInv \cyc_i'$) and $\wit_{i,1}'\cyc_{i+1}\wit_{i,2}'$ (winning since $\wit_i'$ witnesses a competition involving $\cyc_{i+1}$);
		\item $\cycVal{\cyc_i'\witBisBar_i} = \lose$ since $\cyc_i'\witBisBar_i$ can be split into $\cyc_i'\wit_{i,2}\wit_{i,1}$ (losing since $\wit_i$ witnesses a competition involving $\cyc_i'$) and $\cyc_{i,1}'\wit_{i,1}'\cyc_{i+1}\wit_{i,2}'\cyc_{i,2}'$ (losing since $\cyc_i' \cycOrdInv \cyc_{i+1}$).
		We use Remark~\ref{rem:crossPoint} in order to write ``$\cyc_i'\witBisBar_i$''.
	\end{itemize}

	We can perform a symmetric reasoning to show that the competition of any pair $\cyc_i', \cyc_{i+1}$, $i \ge 0$, is witnessed by a cycle $\witBisBar_i'\in\memCycles$ such that $\memState', \memState\in\cycStates{\witBisBar_i'}$.

	By Lemma~\ref{lem:weakWitness}, $\witBisBar'_i$ is not only a witness that~$\cyc_i'$ and~$\cyc_{i+1}$ are competing, but also that $\cyc_i$ and~$\cyc_i'$ are competing (indeed, $\cycVal{\cyc_i} = \cycVal{\cyc_{i+1}} \neq \cycVal{\cyc_i'}$, $\cyc_{i+1}\cycOrd\cyc_i'$, $\cyc_i' \cycOrd \cyc_i$, $\witBisBar_i'$ witnesses that~$\cyc_{i+1}$ and~$\cyc_i'$ are competing, and $\memState\in\cycStates{\witBisBar_i'}\cap \cycStates{\cyc_i} \cap \cycStates{\cyc_{i+1}}$).

	For $i \ge 0$, we can write $\witBisBar'_i = \witBisBar'_{i,1}\witBisBar'_{i,2}$ with $\witBisBar'_{i,1}\in\memPathsOn{\memState'}{\memState}$ and $\witBisBar'_{i,2}\in\memPathsOn{\memState}{\memState'}$.
	We now consider the infinite sequence
	\[
	\xi = \cyc_0\witBisBar'_{0,2}\cyc_0'\witBisBar'_{0,1}
	\cyc_1\witBisBar'_{1,2}\cyc_1'\witBisBar'_{1,1}
	\cyc_2\ldots
	\]
	Notice that for all $i\ge 0$, $\cyc_i\witBisBar'_{i,2}\cyc_i'\witBisBar'_{i,1}$ is a winning cycle on $\memState$ since $\cyc_i \cycOrdInv \cyc_i'$; hence $\colHatInf(\xi)\in\inverse{\memState}\wc$ by $\memSkel$-cycle-consistency.
	Also, for all $i\ge 0$, $\witBisBar'_{i,2}\cyc_i'\witBisBar'_{i,1}
	\cyc_{i+1}$ is a losing cycle on $\memState$ since $\cyc_i' \cycOrdInv \cyc_{i+1}$; hence $\colHatInf(\xi)\in\inverse{\memState}\comp{\wc}$ by $\memSkel$-prefix-independence and $\memSkel$-cycle-consistency.
	This is a contradiction.

	A proof for infinitely increasing sequences can be done in a symmetric way.
\end{proof}

\begin{lemma} \label{lem:finWidth}
	The width of partial preorder $\cycOrd$ on $\quotient{\memCycles}{\cycEq}$ is finite.
\end{lemma}
\begin{proof}
	We recall that $\memSkel = \memSkelFull$.
	We will show that any two cycles $\cyc_1$ and $\cyc_2$ such that $\cycStates{\cyc_1} = \cycStates{\cyc_2}$ are necessarily comparable for $\cycEq$ or $\cycOrd$.
	This will show that the cardinality of a maximal set of pairwise incomparable (for ${\cycOrd}$) elements in $\quotient{\memCycles}{\cycEq}$ is necessarily bounded by $2^{\card{\memStates}}$, which implies that the width of partial preorder $\cycOrd$ is finite.
	Let $\cyc_1$ and $\cyc_2$ be two cycles such that $\cycStates{\cyc_1} = \cycStates{\cyc_2}$ (we recall that there are infinitely many transitions in $\memSkel$ if $\colors$ is infinite, and that two cycles going through the same states may use different transitions and have a different value).

	If $\cycVal{\cyc_1} \neq \cycVal{\cyc_2}$, then as $\cyc_1$ and $\cyc_2$ share a common state, they are competing --- we have either $\cyc_1 \cycOrd \cyc_2$ or $\cyc_2 \cycOrd \cyc_1$ (depending on the value of $\cyc_1\cyc_2$).

	We now assume that $\cycVal{\cyc_1} = \cycVal{\cyc_2}$; we assume w.l.o.g.\ that $\cyc_1$ and $\cyc_2$ are winning.
	If $\cyc_1$ and $\cyc_2$ are such that $\compar{\cyc_1} = \compar{\cyc_2}$, then they can necessarily be compared with $\cycEq$ or $\cycOrd$; indeed,
	\begin{itemize}
		\item if $\dom(\cyc_1) = \dom(\cyc_2)$, then $\cyc_1 \cycEq \cyc_2$;
		\item if $\dom(\cyc_1) \neq \dom(\cyc_2)$, then there is $i\in\{1, 2\}$ and a losing cycle $\cyc'$ in $\dom(\cyc_i)$ that is competing with $\cyc_{3 - i}$ but that is not an element of $\dom(\cyc_{3-i})$.
		Therefore, $\cyc_{3-i} \cycOrd \cyc' \cycOrd \cyc_i$, which means that $\cyc_{3-i} \cycOrd \cyc_i$.
	\end{itemize}
	It is left to deal with the case $\compar{\cyc_1} \neq \compar{\cyc_2}$.
	W.l.o.g., let $\cyc'$ be in $\compar{\cyc_1} \setminus \compar{\cyc_2}$.
	There are two cases to discuss: whether $\cyc_1 \cycOrd \cyc'$ or $\cyc' \cycOrd \cyc_1$.
	\begin{itemize}
		\item Assume $\cyc_1 \cycOrd \cyc'$.
		By Lemma~\ref{lem:transWeak}, as $\cycVal{\cyc_1} = \cycVal{\cyc_2}$ and $\cycStates{\cyc_1} \cap \cycStates{\cyc_2} \neq \emptyset$, $\cyc'$ is also competing with $\cyc_2$, which is a contradiction.
		\item Assume $\cyc' \cycOrd \cyc_1$.
		Let $\wit$ be a witness that $\cyc_1$ and $\cyc'$ are competing.
		We therefore have that $\cyc_1\wit$ is winning and $\cyc'\wit$ is losing.
		As $\cyc_2$ is not competing with $\cyc'$, $\wit$ cannot be a witness that~$\cyc_2$ and~$\cyc'$ are competing.
		Since $\cycStates{\cyc_1} = \cycStates{\cyc_2}$ has a non-empty intersection with $\cycStates{\wit}$, the only possibility for that to happen is that $\cyc_2\wit$ is losing (all other conditions are satisfied).
		This means that $\wit$ must itself be a losing cycle.
		But then, observe that $\wit$ is competing both with $\cyc_1$ and $\cyc_2$ (as $\wit$ has a common state with and a different value than $\cyc_1$ and $\cyc_2$) and $\cyc_2 \cycOrd \wit \cycOrd \cyc_1$ (as $\cyc_2\wit$ is losing and $\cyc_1\wit$ is winning).
		This implies that~$\cyc_2 \cycOrd \cyc_1$.
		% \qedhere
	\end{itemize}
\end{proof}

Lemmas~\ref{lem:finHeight} and~\ref{lem:finWidth} imply together that $\cycEq$ has finite index, and thus that ${\cycOrd}$ (partially) orders only finitely many classes of cycles in $\quotient{\memCycles}{\cycEq}$.
Therefore, for some $n\in\IN$, there exists a function $\priCyc\colon \quotient{\memCycles}{\cycEq} \to \{0, \ldots, n\}$ that is a monotonic function (assuming $\quotient{\memCycles}{\cycEq}$ is preordered with $\cycOrd$ and $\{0,\ldots,n\}$ is ordered with the usual order on $\IN$); such a function is sometimes called a \emph{linear extension of the partial order}.
We extend it to a function $\priCyc\colon \memCycles \to \{0, \ldots, n\}$ such that $\priCyc(\cyc) = \priCyc(\eqClassCyc{\cyc})$.
Moreover, we assume w.l.o.g.\ that $\cycVal{\cyc} = \win$ if and only if $\priCyc(\cyc)$ is even (this might require to increase $n$, but it is always possible).

We fix $n$ and any such function $\priCyc$ for the rest of the proof.

\subparagraph*{Parity automaton on top of \texorpdfstring{$\memSkel$}{M}.}
At this point, it would already be possible to describe words of $\wc$ in terms of the cycles of $\memSkel$ that they visit through function $\memWordSolo$ (there may be multiple such decompositions) and their values by $\priCyc$, but that does not directly correspond to a classical acceptance condition for automata on infinite words.
We can actually obtain something more satisfying: we show that we can assign priorities to \emph{transitions} of $\memSkel$ to recognize $\wc$, in a way that corresponds to a parity acceptance condition on transitions.
We transfer function $\priCyc$ to transitions of $\memSkel$: for $(\memState, \clr) \in \memStates\times\colors$, we define
\begin{align} \label{eq:priDef}
\pri(\memState, \clr) = \min \{\priCyc(\cyc) \mid \cyc\in\memCycles, (\memState, \clr) \in \cyc\}.
\end{align}
We now have a well-defined function assigning priorities to every transition of $\memSkel$.

\begin{example}
	We illustrate our definitions for $\priCyc$ and $\pri$.
	We again consider the example from Figure~\ref{fig:parityAtmtn} (for which, unlike $\wc$, we already know that it describes an $\omega$-regular language).
	For the sake of the example, let us ignore the already-defined priority function $\pri$ of this parity automaton.
	We show that we can recover priorities defining the same language starting from our preorder $\cycOrd$ and our definitions for $\priCyc$ and $\pri$.
	There were four equivalence classes of $\cycEq$, ordered as follows: $\eqClassCyc{(\memState_1, b)} \cycOrd \eqClassCyc{(\memState_1, a)(\memState_2, a)} \cycOrd \eqClassCyc{(\memState_1, c)}$ and $\eqClassCyc{(\memState_2, b)} \cycOrd \eqClassCyc{(\memState_1, c)}$.
	Function $\priCyc$ must be any function that respects the order given by the diagram and that assigns even priorities to winning classes of cycles, and odd priorities to losing classes.
	One such possible choice is $\priCyc(\eqClassCyc{(\memState_1, c)}) = 5$, $\priCyc(\eqClassCyc{(\memState_1, a)(\memState_2, a)}) = 2$, $\priCyc(\eqClassCyc{(\memState_2, b)}) = 4$, and $\priCyc(\eqClassCyc{(\memState_1, b)}) = 1$.
	From this choice of function $\priCyc$, our definition~\eqref{eq:priDef} of function $\pri$ entails $\pri((\memState_1, c)) = 5$, $\pri((\memState_1, a)) = \pri((\memState_2, a)) = 2$, $\pri((\memState_2, b)) = \pri((\memState_2, c)) = 4$, and $\pri((\memState_1, b)) = 1$.
	This choice of priorities defines the same language as the original parity automaton.
	% \lipicsEnd
\end{example}

We will prove that $(\memSkel, \pri)$ recognizes the language $\wc$.
In our proof, we will need to relate the cycles dominated by a cycle $\cyc$ and the ones dominated by cycles in a ``decomposition'' of $\cyc$, i.e., cycles that can be obtained from iteratively removing cycles from $\cyc$.
We formally define this notion and prove two related results.

\newpage
\begin{definition}[Cycle decomposition] \label{def:cycDec}
	Let $\cyc = (\memState_1, \clr_1)\ldots(\memState_k, \clr_k)\in\memCycles$, and $\cyc_1, \ldots, \cyc_l\in\memCycles$.
	We say that $(\cyc_1, \ldots, \cyc_l)$ is a \emph{cycle decomposition of $\cyc$} if
	\begin{itemize}
		\item either $l = 1$ and $\cyc = \cyc_1$,
		\item or $l > 1$ and there exist $i, i' \in \{1, \ldots, k\}$, $i \le i'$, such that cycle $\cyc_1 = (\memState_i, \clr_i)\ldots(\memState_{i'}, \clr_{i'})$, and $(\cyc_2, \ldots, \cyc_l)$ is a cycle decomposition of the smaller cycle \[(\memState_1, \clr_1)\ldots(\memState_{i-1}, \clr_{i-1})(\memState_{i'+1}, \clr_{i'+1})\ldots(\memState_k, \clr_k).\]
	\end{itemize}
\end{definition}

\begin{lemma} \label{lem:domSubcycle}
	Let $\cyc, \cyc_1, \cyc_2, \cyc'\in\memCycles$ be cycles such that $\cyc = \cyc_1\cyc_2$.
	If $\cyc' \cycOrd \cyc_1$, then $\cyc' \cycOrd \cyc$.
\end{lemma}
\begin{proof}
	We assume $\cyc'\cycOrd \cyc_1$.

	If $\cycVal{\cyc_1} \neq \cycVal{\cyc}$, then $\cyc_1 \cycOrd \cyc$ --- indeed, they share at least one state and $\cyc_1\cyc = (\cyc_1)^2\cyc_2$ has the same value as $\cyc_1\cyc_2 = \cyc$ by Lemma~\ref{lem:indRep}.
	Therefore, by transitivity of $\cycOrd$ (Lemma~\ref{lem:strictPreorder}), $\cyc'\cycOrd \cyc$.

	We now assume $\cycVal{\cyc_1} = \cycVal{\cyc}$ and $\cycVal{\cyc'} \neq \cycVal{\cyc_1}$.
	Let $\witBar$ be a witness that $\cyc'$ and $\cyc_1$ are competing.
	We prove that $\witBar$ also witnesses that $\cyc'$ and $\cyc$ are competing: to do so, it is left to show that $\cycBis = \witBar\cyc$ has the same value as $\cyc$.
	We have that $\cycBis$ can be written as $\witBar\cyc_{1,1}\cyc_2\cyc_{1,2}$ for some paths $\cyc_{1,1}$ and $\cyc_{1,2}$ such that $\cyc_1 = \cyc_{1, 1}\cyc_{1, 2}$.
	Cycle $\cycBis$ has the same value as $\witBar(\cyc_{1,1}\cyc_{1,2})\cyc_{1,1}\cyc_2\cyc_{1,2}$ by Lemma~\ref{lem:indRepStrong}.
	This last cycle can be split into $\witBar\cyc_1$ and $\cyc$, which both have the same value as~$\cyc$.
	Therefore $\witBar$ is also a witness that $\cyc'$ and $\cyc$ are competing.
	We can show with a very similar argument that $\cyc'\witBar_1\cyc\witBar_2$ also has the same value as $\cyc$, hence~$\cyc'\cycOrd \cyc$.

	If $\cycVal{\cyc_1} = \cycVal{\cyc}$ and $\cycVal{\cyc'} = \cycVal{\cyc_1}$, then there exists $\cyc''$ with $\cycVal{\cyc''} \neq \cycVal{\cyc_1}$ such that $\cyc' \in \dom(\cyc'')$ and $\cyc'' \in \dom(\cyc_1)$, so $\cyc'\cycOrd\cyc''\cycOrd\cyc_1$.
	By the previous case, $\cyc'' \cycOrd \cyc$, and by transitivity, $\cyc' \cycOrd \cyc$.
\end{proof}

\begin{lemma} \label{lem:decompCyc}
	Let $\cyc$ be a cycle of $\memSkel$ and $(\cyc_1, \ldots, \cyc_l)$ be a cycle decomposition of $\cyc$.
	For all $i\in\{1,\ldots,l\}$, for all $\cyc'\in\memCycles$, if $\cyc'\cycOrd \cyc_i$, then $\cyc'\cycOrd \cyc$.
\end{lemma}
\begin{proof}
	We proceed by induction on $l$.
	If $l = 1$, then the statement is trivial as $\cyc = \cyc_1$.
	For $l > 1$, we now assume that the property holds for $l-1$, and we show that it also holds for $l$.
	Up to a shift of $\cyc$ and of the cycle decomposition, we assume that $\cyc$ is equal to $\cyc_1\cycBis$, where $(\cyc_2,\ldots,\cyc_l)$ is a cycle decomposition of $\cycBis$.

	Let $\cyc'\in\memCycles$ be such that $\cyc' \cycOrd \cyc_i$ for some $i\in\{1,\ldots,l\}$.
	This implies that $\cyc'\cycOrd\cyc_1$ if $i = 1$ or, using the induction hypothesis, that $\cyc'\cycOrd\cycBis$.
	In any case, by Lemma~\ref{lem:domSubcycle}, we immediately have that $\cyc'\cycOrd \cyc$.
\end{proof}

We can now prove that $\wc$ is recognized by the parity automaton $(\memSkel, \pri)$.
We do this in the next two results.
First, we show that winning cycles of $\memSkel$ are exactly the ones that have an even maximal priority given by $\pri$.
It is then straightforward to conclude that infinite words in $\wc$ are exactly the ones whose maximal infinitely visited priority is even.

\begin{lemma} \label{lem:winCyc}
	Let $\cyc = (\memState_1, \clr_1)\ldots(\memState_k, \clr_k)\in\memCycles$.
	Then, $\cyc$ is winning if and only if\linebreak $\max_{1\le i\le k} \pri(\memState_i, \clr_i)$ is even.
\end{lemma}
\begin{proof}
	For conciseness, let $\maxPri = \max_{1\le i\le k} \pri(\memState_i, \clr_i)$ and $\edge_i = (\memState_i, \clr_i)$.
	By definition of function~$\pri$, for all $i\in\{1,\ldots,k\}$, $\pri(\edge_i) \le \priCyc(\cyc)$.
	Hence, $\maxPri \le \priCyc(\cyc)$.
	We want to show that $\cyc$ is winning if and only if $\maxPri$ is even.
	By contradiction, we assume that we do not have this equivalence.
	We assume w.l.o.g.\ that $\cyc$ is losing and that $\maxPri$ is even; we could obtain in a symmetric way a contradiction for $\cyc$ winning and $\maxPri$ odd.

	As $\cyc$ is losing, we have that $\priCyc(\cyc)$ is odd --- as $\maxPri$ is even, $\maxPri < \priCyc(\cyc)$.
	We assume (up to a shift of the transitions) that $\maxPri = \pri(\edge_1)$.
	Since $\maxPri < \priCyc(\cyc)$, there exists, for all $i\in\{1,\ldots,k\}$, a cycle $\cyc_i \neq \cyc$ such that $\edge_i\in\cyc_i$ and $\pri(\edge_i) = \priCyc(\cyc_i)$.
	We assume $\cyc_i = \edge_i\memPth_i$ for a suitable path $\memPth_i$.
	The situation is represented in Figure~\ref{fig:winCyc}.

	The rest of the proof consists in exhibiting two cycles, building on the ones we know, showing that one of them is winning and one of them is losing, and finally showing that they must have the same value, which provides a contradiction.

	\begin{figure}[t]
		\centering
		\begin{tikzpicture}[every node/.style={font=\small,inner sep=1pt}]
			\draw (0,0) node[diamant] (m1) {$\memState_1$};
			\draw ($(m1)+(2.5,0)$) node[diamant] (m2) {$\memState_2$};
			\draw ($(m2)+(1,-2.5)$) node[diamant] (m3) {$\memState_3$};
			\draw ($(m3)+(-4.5,0)$) node[diamant] (mk) {$\memState_k$};
			\draw (m1) edge[-latex'] node[below=3pt] {$\edge_1$} (m2);
			\draw (m2) edge[-latex',decorate,out=150,in=30] node[above=3pt] {$\memPth_1$} (m1);
			\draw (m2) edge[-latex'] node[left=3pt] {$\edge_2$} (m3);
			\draw (m3) edge[-latex',decorate,out=60,in=-30] node[right=4pt] {$\memPth_2$} (m2);
			\draw (m3) edge[dashed] (mk);
			\draw (mk) edge[-latex'] node[right=3pt] {$\edge_k$} (m1);
			\draw (m1) edge[-latex',decorate,out=210,in=120] node[left=4pt] {$\memPth_k$} (mk);
		\end{tikzpicture}
		\caption{Situation in the proof of Lemma~\ref{lem:winCyc}, with $\cyc = \edge_1\ldots\edge_k$.}
		\label{fig:winCyc}
	\end{figure}

	We will first consider cycle $\edge_1\ldots\edge_k\memPth_k\ldots\memPth_1$ (on $\memState_1$).
	We will prove by induction that it is winning.
	First, $\edge_1\memPth_1$ is winning since $\priCyc(\edge_1\memPth_1) = \maxPri$ is even.
	Assume now that for $1 < l \le k$, $\edge_1\ldots\edge_{l-1}\memPth_{l-1}\ldots\memPth_1$ is winning.
	We show that $\edge_1\ldots\edge_{l-1}(\edge_l\memPth_l)\memPth_{l-1}\ldots\memPth_1$ is winning.
	\begin{itemize}
		\item If $\edge_l\memPth_l$ is a winning cycle, it follows from $\memSkel$-cycle-consistency.
		\item If $\edge_l\memPth_l$ is a losing cycle, we distinguish two cases.
		If $\edge_2\ldots\edge_{l-1}(\edge_l\memPth_l)\memPth_{l-1}\ldots\memPth_2$ is winning, then so is $\edge_1\ldots\edge_{l-1}(\edge_l\memPth_l)\memPth_{l-1}\ldots\memPth_1$ because we just concatenate the winning cycle $\memPth_1\edge_1$\linebreak to a winning cycle ($\memSkel$-cycle-consistency).
		If $\edge_2\ldots\edge_{l-1}(\edge_l\memPth_l)\memPth_{l-1}\ldots\memPth_2$ is losing, then\linebreak $\edge_2\ldots\edge_{l-1}\memPth_{l-1}\ldots\memPth_2$ witnesses that $\edge_1\memPth_1$ and $\edge_l\memPth_l$ are competing.
		Since $\priCyc(\edge_l\memPth_l)$ is odd, and $\priCyc(\edge_1\memPth_1)$ is even and is equal to the maximum of $i\mapsto \priCyc(\edge_i\memPth_i)$, we have that $\priCyc(\edge_l\memPth_l) < \priCyc(\edge_1\memPth_1)$.
		Since $\priCyc$ is monotonic and $\edge_1\memPth_1$ and $\edge_l\memPth_l$ are competing, this implies $\edge_l\memPth_l \cycOrd \edge_1\memPth_1$.
		Thus $\edge_1\ldots\edge_{l-1}(\edge_l\memPth_l)\memPth_{l-1}\ldots\memPth_1$ is winning.
	\end{itemize}

	We now consider the cycle $\edge_1(\memPth_1\edge_1)\ldots\edge_k(\memPth_k\edge_k)$ (on $\memState_1$).
	We show by induction that it is losing.
	We start from $\cyc$, which is losing by hypothesis, and we add cycles $\memPth_i\edge_i$ one by one.
	We denote $\cyc^{(l)} = \edge_1(\memPth_1\edge_1)\ldots\edge_l(\memPth_l\edge_l)\edge_{l+1}\ldots\edge_k$.
	Assume that $\cyc^{(l-1)}$ is losing for $1 < l \le k$.
	We want to show that $\cyc^{(l)}$ is also losing.

	\begin{itemize}
		\item If $\memPth_l\edge_l$ is a losing cycle, it follows from $\memSkel$-cycle-consistency.
		\item If $\memPth_l\edge_l$ is a winning cycle, then as $\priCyc(\memPth_l\edge_l) \le \maxPri < \priCyc(\cyc)$ and $\memPth_l\edge_l$ is competing with $\cyc$ (they share common states), we have $\memPth_l\edge_l \cycOrd \cyc$.
		Notice that $(\memPth_1\edge_1, \ldots, \memPth_{l-1}\edge_{l-1}, \cyc)$ is a cycle decomposition of $\cyc^{(l-1)}$ as in Definition~\ref{def:cycDec}.
		Thus by Lemma~\ref{lem:decompCyc}, as $\memPth_l\edge_l \cycOrd \cyc$, we also have $\memPth_l\edge_l \cycOrd \cyc^{(l-1)}$.
		We conclude that $\cyc^{(l)}$ is also losing.
	\end{itemize}

	We have now considered two cycles on $\memState_1$: the winning $\edge_1\ldots\edge_k\memPth_k\ldots\memPth_1$ and the losing $\edge_1(\memPth_1\edge_1)\ldots\edge_k(\memPth_k\edge_k)$.
	We show that it is possible to transform the latter into the former using only value-preserving transformations (given by Lemmas~\ref{lem:indRep} and~\ref{lem:indOrdThree}), which provides the desired contradiction.

	We show inductively that for all $l\in\{1,\ldots,k\}$, cycle $\edge_1(\memPth_1\edge_1)\ldots\edge_k(\memPth_k\edge_k)$ can be transformed into
	\[
	\cycBis^{(l)} = (\edge_1\ldots\edge_l\memPth_l\ldots\memPth_1\edge_1\ldots\edge_l)\edge_{l+1}(\memPth_{l+1}\edge_{l+1})\ldots\edge_k(\memPth_k\edge_k)
	\]
	using value-preserving transformations.
	Notice that $\edge_1(\memPth_1\edge_1)\ldots\edge_k(\memPth_k\edge_k)$ is equal to $\cycBis^{(1)}$, which deals with the base case of the induction.
	Now assume that $\edge_1(\memPth_1\edge_1)\ldots\edge_k(\memPth_k\edge_k)$ has the same value as $\cycBis^{(l-1)}$ for $1 < l \le k$.
	In the expression of $\cycBis^{(l-1)}$, notice that $\memPth_{l-1}\ldots\memPth_1\edge_1\ldots\edge_{l-1}$ and $\edge_{l}\memPth_{l}$ are two consecutive cycles on $\memState_{l}$.
	By Lemma~\ref{lem:indOrdThree}, they can thus be swapped while keeping a cycle with the same value.
	Notice that this gives exactly the cycle $\cycBis^{(l)}$.

	We obtain that $\edge_1(\memPth_1\edge_1)\ldots\edge_k(\memPth_k\edge_k)$ has the same value as $\cycBis^{(k)} = \edge_1\ldots\edge_k\memPth_k\ldots\memPth_1\edge_1\ldots\edge_k$, which has the same value as $\edge_1\ldots\edge_k\memPth_k\ldots\memPth_1$ by Lemma~\ref{lem:indRep}.
\end{proof}

\begin{proposition}
	Let $\word = \clr_1\clr_2\ldots\in\colors^\omega$ with $\memWord{\word} = (\memState_1, \clr_1)(\memState_2, \clr_2)\ldots\in (\memStates\times\colors)^\omega$.
	Then,
	\[\word\in\wc\ \textnormal{if and only if}\ \limsup_{i\ge 1} \pri(\memState_i, \clr_i)\ \textnormal{is even}.
	\]
\end{proposition}
\begin{proof}
	Let $\maxPri = \limsup_{i\ge 1} \pri(\memState_i, \clr_i)$.
	Let $j \ge 1$ be an index such that for all $i\ge j$, $\pri(\memState_i, \clr_i) \le \maxPri$.
	Let $I^* = \{i \ge j \mid \pri(\memState_i, \clr_i) = \maxPri\}$ be the infinite set of indices of transitions with priority $\maxPri$ occurring after index $j$.
	We write $i_1, i_2,\ldots{}$ for the elements of $I^*$ in order.
	Let $\memState^*$ be a state appearing infinitely often in $\{\memState_i\mid i\in I^*\}$ (such a state exists necessarily as the state space of $\memSkel$ is finite).
	This implies that $\memWord{\word}$ can be written as the concatenation of a finite prefix $(\memState_1, \clr_1)\ldots(\memState_{i_1 - 1}, \clr_{i_1 - 1})$ and infinitely many cycles $\cyc_k = (\memState_{i_k}, \clr_{i_k})\ldots(\memState_{i_{k+1} - 1}, \clr_{i_{k+1} - 1})$ with $\memState_{i_k} = \memState^*$ and $\pri(\memState_{i_k}, \clr_{i_k}) = \maxPri$, for $k\ge 1$.

	For all $k\ge 1$, we have that $\max_{i_k \le i < i_{k+1}} \pri(\memState_i, \clr_i) = \maxPri$ (it is $\le \maxPri$ as $i_k \ge j$, and it is $\ge \maxPri$ as $\pri(\memState_{i_k}, \clr_{i_k}) = \maxPri$).
	By Lemma~\ref{lem:winCyc}, we conclude that cycles $\cyc_k$ are all cycles on $\memState^*$ that have the same value: they are winning if $\maxPri$ is even, and losing if $\maxPri$ is odd.
	By $\memSkel$-prefix-independence and $\memSkel$-cycle-consistency, $\word$ is in $\wc$ if $\maxPri$ is even, and $\word$ is in $\comp{\wc}$ if $\maxPri$ is odd.
\end{proof}

We have therefore reached our goal for this section.

\begin{corollary}[Second item of Theorem~\ref{thm}]
	\secondItem
\end{corollary}

\begin{remark}
	As discussed in Remark~\ref{rem:finColors}, our proof shows as a by-product that even if $\colors$ is infinite, many colors can be assumed to be equal (w.r.t.\ $\wc$) --- there are only finitely many classes of truly different colors.
	% \lipicsEnd
\end{remark}

\section{Applications} \label{sec:application}
We provide a thorough application of our concepts to a discounted-sum winning condition.
We then discuss more briefly mean-payoff and total-payoff winning conditions.

\subsection{Discounted sum} \label{sec:DS}
We apply our results to a \emph{discounted-sum} condition in order to illustrate our notions.
A specificity of this example is that its $\omega$-regularity depends on some chosen parameters --- we use our results to characterize the parameters for which it is $\omega$-regular or, equivalently (Theorem~\ref{thm:regular}), chromatic-finite-memory determined.
The $\omega$-regularity of discounted-sum conditions has also been studied in~\cite{CDH09,BCV18} with different techniques and goals.

Let $\colors \subseteq \IQ$ be non-empty and bounded.
For $\disc\in\intervaloo{0, 1}\cap\IQ$, we define the \emph{discounted-sum function} $\DSSolo{\disc}\colon \colors^\omega \to \IR$ such that for $\word = \clr_1\clr_2\ldots\in\colors^\omega$,
\[
	\DS{\disc}{\word} = \sum_{i = 1}^{\infty} \disc^{i-1} \cdot \clr_i.
\]
This function is always well-defined for a bounded $\colors$, and takes values in $\intervalcc{\frac{\inf \colors}{1-\disc}, \frac{\sup \colors}{1 - \disc}}$.

We define the winning condition $\DSObj = \{\word\in\colors^\omega \mid \DS{\disc}{\word} \ge 0\}$ as the set of infinite words whose discounted sum is non-negative, and let $\prefEq$ be its right congruence.
We will analyze cycle-consistency and prefix-independence of $\DSObj$ to conclude under which conditions (on~$\colors$ and $\disc$) it is chromatic-finite-memory determined.
First, we discuss a few properties of the discounted-sum function.

\subparagraph*{Basic properties.}
We extend function $\DSSolo{\disc}$ to finite words in a natural way: for $\word\in\colors^*$, we define $\DS{\disc}{\word} = \DS{\disc}{\word0^\omega}$.
For $\word\in\colors^*$, we define $\card{\word}$ as the length of $\word$ (so $\word\in\colors^{\card{\word}}$).
First, we notice that for $\word\in\colors^*$ and $\word'\in\colors^\omega$, we have
\[
	\DS{\disc}{\word\word'} = \DS{\disc}{\word} + \disc^{\card{\word}}\DS{\disc}{\word'}.
\]
Therefore,
\[
	\word\word'\in\DSObj \Longleftrightarrow \frac{\DS{\disc}{\word}}{\disc^{\card{\word}}} \ge -\DS{\disc}{\word'}.
\]
This provides a characterization of the winning continuations of a finite word $\word\in\colors^*$ by comparing their discounted sum to the value $\frac{\DS{\disc}{\word}}{\disc^{\card{\word}}}$.

This leads us to define the \emph{gap} of a finite word $\word\in\colors^*$, following ideas in~\cite{BHO15}, as
\[
\gap{\word} =
\begin{cases*}
	\top &\text{if $\frac{\DS{\disc}{\word}}{\disc^{\card{\word}}} \ge -\frac{\inf\colors}{1-\disc}$}, \\
	\bot &\text{if $\frac{\DS{\disc}{\word}}{\disc^{\card{\word}}} < -\frac{\sup\colors}{1-\disc}$}, \\
	\frac{\DS{\disc}{\word}}{\disc^{\card{\word}}} &\text{otherwise}.
\end{cases*}
\]
Intuitively, the gap of a finite word $\word\in\colors^*$ represents how far it is from going back to $0$: if $\word'\in\colors^\omega$ is such that $\DS{\disc}{\word'} = -\gap{\word}$, then $\DS{\disc}{\word\word'} = 0$.
We can see that for all words $\word\in\colors^*$, if $\gap{\word} = \top$, then all continuations are winning (i.e., $\inverse{\word}\wc = \colors^\omega$) as it is not possible to find an infinite word with a discounted sum less than $\frac{\inf\colors}{1-\disc}$.
Similarly, if $\gap{\word} = \bot$, then all continuations are losing (i.e., $\inverse{\word}\wc = \emptyset$).

\subparagraph*{Cycle-consistency.}
We can show that $\DSObj$ is always $\memSkelTriv$-cycle-consistent.

\begin{restatable}{proposition}{DSCycCons} \label{prop:DSCycCons}
	For all bounded $\colors \subseteq \IQ$, $\disc\in\intervaloo{0, 1}\cap\IQ$, winning condition $\DSObj$ is $\memSkelTriv$-cycle-consistent.
\end{restatable}%
The proof of this result is elementary and is provided in Appendix~\ref{app:DS}.

\subparagraph*{Prefix-independence.}
If $\colors = \intervalcc{-k, k} \cap \IQ$ for some $k\in\IN\setminus \{0\}$, winning condition $\DSObj$ is not $\memSkel$-prefix-independent for any $\memSkel$, as $\prefEq$ has infinite index.
Indeed, for $i\ge 1$ and $\word_i = \frac{1}{i} \in \colors^*$, we have $\word_1 \strictInvPrefOrd \word_2 \strictInvPrefOrd \ldots{}$ --- we can see how to use this to exhibit an arena in which $\Pone$ can win but needs infinite memory to do so in Figure~\ref{fig:DSInfBasic}.%
\begin{figure}[t]
	\centering
	\begin{tikzpicture}[every node/.style={font=\small,inner sep=1pt}]
		\draw (0,0) node[carre] (s1) {$\s_1$};
		\draw ($(s1)+(2.5,0)$) node[rond] (s2) {$\s_2$};
		\draw ($(s2)+(2.5,0)$) node[carre] (s3) {$\s_3$};

		\draw (s1) edge[-latex',out=60,in=180-60,distance=1.1cm] node[above=2pt] {$1$} (s2);
		\draw (s1) edge[-latex',out=15,in=180-15] node[above=2pt] {$\frac{1}{2}$} (s2);
		\draw (s1) edge[draw=none,out=-15,in=180+15] node[] {$\vdots$} (s2);

		\draw (s2) edge[-latex',out=60,in=180-60,distance=1.1cm] node[above=2pt] {$\frac{-1}{\disc}$} (s3);
		\draw (s2) edge[-latex',out=15,in=180-15] node[above=2pt] {$\frac{-1}{2\disc}$} (s3);
		\draw (s2) edge[draw=none,out=-15,in=180+15] node[] {$\vdots$} (s3);

		\draw (s3) edge[-latex',out=30,in=-30,distance=0.8cm] node[right=3pt] {$0$} (s3);
	\end{tikzpicture}
	\caption{Arena with infinitely many edges in which $\Pone$ needs infinite memory to win for condition $\DSObj$ from $\s_1$ for any $\disc\in\intervaloo{0,1}\cap\IQ$, with $\colors = \intervalcc{-k, k} \cap \IQ$ for $k$ sufficiently large.}
	\label{fig:DSInfBasic}
\end{figure}%

For finite $\colors\subseteq \IZ$, the picture is more complicated; for $\colors = \intervalcc{-k, k}\cap\IZ$ for some $k\in\IN$, we characterize when $\DSObj$ is $\memSkel$-prefix-independent for some finite skeleton $\memSkel$.
We give an intuition of the two situations in which that happens: $(i)$ if $\colors$ is too small, then the first non-zero color seen determines the outcome of the game, as it is not possible to compensate this color to change the sign of the discounted sum; $(ii)$ if $\disc = \frac{1}{n}$ for some integer $n \ge 1$, then the $\gapSolo$ function actually takes only finitely many values, which is not the case for a different $\disc$.

\begin{proposition} \label{prop:DSPrefInd}
	Let $\disc\in\intervaloo{0, 1}\cap\IQ$, $k \in \IN$, and $\colors = \intervalcc{-k, k}\cap\IZ$.
	Then, the right congruence $\prefEq$ of $\DSObj$ has finite index if and only if $k < \frac{1}{\disc} - 1$ or $\disc$ is equal to $\frac{1}{n}$ for some integer $n \ge 1$.
\end{proposition}

\begin{proof}
	We define $\maxDS = \frac{\sup \colors}{1-\disc} = \frac{k}{1 - \disc}$ and $\minDS = -\frac{k}{1 - \disc}$, as respectively the maximal and minimal discounted-sum value achievable with colors in $\colors$.

	The key property that we will show is that gaps characterize equivalence classes of prefixes: for $\word_1, \word_2\in\colors^*$,
	\begin{equation} \label{eq:gapChar}
	\word_1 \prefEq \word_2 \Longleftrightarrow
	\gap{\word_1} = \gap{\word_2}.
	\end{equation}
	Once this is proven, it is left to determine the number of different gap values, which will correspond to the index of $\prefEq$.
	The right-to-left implication of~\eqref{eq:gapChar} is clear: if the gaps are $\top$, all the continuations are winning; if the gaps are $\bot$, all the continuations are losing; else, for any continuation, the final discounted-sum values will have the same sign.
	We prove the left-to-right implication for each case of the disjunction from the statement and discuss the number of gap values.

	We first assume $k < \frac{1}{\disc} - 1$.
	The case $k = 0$ is trivial (as all words are winning) --- we assume $k \ge 1$.
	The inequality $k < \frac{1}{\disc} - 1$ is equivalent to $\frac{1}{\disc} > \frac{k}{1 - \disc}$.
	In this case, there are only three possible gaps:
	\begin{itemize}
		\item for $\word\in 0^*$, $\gap{\word} = 0$.
		\item for $\word\in 0^*\clr$ with $\clr \ge 1$, then $\frac{\DS{\disc}{\word}}{\disc^{\card{\word}}} = \frac{c}{\disc} \ge \frac{1}{\disc} > \frac{k}{1 - \disc} = -\minDS$ --- so for any word $\word\in 0^*\clr\colors^*$, $\gap{\word} = \top$.
		\item for $\word\in 0^*\clr\colors^*$ with $\clr \le -1$, symmetrically, $\gap{\word} = \bot$.
	\end{itemize}
	These three possible gaps clearly correspond to different equivalence classes of the right congruence $\prefEq$, so there are three such equivalence classes.
	Hence the minimal-state automaton~$\minStateAtmtn$ has three states $\eqClass{\emptyWord}$, $\eqClass{1}$, and $\eqClass{-1}$.

	We now assume that $k \ge \ceil{\frac{1}{\disc} - 1}$.
	The left-to-right implication of~\eqref{eq:gapChar} is clear in the cases in which all, or none, of the continuations are winning.
	The difficult case is when both $\word_1$ and~$\word_2$ have a rational gap.
	We show that if their gaps are different rational numbers, then they have different winning continuations.
	We assume w.l.o.g.\ that $\gap{\word_2} < \gap{\word_1}$.
	We show that there is an infinite continuation that has a discounted sum exactly equal to $-\gap{\word_1}$: this infinite continuation is winning after $\word_1$ but losing after $\word_2$.

	Showing that there exists $\word\in\colors^\omega$ such that $\DS{\disc}{\word} = - \gap{\word_1}$ amounts to showing that there is a representation of $- \gap{\word_1}$ in the (rational but not necessarily integral) base $\frac{1}{\disc}$ with digits in $\colors$, with one digit before the decimal point.
	We can adapt the well-known \emph{greedy expansion}~\cite{Ren57} to our context to show this (details in Appendix~\ref{app:DS}, Proposition~\ref{prop:numbRep}).

	It is left to show that there are finitely many gap values if and only if $\disc$ equals $\frac{1}{n}$ for some integer $n\ge 1$.
	One implication is clear: if $\disc = \frac{1}{n}$ for some integer $n \ge 1$, then there are finitely many possible gaps as gaps are then always integers between $\minDS$ and $\maxDS$, $\top$, or $\bot$.
	We illustrate this implication by depicting the minimal-state automaton of $\prefEq$ for $\disc = \frac{1}{2}$ and $k = 2$ in Figure~\ref{fig:minStateDS12}.%
	\begin{figure}[t]
		\centering
		\begin{tikzpicture}[every node/.style={font=\small,inner sep=1pt}]
			\draw (0,0) node[diamant] (s0) {$0$};
			\draw ($(s0)+(2.5,0)$) node[diamant] (s2) {$2$};
			\draw ($(s2)+(2.5,0)$) node[diamant] (top) {$\top$};
			\draw ($(s0)-(2.5,0)$) node[diamant] (s-2) {$-2$};
			\draw ($(s-2)-(2.5,0)$) node[diamant] (s-4) {$-4$};
			\draw ($(s-4)-(2.5,0)$) node[diamant] (bot) {$\bot$};

			\draw ($(s0)+(0,.8)$) edge[-latex'] (s0);
			\draw (s0) edge[-latex',out=270-30,in=270+30,distance=0.8cm] node[below=2pt] {$0$} (s0);
			\draw (s0) edge[-latex',out=15,in=180-15] node[above=2pt] {$1$} (s2);
			\draw (s0) edge[-latex',out=180-15,in=15] node[above=2pt] {$-1$} (s-2);
			\draw (s0) edge[-latex',out=45,in=180-45] node[above=2pt] {$2$} (top);
			\draw (s0) edge[-latex',out=180-45,in=45] node[above=2pt] {$-2$} (s-4);

			\draw (s2) edge[-latex',out=270-30,in=270+30,distance=0.8cm] node[below=2pt] {$-1$} (s2);
			\draw (s2) edge[-latex'] node[above=2pt] {$0,1,2$} (top);
			\draw (s2) edge[-latex',out=180+15,in=-15] node[below=2pt] {$-2$} (s0);

			\draw (s-2) edge[-latex',out=270-30,in=270+30,distance=0.8cm] node[below=2pt] {$1$} (s-2);
			\draw (s-2) edge[-latex',out=180-45,in=45] node[above=2pt] {$-2,-1$} (bot);
			\draw (s-2) edge[-latex',out=-15,in=180+15] node[below=2pt] {$2$} (s0);
			\draw (s-2) edge[-latex'] node[above=2pt] {$0$} (s-4);

			\draw (s-4) edge[-latex',out=270-30,in=270+30,distance=0.8cm] node[below=2pt] {$2$} (s-4);
			\draw (s-4) edge[-latex'] node[above=2pt] {$\colors \setminus \{2\}$} (bot);

			\draw (top) edge[-latex',out=-30,in=30,distance=0.8cm] node[right=2pt] {$\colors$} (top);
			\draw (bot) edge[-latex',out=180-30,in=180+30,distance=0.8cm] node[left=2pt] {$\colors$} (bot);
		\end{tikzpicture}
		\caption{Minimal-state automaton of $\prefEq_{\DSObj}$ for $\disc = \frac{1}{2}$ and $\colors = \{-2, -1, 0, 1, 2\}$.
		The value in a state is the gap characterizing the equivalence class of $\prefEq$ corresponding to that state.
		Here, $\frac{\sup\colors}{1-\disc} = 4$ and $\frac{\inf\colors}{1-\disc} = -4$.
		The asymmetry around $0$ comes from the $\ge 0$ inequality in the definition of the condition: when state $-4$ is reached, there is exactly one winning continuation ($2^\omega$), but a state with gap value $4$ would only have winning continuations (hence, it is part of state $\top$).
		Notice that we can define a parity condition on top of this automaton that recognizes $\DSObj$: an infinite word is winning as long as it does not reach $\bot$.}
		\label{fig:minStateDS12}
	\end{figure}
	The proof of the other implication is provided in Appendix~\ref{app:DS}, Proposition~\ref{prop:infGaps}.
\end{proof}

\newpage
\begin{corollary}
	Let $\disc\in\intervaloo{0, 1}\cap\IQ$, $k \in \IN$, and $\colors = \intervalcc{-k, k}\cap \IZ$.
	\begin{itemize}
		\item If $k < \frac{1}{\disc} - 1$, then $\DSObj$ is memoryless-determined.
		\item If $k \ge \ceil{\frac{1}{\disc} - 1}$, then $\DSObj$ is chromatic-finite-memory determined if and only if $\disc$ is equal to $\frac{1}{n}$ for some integer $n \ge 1$.
	\end{itemize}
\end{corollary}
\begin{proof}
	This follows from Propositions~\ref{prop:DSCycCons} and~\ref{prop:DSPrefInd}, thanks to Theorem~\ref{thm}.
	The only thing to clarify is that memoryless strategies suffice in case $k < \frac{1}{\disc} - 1$.
	The proof of Proposition~\ref{prop:DSPrefInd} tells us that in this case, $\DSObj$ is $\omega$-regular and can be recognized by a parity automaton that can be defined on top of $\minStateAtmtn \memProduct \memSkelTriv$, which has three states.
	To use this skeleton as a memory, we can notice that the game is already over in states $\eqClass{1}$ and $\eqClass{-1}$ (as every continuation wins or every continuation loses).
	Thus, it is not necessary to use these states to play, and we can consider that we always stay in state $\eqClass{\emptyWord}$.
\end{proof}

\subsection{Other winning conditions} \label{sec:otherObj}
\subparagraph*{Mean payoff.}
Let $\colors \subseteq \IQ$ be non-empty.
We define the \emph{mean-payoff function} $\MPSolo\colon \colors^\omega \to \IR\cup\{-\infty, \infty\}$ such that for $\word = \clr_1\clr_2\ldots\in\colors^\omega$,
\[
\MP{\word} = \limsup_{n\to\infty} \frac{1}{n} \sum_{i=1}^n \clr_i.
\]
We define the winning condition $\MPObj = \{\word\in\colors^\omega \mid \MP{\word} \ge 0\}$ as the set of infinite words whose mean payoff is non-negative.
This condition is $\memSkelTriv$-prefix-independent for any set of colors.
However, it is known that infinite-memory strategies may be required to play optimally in some infinite arenas~\cite[Section~8.10]{Put94}; the example provided uses infinitely many colors.
Here, we show that chromatic-finite-memory strategies do not suffice to play optimally even when $\colors = \{-1, 1\}$.
Let us analyze cycle-consistency of $\MPObj$.
If we consider, for $n \in \IN$,
\[
	\word_n = \underbrace{1\ldots1}_{n\ \text{times}}\underbrace{-1\ldots{-1}}_{n+1\ \text{times}},
\]
we have that $(\word_n)^\omega$ is losing for all $n\in\IN$, but the infinite word $\word_0\word_1\word_2\ldots{}$ has a mean payoff of~$0$ and is thus winning.
This shows directly that $\MPObj$ is not $\memSkelTriv$-cycle-consistent.
The argument can be adapted to show that $\MPObj$ is not $\memSkel$-cycle-consistent for any skeleton $\memSkel$ (see Appendix~\ref{app:DS}, Lemma~\ref{lem:MPCyc}).

\begin{remark}
	As memoryless strategies suffice to play optimally for both players in \emph{finite} arenas for mean-payoff games~\cite{EM79}, winning condition $\MPObj$ distinguishes finite-memory determinacy in finite and in infinite arenas.
	For a skeleton $\memSkel$, $\memSkel$-determinacy in \emph{finite} arenas has also been shown~\cite{GZ05,BLORV20} to be equivalent to the conjunction of a property about prefixes (\emph{$\memSkel$-monotony}) and a property about cycles (\emph{$\memSkel$-selectivity}).
	The concepts of $\memSkel$-selectivity and $\memSkel$-cycle-consistency share the similar idea that combining losing cycles cannot produce a winning word, but they are distinct notions with different quantification on the families of cycles considered.
	Here, $\MPObj$ is $\memSkelTriv$-selective but not $\memSkelTriv$-cycle-consistent.
	% \lipicsEnd
\end{remark}

\subparagraph*{Total payoff.}
Let $\colors \subseteq \IQ$ be non-empty.
We define the \emph{total-payoff function} $\TPSolo\colon \colors^\omega \to \IR\cup\{-\infty, \infty\}$ such that for $\word = \clr_1\clr_2\ldots\in\colors^\omega$,
\[
\TP{\word} = \limsup_{n\to\infty} \sum_{i=1}^n \clr_i.
\]
We define the winning condition $\TPObj = \{\word\in\colors^\omega \mid \TP{\word} \ge 0\}$ as the set of infinite words whose total payoff is non-negative.

The right congruence $\prefEq$ of $\TPObj$ does not have finite index, even for $\colors = \{-1, 1\}$: we indeed have that $(-1) \strictInvPrefOrd (-1)(-1) \strictInvPrefOrd\ldots$.
Condition $\TPObj$ is therefore not $\memSkel$-prefix-independent for any skeleton $\memSkel$.
We can also show that $\TPObj$ is not $\memSkel$-cycle-consistent for any $\memSkel$, using the exact same argument as for $\MPObj$.
Chromatic-finite-memory strategies are therefore insufficient to play optimally for $\TPObj$ in infinite arenas.
Once again, this situation contrasts with the case of finite arenas, in which memoryless strategies suffice to play optimally~\cite{GZ04}.

\section{Conclusion}
We proved an equivalence between chromatic-finite-memory determinacy of a winning condition in games on infinite graphs and $\omega$-regularity of the corresponding language of infinite words, generalizing a result by Colcombet and Niwi\'nski~\cite{CN06}.
A ``strategic'' consequence of our result is that chromatic-finite-memory determinacy in one-player games of both players implies the seemingly stronger chromatic-finite-memory determinacy in zero-sum games.
A ``language-theoretic'' consequence is a relation between the representation of $\omega$-regular languages by parity automata and the memory structures used to play optimally in zero-sum games, using as a tool the minimal-state automata classifying the equivalence classes of the right congruence.

For future work, one possible improvement over our result is to deduce tighter chromatic memory requirements in two-player games compared to one-player games.
Our proof technique gives as an upper bound on the two-player memory requirements a product between the minimal-state automaton and a sufficient skeleton for one-player arenas, but smaller skeletons often suffice.
We do not know whether the product with the minimal-state automaton is necessary in general in order to play optimally in two-player arenas (although it is necessary in Theorem~\ref{thm} to describe $\wc$ using a parity automaton).
This behavior contrasts with the case of finite arenas, in which it is known that a skeleton sufficient for both players in finite one-player arenas also suffices in finite two-player arenas~\cite{BLORV20,BORV21}.
More generally, it would be interesting to characterize precisely the (chromatic) memory requirements of $\omega$-regular winning conditions, extending work on the subclass of Muller conditions~\cite{DJW97,Cas22}.

% \bibliography{articlesFM}
\printbibliography

% \newpage

\appendix
\section{Proof of claim from Section~\ref{sec:intro}} \label{app:CF13}
We argue that the result on the strategy complexity of finitary games in~\cite[Theorem~2]{CF13}, mentioned in Section~\ref{sec:intro} (paragraph \textbf{Related works}), also apply to our setting.
The main difference is that our setting considers arenas with \emph{edges} labeled with colors rather than \emph{states}.
As we have discussed, this difference may in general have an impact on the strategy complexity~\cite{GW06,CN06}.
However, we argue that this difference has no impact for the strategy complexity of finitary games as defined in~\cite{CF13}.
A similar argument was stated in~\cite{CN06} to transfer the memoryless determinacy of parity conditions from state-labeled to edge-labeled arenas.
We refer to~\cite{CF13} for formal definitions of finitary conditions.
Informally, finitary B\"uchi is a winning condition defined over the alphabet of colors $\colors = \{0, 1\}$.
It contains the infinite words for which there exists a uniform bound $N\in\IN$ such that each $1$ is followed by a $0$ after at most $N$ steps.

\begin{lemma}[{\cite[Theorem~2]{CF13}} for edge-labeled arenas]
	For finitary B\"uchi games, $\Pone$ has a memoryless optimal strategy in all infinite (edge-labeled) arenas.
\end{lemma}
\begin{proof}[Proof (sketch)]
	Let $\colors = \{0, 1\}$.
	Let $\arena = \arenaFull$ be an (edge-labeled) arena, as defined in Section~\ref{sec:preliminaries}, and $\wc$ be a finitary B\"uchi condition.
	We transform $\arena$ into a \emph{state-labeled arena} $\arena' = (\states', \states_1', \states_2', \edges', \clr)$ such that $\edges'\subseteq \states\times \states$ and $\clr\colon \states \to \colors$ labels each state with a color.
	For each edge $\edge\in\edges$, we define a new state $\s_\edge$ to which we assign the color $\col(\edge)$, and we insert this state between $\edgeIn(\edge)$ and $\edgeOut(\edge)$.
	To any other state, we assign the color $1$.
	We define
	\begin{align*}
		&\states'_1 = \states_1 \disjUnion \{\s_\edge \mid \edge\in\edges, \edgeIn(\edge)\in\states_1\}, \states'_2 = \states_2 \disjUnion \{\s_\edge \mid \edge\in\edges, \edgeIn(\edge)\in\states_2\}, \states' = \states_1'\disjUnion\states_2', \\
		&\edges' = \{(\edgeIn(\edge), \s_\edge) \mid \edge\in\edges\} \disjUnion \{(\s_\edge, \edgeOut(\edge)) \mid \edge\in\edges\}, \\
		&\text{for $\s\in\states$, $\clr(\s) = 1$, and for $\edge\in\edges$, $\clr(\s_\edge) = \col(\edge)$}.
	\end{align*}
	There is a natural bijection between strategies on $\arena$ and strategies on $\arena'$ which preserves the memoryless feature of strategies.
	Moreover, the winning feature of strategies is also preserved: the infinite words of colors seen in the state-labeled arena will be the same as the ones in the edge-labeled arenas, with extra $1$'s in every other position.
	This alteration does not change the winning status of infinite words for the finitary B\"uchi condition.
	We can therefore reuse~\cite[Theorem~2]{CF13} on $\arena'$, and recover a memoryless optimal strategy of $\Pone$ on~$\arena$.
\end{proof}

An example showing that $\Ptwo$ needs infinite memory to play optimally in some state-labeled arena for finitary B\"uchi conditions~\cite[Example~1]{CF13} can also be easily transformed to work with edge-labeled arenas.
Similar arguments can be used to show that $\Pone$ has optimal strategies based on a skeleton for finitary parity conditions in edge-labeled arenas.

\section{Missing proofs of Section~\ref{sec:application}} \label{app:DS}
We prove statements that were given without proof in Section~\ref{sec:DS} about the discounted-sum winning condition: Proposition~\ref{prop:DSCycCons}, as well as the two properties used in the proof of Proposition~\ref{prop:DSPrefInd} (Propositions~\ref{prop:numbRep} and~\ref{prop:infGaps}).

\DSCycCons*
\begin{proof}
	Let $\word\in\colors^*$.
	We show that $(\cc{\loseCycWordAtmtn{\word}{\memSkelTriv}})^\omega \subseteq \inverse{\word}\comp{\DSObj}$ --- we discuss how to adapt the proof to show that $(\cc{\winCycWordAtmtn{\word}{\memSkelTriv}})^\omega \subseteq \inverse{\word}\DSObj$ at the end.
	Let $\word_1, \word_2, \ldots\in\cc{\loseCycWordAtmtn{\word}{\memSkelTriv}}$.
	We want to show that $\word\word_1\word_2\ldots\in\comp{\DSObj}$, i.e., that $\DS{\disc}{\word\word_1\word_2\ldots} < 0$.

	For $k\ge 1$, as $\word_k\in\loseCycWordAtmtn{\word}{\memSkelTriv}$, we have $\DS{\disc}{\word\word_k^\omega} < 0$.
	Since, moreover,
	\begin{align*}
		\DS{\disc}{\word\word_k^\omega}
		&= \DS{\disc}{\word} + \disc^{\card{\word}}\DS{\disc}{\word_k^\omega} \\
		&= \DS{\disc}{\word} + \disc^{\card{\word}}\sum_{i=0}^\infty \disc^{i\card{\word_k}}\DS{\disc}{\word_k} \\
		&= \DS{\disc}{\word} + \disc^{\card{\word}}\DS{\disc}{\word_k}\frac{1}{1 - \disc^{\card{\word_k}}},
	\end{align*}
	we obtain
	\begin{equation} \label{eq:strict}
		\DS{\disc}{\word_k} < -\DS{\disc}{\word}\frac{1 - \disc^{\card{\word_k}}}{\disc^{\card{\word}}}.
	\end{equation}
	In particular, for $k = 1$, there exists $\varepsilon > 0$ such that
	\begin{equation} \label{eq:eps}
		\DS{\disc}{\word_1} = -\varepsilon - \DS{\disc}{\word}\frac{1 - \disc^{\card{\word_1}}}{\disc^{\card{\word}}}.
	\end{equation}
	We have that
	\begin{align*}
		\DS{\disc}{\word\word_1\word_2\ldots}
		&= \DS{\disc}{\word} + \disc^{\card{\word}}\sum_{k = 1}^{\infty}\disc^{\sum_{i = 1}^{k-1} \card{\word_i}} \DS{\disc}{\word_k} \\
		&\le \DS{\disc}{\word} - \disc^{\card{\word}}\varepsilon -
		\disc^{\card{\word}}\sum_{k = 1}^\infty \disc^{\sum_{i = 1}^{k-1} \card{\word_i}}\DS{\disc}{\word}\frac{1 - \disc^{\card{\word_k}}}{\disc^{\card{\word}}} \\
		&\le \DS{\disc}{\word} - \disc^{\card{\word}}\varepsilon -
			\DS{\disc}{\word}\sum_{k = 1}^{\infty}\disc^{\sum_{i = 1}^{k-1} \card{\word_i}}(1 - \disc^{\card{\word_k}}),
	\end{align*}
	where the second line uses~\eqref{eq:eps} for $k = 1$, and~\eqref{eq:strict} for $k \ge 2$.
	The series
	\[
		\sum_{k = 1}^{\infty}\disc^{\sum_{i = 1}^{k-1} \card{\word_i}}(1 - \disc^{\card{\word_k}}) = \sum_{k = 1}^{\infty}\disc^{\sum_{i = 1}^{k-1} \card{\word_i}} - \disc^{\sum_{i = 1}^{k} \card{\word_i}}
	\]
	is telescoping (we can expand it as $1 - \disc^{\card{\word_1}} + \disc^{\card{\word_1}} - \disc^{\card{\word_1} + \card{\word_2}} + \disc^{\card{\word_1} + \card{\word_2}} - \ldots$).
	As $\lim_{k\to\infty} \disc^{\sum_{i = 1}^{k}\card{\word_i}} = 0$, this series converges to $1$.
	We conclude that
	\[
		\DS{\disc}{\word\word_1\word_2\ldots}
		\le - \disc^{\card{\word}}\varepsilon < 0,
	\]
	as required.
	A proof that $(\cc{\winCycWordAtmtn{\word}{\memSkelTriv}})^\omega \subseteq \inverse{\word}\DSObj$ can be done in a similar way, with no need to extract an $\varepsilon$ as we are then only looking for a non-strict inequality.
\end{proof}

We now prove the two properties used in Proposition~\ref{prop:DSPrefInd} whose proofs were omitted.
We use notations from the proof of Proposition~\ref{prop:DSPrefInd} itself.

\begin{proposition} \label{prop:numbRep}
	Let $\disc\in\intervaloo{0, 1}\cap\IQ$, $k \in \IZ$ such that $k \ge \ceil{\frac{1}{\disc} - 1}$, and $\colors = \intervalcc{-k, k} \cap \IZ$.
	For any real number $x$, $\frac{-k}{1-\disc} \le x \le \frac{k}{1-\disc}$, there exists $\word\in\colors^\omega$ such that $x = \DS{\disc}{\word}$.
\end{proposition}
\begin{proof}
	This problem can be rephrased as a number representation problem: we are looking for a sequence of ``digits'' $(x_i)_{i\ge 0}$ in $\colors$ such that $x = x_0.x_1x_2\ldots{}$ in base $\frac{1}{\disc}$, i.e., such that
	\[
		x = \sum_{i = 0}^\infty x_i\disc^i.
	\]
	Notice that $\sum_{i = 0}^\infty x_i\disc^i = \DS{\disc}{x_0x_1\ldots}$.
	It is known that every number $x\in\intervalco{0, 1}$ has (at least) one representation $0.x_1x_2\ldots{}$ in base $\frac{1}{\disc}$ with digits in $\{0, 1, \ldots, \ceil{\frac{1}{\disc} - 1}\}$, and one such representation can be found using the \emph{greedy expansion}~\cite{Ren57}.

	We adapt this greedy expansion to our setting (for a potentially greater $x$ and larger set $\colors$ of digits).
	Let $x\in\IR$ such that $\frac{-k}{1-\disc} \le x \le \frac{k}{1-\disc}$.
	We deal with the case $x \ge 0$ --- the negative case is symmetric.
	We set $x_0 = \min(k, \floor{x})$; clearly, $x_0 \le x$.
	Then inductively, if $x_0, \ldots, x_{l-1}$ have been defined, we define $x_l$ as the greatest integer in $\{0, \ldots, k\}$ such that
	\[
		\sum_{i = 0}^{l} x_i\disc^i \le x.
	\]
	The series $\sum_{i = 0}^{\infty} x_i\disc^i$ is converging, as every term is non-negative and it is bounded from above by $x$.
	We show that it converges to $x$, which ends the proof.
	Let $\varepsilon \ge 0$.
	Assume by contradiction that $\sum_{i = 0}^{\infty} x_i\disc^i \le x - \varepsilon$.
	Let $j$ be the least index such that $\disc^j \le \varepsilon$.
	Clearly, for any $j' \ge j$, $x_{j'} = k$ --- otherwise, a greater digit could have been picked during the inductive greedy selection.
	Still, not every digit $x_0, x_1, \ldots{}$ can be $k$, as $\sum_{i = 0}^{\infty} k\disc^i = \frac{k}{1-\disc} > x - \varepsilon$.
	Let $l$ be the greatest index such that $x_l \neq k$.
	We show that a digit $\ge x_l + 1$ should have been picked instead of $x_l$ for the digit at index~$l$, leading to a contradiction.
	To do so, it is sufficient to show that
	\[
	(x_l + 1)\disc^l + \sum_{i = 0}^{l-1} x_i\disc^i \le x.
	\]

	We have
	\begin{align*}
	(x_l + 1)\disc^l + \sum_{i = 0}^{l-1} x_i\disc^i
	&= \sum_{i = 0}^{\infty} x_i\disc^i + \disc^l - \sum_{i=l + 1}^{\infty} x_i\disc^i \\
	&= \sum_{i = 0}^{\infty} x_i\disc^i + \disc^l - \sum_{i=l + 1}^{\infty} k\disc^i &&\text{as $x_i = k$ for $i \ge l+1$}\\
	&= \sum_{i = 0}^{\infty} x_i\disc^i + \disc^l - \frac{k\disc^{l+1}}{1 - \disc} \\
	&= \sum_{i = 0}^{\infty} x_i\disc^i + \disc^l(1 - \frac{k\disc}{1 - \disc}).
	\end{align*}

	Since $\frac{k\disc}{1 - \disc} \ge \frac{\ceil{\frac{1}{\disc} - 1}\disc}{1 - \disc} = \ceil{\frac{1-\disc}{\disc}}\frac{\disc}{1 - \disc} \ge 1$, we have that $\disc^l(1 - \frac{k\disc}{1 - \disc}) \le 0$, which implies that
	\[
		(x_l + 1)\disc^l + \sum_{i = 0}^{l-1} x_i\disc^i \le \sum_{i = 0}^{\infty} x_i\disc^i < x,
	\]
	a contradiction.
	We conclude that $x = \DS{\disc}{x_0x_1\ldots}$.
\end{proof}

\begin{proposition} \label{prop:infGaps}
	Let $\disc\in\intervaloo{0, 1}\cap\IQ$, $k \in \IZ$ such that $k \ge \ceil{\frac{1}{\disc} - 1}$, and $\colors = \intervalcc{-k, k} \cap \IZ$.
	If $\disc \neq \frac{1}{n}$ for all integers $n\ge 1$, the $\gapSolo$ function takes infinitely many values.
\end{proposition}
\begin{proof}
	We assume that $\disc = \frac{p}{q}$ with $p, q\in\IN$ co-prime, $p\ge 2$ and $q > p$, and we show that the $\gapSolo$ function takes infinitely many values.
	To do so, we exhibit an infinite word $\word = \clr_1\clr_2\ldots\in\colors^\omega$ such that the sequence of rationals $(\gap{\clr_1\ldots\clr_i})_{i\ge 1}$ never takes the same value twice.

	We will use the following inductive property of gaps: for $\word\in\colors^*$ and $\clr\in\colors$,
	\begin{align} \label{eq:gapInd}
		\gap{\word\clr} = \frac{\gap{\word}}{\disc} + \frac{\disc^{\card{\word}-1}\clr}{\disc^{\card{\word}}}
						= \frac{1}{\disc}\left(\gap{\word} + \clr\right),
	\end{align}
	unless some gap in this equation equals $\top$ or $\bot$.
	Notice that under our hypotheses, $\ceil{\frac{1}{\disc} - 1} = \floor{\frac{1}{\disc}}$ (this equality is not verified when $\disc = \frac{1}{n}$ for some integer $n \ge 1$).

	We set $\clr_1 = 1$.
	Then, $\gap{\clr_1} = \frac{1}{\disc} = \frac{q}{p}$.
	Inductively, if $\clr_1, \ldots, \clr_{i-1}$ are defined, we set $\clr_i = - \floor{\gap{\clr_1\ldots\clr_{i-1}}}$ (we remove the largest possible integer from the current gap, while keeping a positive gap value).
	We set $g_i = \gap{\clr_1\ldots\clr_i}$ for conciseness.

	We first show that if all $g_i$'s are rational, then no two $g_i$'s can be equal.
	To do so, we show inductively that the reduced denominator of fraction $g_i$ is $p^i$ for all $i \ge 1$.
	This is true for $i = 1$.
	For $i > 1$, assume it is true for $i - 1$.
	Then, $g_{i-1} = \frac{m}{p^{i-1}}$ for some $m$ co-prime with $p$.
	Using~\eqref{eq:gapInd},
	\[
	g_i
	= \frac{1}{\disc}\cdot\left(g_{i-1} + \clr_i\right)
	= \frac{1}{\disc}\cdot\left(\frac{m}{p^{i-1}} + c_i\right)
	= \frac{q(m + c_ip^{i-1})}{p^i}.
	\]
	This last fraction is irreducible: $q$ and $p$ are co-prime, and the fact that $m$ and $p$ are co-prime implies that $m + c_ip^{i-1}$ and $p$ are co-prime.

	We now prove by induction that our scheme is well-defined by showing that for all $i\ge 1$, $\clr_i \in \colors$ and $0 < g_i \le \frac{1}{\disc}$.
	This is true for $i = 1$ (as $k \ge 1$ for any possible value of $\disc$).
	For $i > 1$, if this is true for $i - 1$, then $-\clr_i = \floor{g_{i-1}} \le \floor{\frac{1}{\disc}}$, so $\clr_i\in\colors$.
	Moreover, $g_i = \frac{1}{\disc}\left(g_{i-1} + \clr_i\right)$.
	Since $g_i$'s cannot be integers (as their reduced denominator is not $1$ by an earlier property), we have that $g_{i-1} + \clr_i$ is not an integer either.
	Therefore, $0 < g_{i-1} + \clr_i < 1$, so $0 < g_i < \frac{1}{\disc}$.
	As $\frac{1}{\disc} < \frac{k}{1-\disc} = \maxDS$, the values of the considered gaps are never $\top$ or $\bot$.
\end{proof}

\begin{remark}
	For $\disc\in\intervaloo{0, 1}\cap\IQ$ with $\disc\neq\frac{1}{n}$ for all integers $n\ge 1$, $k \ge \ceil{\frac{1}{\disc} - 1}$, and $\colors = \intervalcc{-k, k} \cap \IZ$, Proposition~\ref{prop:DSPrefInd} along with Theorem~\ref{thm} implies that any chromatic skeleton is insufficient to play optimally (for at least one player).
	However, this does not directly give an explicit arena in which some player requires infinite memory to play optimally.
	Here, we show how to construct such an arena given the extra results from Appendix~\ref{app:DS}.

	The proof of Proposition~\ref{prop:infGaps} gives us $\clr_1\clr_2\ldots\in\colors^\omega$ such that $(\gap{\clr_1\ldots\clr_i})_{i\ge 1}$ is a sequence of distinct values in $\intervalcc{0, \maxDS}$.
	Hence, by compacity of $\intervalcc{0, \maxDS}$, there exists a subsequence $(i_j)_{j\ge 1}$ and $x\in\intervalcc{0, \maxDS}$ such that $\lim_{j\to\infty}\gap{\clr_1\ldots\clr_{i_j}} = x$.
	We can further extract a subsequence such that either all elements are greater than $x$, or all elements are less than $x$.
	We assume w.l.o.g.\ that for all $j \ge 1$, $\gap{\clr_1\ldots\clr_{i_j}} < x$ (this implies that $x\neq 0$).
	The proof is symmetric if all the gaps are greater than $x$ (which would imply that $x\neq \maxDS$).

	By Proposition~\ref{prop:numbRep}, for all $\varepsilon > 0$ sufficiently small, there exists $\word_\varepsilon\in\colors^\omega$ such that $\DS{\disc}{\word_\varepsilon} = -x + \varepsilon$.
	We can define an infinite arena in which $\Ptwo$ needs infinite memory to win, depicted in Figure~\ref{fig:DSInf}.
	In this arena, $\Pone$ may choose to reach a gap arbitrarily close (but not equal) to $x$ in $\s_2$, and then $\Ptwo$ is always able to bring the discounted sum below $0$ by choosing a word reaching a discounted sum sufficiently close to $-x$.
	\begin{figure}[t]
		\centering
		\begin{tikzpicture}[every node/.style={font=\small,inner sep=1pt}]
			\draw (0,0) node[rond] (s1) {$\s_1$};
			\draw ($(s1)+(2.5,0)$) node[carre] (s2) {$\s_2$};
			\draw ($(s2)+(2.5,0)$) node[inner sep=3pt] (dots) {$\dots$};
			\draw ($(s2)+(2,1.5)$) node[inner sep=3pt] (dots2) {\reflectbox{$\ddots$}};

			\draw (s1) edge[-latex',out=60,in=180-60,decorate,distance=1.1cm] node[above=3pt] {$\clr_1\ldots\clr_{i_1}$} (s2);
			\draw (s1) edge[-latex',out=15,in=180-15,decorate] node[above=2pt] {$\clr_1\ldots\clr_{i_2}$} (s2);
			\draw (s1) edge[draw=none,out=-15,in=180+15,decorate] node[] {$\vdots$} (s2);

			\draw (s2) edge[-latex',decorate] node[above=5pt,xshift=-2pt] {$\word_{1}$} (dots2);
			\draw (s2) edge[-latex',decorate] node[above=3pt] {$\word_{\frac{1}{2}}$} (dots);
			\draw (s2) edge[draw=none,out=-15,decorate] node[] {$\vdots$} ($(s2)+(2.5,-1.5)$);
		\end{tikzpicture}
		\caption{Infinite arena in which $\Ptwo$ needs infinite memory to win from $\s_1$ for condition $\DSObj$ for $\disc\in\intervaloo{0, 1}\cap\IQ$ with $\disc\neq\frac{1}{n}$ for all integers $n\ge 1$, $k = \ceil{\frac{1}{\disc} - 1}$, and $\colors = \intervalcc{-k, k} \cap \IZ$.}
		\label{fig:DSInf}
	\end{figure}%
	% \lipicsEnd
\end{remark}

We now prove the claim about the mean-payoff condition (Section~\ref{sec:otherObj}) stating that it is not $\memSkel$-cycle-consistent for any $\memSkel$.
\begin{lemma} \label{lem:MPCyc}
	For all skeletons $\memSkel$, $\MPObj$ is not $\memSkel$-cycle-consistent.
\end{lemma}
\begin{proof}
	Let $\memSkel = \memSkelFull$ be a skeleton.
	For $n \in \IN$, let
	\[
		\word_n = \underbrace{1\ldots1}_{n\ \text{times}}\underbrace{-1\ldots{-1}}_{n+1\ \text{times}}.
	\]
	Let $\memStates_n$ be the set of states $\memState\in\memStates$ such that there exists $k \ge 1$ with $\memState = \memUpdHat(\memState, \word_n^k)$.
	Each~$\memStates_n$ is non-empty as, $\memStates$ being finite, iterating the function $\memState \mapsto \memUpdHat(\memState, \word_n)$ necessarily goes multiple times through at least one state.
	Let $\memState\in\memStates$ be a state in set $\memStates_n$ for infinitely many $n$'s, and~$\word$ be any finite word in $\cc{\memPathsOn{\memInit}{\memState}}$.
	Let $n_1, n_2, \ldots{}$ be the indices such that $\memState\in\memStates_{n_i}$, and let $k_1, k_2, \ldots{}$ be such that $\memState = \memUpdHat(\memState, \word_{n_i}^{k_i})$.
	Every word $\word_{n_i}^{k_i}$ is a losing cycle after any finite word (in particular after $\word$).
	However, it is possible to find a subsequence of $(\word_{n_i}^{k_i})_{i\ge 1}$ with a non-negative mean payoff by always taking a word $\word_{n_i}$ that bring the sum of the colors above $0$ during the first $n_i$ $1$'s.
\end{proof}
\end{document}